%% file: main.tex
\documentclass[a4paper,12pt,dvipsnames]{scrartcl}

\usepackage[utf8]{inputenc}	
\usepackage[T1]{fontenc} 				
\usepackage[english]{babel} 

\usepackage{pdfcomment} 
\usepackage{todonotes} 
\usepackage{mdwlist} 
\usepackage{abstract} 

\usepackage[normalem]{ulem} 
\usepackage[autostyle=true,english=american]{csquotes}

\usepackage{ragged2e} 
\usepackage{enumerate}

\usepackage{setspace} 
\expandafter\def\expandafter\quotation\expandafter{\quote\singlespacing}
\usepackage{amsmath,amsfonts,amssymb,mathtools,bm,xcolor,amsthm}
\usepackage{booktabs, caption, longtable}
\usepackage{array} 
\usepackage{textcomp} 

\DeclareMathOperator*{\argmax}{arg\,max}
\DeclareMathOperator*{\plim}{plim}
\newcommand{\dd}{\,\mathrm{d}}
\newtheorem{assumption}{Assumption}
\newtheorem{theorem}{Theorem}

\newtheorem{lemma}{Lemma}
\theoremstyle{remark}

\usepackage{xcolor}

\usepackage{graphicx}
\usepackage{svg} 
\usepackage{float} 
\usepackage[clockwise]{rotating}

\usepackage{tikz}

\usepackage{tabularx} 
\usepackage{booktabs} 
\usepackage{multirow}
\usepackage{tabu}
\usepackage{array}
\newcolumntype{L}[1]{>{\RaggedRight\arraybackslash}p{#1}}
\newcolumntype{C}[1]{>{\Centering\arraybackslash}p{#1}}
\newcolumntype{R}[1]{>{\RaggedLeft\arraybackslash}p{#1}}

\usepackage{caption}
\usepackage[font={footnotesize,it}]{caption}
\usepackage{capt-of}
\usepackage[labelfont=bf]{caption}			
\usepackage{subcaption}
\usepackage{makecell}

\usepackage[left= 3cm, right= 2cm, bottom= 2.5cm, top= 2.5cm]{geometry} 
\usepackage{pdflscape}
\usepackage{float}

\addtokomafont{disposition}{\rmfamily}
\setkomafont{section}{\rmfamily\Large\bfseries}
\setkomafont{subsection}{\rmfamily\large\bfseries}

\usepackage{datetime}
\newdateformat{myformat}{\THEDAY{. }\monthname[\THEMONTH] \THEYEAR}

\usepackage{tocloft}

\usepackage{csquotes}
\usepackage[style=apa,sortcites=true]{biblatex}
\usepackage{hyperref}
\hypersetup{
	pdfencoding=auto,
	pdftitle={},
	pdfsubject={},
	pdfauthor={},
	pdfkeywords={},
	colorlinks=true,
	breaklinks=true,
	citecolor=blue,
	linkcolor=black,
	menucolor=black,
	urlcolor=blue
}
\begin{document}

\input{Titlepage.tex}

\newpage
\pagenumbering{Roman} 
\selectlanguage{english}

\newpage 
\pagenumbering{arabic}
\setcounter{page}{1}
\input{Outline.tex} 

\newpage 
\addcontentsline{toc}{section}{References}
\nocite{*}
\printbibliography[heading=bibliography]

\newpage
\addcontentsline{toc}{section}{Appendix}
\input{Appendix.tex}

\end{document}

%% file: Titlepage.tex
\begin{titlepage}
	\begin{center}
		{\Large \textbf{Functional Factor Regression with an Application to Electricity Price Curve Modeling}\par}
		
		\vspace{1cm}
		
		\begin{tabular}{cc}
			{\large Sven Otto\textsuperscript{†}} & {\large Luis Winter\textsuperscript{‡}} \\
			\textit{University of Cologne} & \textit{University of Cologne}
		\end{tabular}
		
		\vspace{1cm}
		
		\begin{minipage}{\textwidth}
			\input{Abstract.tex}
		\end{minipage}
		
		\vspace{\fill}
	
		\begin{minipage}{\textwidth}
			\rule{0.5\textwidth}{0.5pt}\\
			\scriptsize
			\textsuperscript{†}Institute of Econometrics and Statistics—University of Cologne, \href{mailto:sven.otto@uni-koeln.de}{sven.otto@uni-koeln.de}\\
			\textsuperscript{‡}Corresponding author: Institute of Econometrics and Statistics—University of Cologne, \href{mailto:l.winter@uni-koeln.de}{l.winter@uni-koeln.de}\\
			We thank Jörg Breitung, Frédéric Ferraty, Simon Hirsch, Siegfried Hörmann, Dominik Liebl, Alexander Mayer, Oliver Ruhnau, Nazarii Salish, and Fabian Telschow for their valuable comments and suggestions. This work was supported by the
			Deutsche Forschungsgemeinschaft (DFG) under project number 511905296.
			
		\end{minipage}
	
	\end{center}
	
	\setcounter{page}{0}
	\thispagestyle{empty}
\end{titlepage}

%% file: Abstract.tex
\begin{abstract}
	\indent We propose a function-on-function linear regression model for time-dependent curve data that is consistently estimated by imposing factor structures on the regressors. An integral operator based on cross-covariances identifies two components for each functional regressor: a predictive low-dimensional component, along with associated factors that are guaranteed to be correlated with the dependent variable, and an infinite-dimensional component that has no predictive power. In order to consistently estimate the correct number of factors for each regressor, we introduce a functional eigenvalue difference test. While conventional estimators for functional linear models fail to converge in distribution, we establish asymptotic normality, making it possible to construct confidence bands and conduct statistical inference. The model is applied to forecast electricity price curves in three different energy markets. Its prediction accuracy is found to be comparable to popular machine learning approaches, while providing statistically valid inference and interpretable insights into the conditional correlation structures of electricity prices.\\
	\vspace{0.5cm}
	
	\bigskip
\end{abstract}

%% file: Outline.tex
\input{Introduction.tex}

\input{Function_on_function_regression_problem.tex}
\input{Functional_factor_regression.tex}
\input{Estimation.tex}
\input{Simulation.tex}
\input{Application_to_electricity_price_curve_forecasting.tex}

\input{Conclusion.tex}

%% file: Introduction.tex
\section{Introduction}

Empirical economic research increasingly relies on curve-valued data to capture rich, time-varying heterogeneity in economic variables. Recent macroeconomic studies represent heterogeneous household earnings or consumption as entire functions to reflect distributional dynamics (\cites{Chang.2024}{Bayer.2025}), and model term structures of interest rates and inflation expectations as curves (\cites{Aruoba.2019}{Inoue.2021}). In energy economics, daily electricity price curves are similarly treated as continuous functions (\cites{Liebl.2013}{Chen.2017}{Gonzalez.2018}). Observing entire curves repeatedly over time requires econometric methods flexible enough to handle infinite-dimensional objects and parsimonious enough to support statistical inference.\\
\indent When both regressors and dependent variables are curves, linear function-on-function regression presents a natural modeling framework. Despite its intuitive appeal, a key limitation has hindered its widespread use in empirical economics: classical estimators achieve consistency but do not yield asymptotic normality for the functional slope operator. The most common estimation approaches are based on principal components spectral truncation (\cite{Imaizumi.2018}) or Tikhonov/ridge regularization (\cite{Benatia.2017}). As demonstrated by \textcite{Mas.2007}, \textcite{Crambes.2013}, and \textcite{Babii.2020}, when regressors truly reside in an infinite-dimensional space, the finite-dimensional truncated part may converge at a parametric rate while the infinite-dimensional remaining truncation tail converges only at a slower, nonparametric rate. This discrepancy prevents weak convergence in the underlying operator norm, making the development of inferential methods based on standard errors and confidence bands impossible through conventional central limit theorems. Asymptotic normality has only been developed for the biased ridge estimator with fixed shrinkage parameter (\cite{Benatia.2017}) and for the predicted variable from principal components estimation (\cite{Crambes.2013}), but these limited results do not enable the construction of valid confidence bands for the regression slope coefficient functions. Some workarounds for constructing confidence sets involve relaxing coverage requirements to ensure bands contain the true function at most points in the domain rather than everywhere (\cite{Imaizumi.2019}) or are conservative relying on concentration inequality bounds (\cite{Babii.2020}). Hypothesis tests have only been developed for the lack of dependence in the i.i.d.\ case (\cites{Cardot.2003}{Kokoszka.2008}{Benatia.2017}) or for the operator norm itself (\cite{Kutta.2022}). However, functional extensions of standard central limit theorem-based regression diagnostic tools, such as standard errors and confidence bands, are ruled out by this fundamental impossibility result.\\
\indent This paper overcomes these limitations by developing a novel functional factor regression framework that preserves the flexibility of unrestricted function‐on‐function regression while allowing for statistical inference known from multiple linear regression. We argue that the predictive information contained in typical macroeconomic functional regressors (e.g., yield curves, inflation expectation curves, energy spot price curves) is confined to a finite‐dimensional subspace, while the remaining infinite‐dimensional component does not contribute to the regression relationship. For instance, \textcite{Otto.2025} provide empirical evidence that the functional autoregression operator of yield and mortality curves is indeed of finite rank.\\
\indent This premise leads to an approximate factor structure in which a finite-dimensional factor component of $K<\infty$ factors drives the regression relationship, while the remainder component can be disregarded as it is uncorrelated with all variables in the system. Consequently, the functional slope operator has finite rank and can be identified using the cross-covariance operator between regressor and regressand. This setup naturally extends to multiple functional regressors and accommodates additional scalar covariates and lagged dependent functional variables, making it particularly suitable for functional time series applications.\\
\indent The core of our estimation approach is the product of the cross‐covariance operator with its adjoint, denoted as $D$, that identifies precisely which directions in the regressor have predictive power for the response variable. Unlike functional principal component analysis, which captures directions of maximum variance in the regressor, our method isolates factors guaranteed to be correlated with the dependent variable. The number of positive eigenvalues of $D$ equals the number of factors $K$ and its eigenfunctions reveal the directions in which the regressor exhibits predictive power. This structure allows us to estimate the corresponding factor loadings from the sample counterpart of $D$. A similar operator for the purely autoregressive case has been considered by \textcite{Lam.2012} and \textcite{Zhang.2019} in the multivariate literature, and by \textcite{Bathia.2010} and \textcite{Otto.2025} in the context of functional data. \textcite{Cho.2013} identify the predictive components in a univariate function-on-function regression by singular value decomposition of the cross-covariance operator. However, their theoretical framework does not address the issue of statistical inference. Moreover, to consistently estimate the number of factors, we propose a functional equivalent of the eigenvalue difference test proposed by \textcite{Wu.2018}, which exploits the property that the population eigenvalues of $D$ are exactly zero beyond the $K$-th eigenvalue.\\
\indent Building on this identification strategy, we develop a consistent estimation procedure that enables valid inference. To estimate the functional slope parameter, we perform an auxiliary regression of the functional dependent variable on the estimated factors using ordinary least squares. Our theoretical results establish that the functional factor regression estimator is asymptotically normal, allowing for the construction of valid confidence bands and hypothesis tests. The asymptotic variance accounts for the uncertainty introduced by the generated regressors problem and allows for heteroskedasticity in the functional error term.\\
\indent We demonstrate the practical value of our method through an application to electricity market data to model and forecast daily price curves using demand forecasts and renewable energy generation as functional regressors. Our empirical results reveal economically interpretable patterns in how these variables influence electricity prices throughout the day, visualized through three-dimensional coefficient plots with corresponding p-value heatmaps. The interpretation of the estimated bivariate slope coefficient functions is in line with economic theory and extends our knowledge of prevalent correlation structures in power markets. In out-of-sample forecasting performance, our functional factor regression outperforms both expert and naive models while achieving accuracy comparable to machine learning approaches like LASSO regression, but with the significant advantage of providing interpretable coefficients and valid statistical inference. Our methodology can be applied using our accompanying \textsf{R} package "ffr".\\
\indent Function-on-function linear regression models have a long tradition in the mathematical statistics literature (see \cites{Ramsay.1991} {Bosq.2000}{Yao.2005}{Crambes.2013}{Hormann.2014}{Happ.2018}). While these general frameworks impose fewer restrictions than our approach, they fundamentally cannot permit the construction of valid confidence bands through functional standard errors. Our key insight is that accepting a finite-dimensional predictive component enables asymptotic normality. In contrast, if the predictive component were infinite-dimensional, the impossibility result discussed above would apply. The eigenvalue difference test we propose not only consistently estimates $K$ when finite, but also serves to validate this finite-dimensionality assumption empirically, as the criterion would diverge if $K$ were infinite. Our application demonstrates that just 3 to 5 factors suffice to capture the predictive information of functional energy market data.\\
\indent In contrast to the general infinite-dimensional approach, the simpler case assumes finite-dimensional functional regressors (see, e.g., \cite{Hormann.2023}). Under this stronger restriction, the functional model can be expressed as a well-posed regression with estimators that achieve parametric rates of convergence. While our framework encompasses this as a special case, such models require the covariance operator of the regressors to have finite rank—an overly restrictive assumption that constrains regressors to a finite-dimensional subspace of the function space. Instead, our model requires only the pairwise cross-covariance operators between the functional regressor and response to be of finite rank, while the covariance operator of the regressor itself may have infinite rank.\\[7pt]
\textbf{Notations}\\
Let $H=L^2([a,b])$ be the space of functions $y: [a,b] \rightarrow \mathbb{R}$ with $\int_{a}^{b}(y(r))^2 \,dr < \infty$. Together with the inner product $\langle x,y \rangle = \int_{a}^{b} x(r)y(r) \,dr$, $x,y \in H$, and the norm $\|y\| = \langle y,y \rangle^{1/2}$, the space $H$ is a Hilbert space. Every square-integrable bivariate function $\tau: [a,b] \times [a,b] \rightarrow \mathbb{R}$ defines an integral operator $\mathcal T: H \rightarrow H$ such that $\mathcal T(y)(r) = \int_{a}^{b}\tau(r,s)y(s)\,ds$ for all $y \in H$.
The Hilbert-Schmidt operator norm of $\mathcal T$ is $\|\mathcal T\|_{\mathcal S} = (\int_a^b \int_a^b \tau(r,s)^2 \,ds \,dr)^{1/2}$.
If $\kappa$ and $v$ satisfy the eigenequation $\int_{a}^{b}\tau(r,s)v(s) \,ds = \kappa v(r)$, then $\kappa$ is called eigenvalue and $v$ a corresponding eigenfunction of $\mathcal T$.

%% file: Function_on_function_regression_problem.tex
\section{Function-on-function regression problem}\label{sec:regression}

Consider the covariance stationary time series of curves $Y_1(r), \ldots, Y_T(r)$, defined on the closed domain $r \in [a,b ]$, to be our dependent random variable of interest. Moreover, there is a set of $j=1, \ldots, J$ time-dependent covariance stationary functional regressors $X_{jt}(s)$, $s \in [a,b]$, available. Besides that, consider an $\mathbb{R}^N$ valued vector $w_t = (1, w_{2t}, \ldots, w_{Nt})'$, including an intercept and covariance stationary covariates. A fully functional linear regression model thus takes the form
\begin{equation}\label{eq:basemodel}
	Y_t(r) = w_t'\alpha(r) + \sum_{j=1}^J \int_a^b \beta_j(r,s) X_{jt}(s) \,ds + u_t(r),
\end{equation}
where $r,s\in [a,b]$ and $t=1, \ldots, T$. As standard in the regression framework, $\alpha(r) = (\alpha_1(r), \ldots, \alpha_N(r))'$ and $\beta_1(r,s), \ldots, \beta_J(r,s)$ are deterministic coefficient functions while $u_t(r)$ is an error term. We assume exogenous regressors with a martingale difference sequence error term:
\begin{assumption}\label{as:exo}
	\textbf{\textup{(Exogeneity).}}\\
	Define the information set $\mathcal{F}_{t-i}=\sigma\left(u_{t-1-i}, w_{t-i}, X_{1(t-i)}, \ldots, X_{J(t-i)}\right).$ Then
	\begin{equation*}
		E[u_t(r) \mid \{\mathcal{F}_{t-i}: i \geq 0 \}] = 0,\quad \text{for all} \ r \in [a,b].
	\end{equation*}
\end{assumption}

Our setup includes the functional autoregressive model (\cite{Bosq.2000}), where the regressor functions $X_{jt}$ are lags of $Y_t$.
A common identifiability assumption is that the covariance function $Cov(X_{jt}(r), X_{jt}(s))$ is strictly positive definite (\cites{Mas.2007}{Imaizumi.2018}), which implies that $Var(\langle X_{jt}, h \rangle) > 0$ for all $\|h\| > 0$. This condition defines the functional equivalent of the positive definite design matrix condition imposed in classical regression models.
However, if the regressor function lives in some subspace of $H$, then $\beta_j(r,s)$ remains unidentified in the directions where the regressor has zero variance. 
Therefore, we adopt a weaker condition which restricts the domain of the regression operator to the subspace where the regressor has positive variance (\cite{Caponera.2022}). This condition allows for both infinite- and finite-rank covariance operators.
\begin{assumption}\label{as:identifiability}
	\textbf{\textup{(Identification I).}}\\
	Let $\Gamma_j$ denote the covariance operator of $X_{jt}$ mapping any function $h \in H$ to the function $\int_a^b Cov(X_{jt}(\cdot), X_{jt}(s)) h(s) \,ds \in H$. Then,
	$$\beta_j(r,\cdot) \in \overline{Im(\Gamma_j)} \subseteq H,\quad \text{for all} \ r \in [a,b].\footnote{The bar indicates the closure of the set, which is required because $Im(\Gamma_j)$ is not necessarily closed if $\Gamma_j$ has infinite rank. An equivalent condition, not requiring operator theory, is: $\int_a^b \beta_j(r,s) h(s) \,ds = 0$ for all $h \in H$ with $Var(\langle X_{jt}, h \rangle) = 0$, and all $r \in [a,b]$.}$$
\end{assumption}

To illustrate the challenges of estimating model \eqref{eq:basemodel}, consider, for simplicity, the case $N=0$ and $J=1$ with one functional regressor $X_t$ and centered variables, i.e., $E[Y_t(r)] = E[X_{t}(s)] = 0$ for all $r,s \in [a,b]$. 
The model equation becomes
\begin{equation} \label{eq:modelsimple}
	Y_t(r) = \int_a^b \beta(r,s) X_{t}(s) \,ds + u_t(r).
\end{equation}
It is useful to represent the model in terms of the functional principal components (FPCs) $\phi_l$ of $X_t$, which are a sequence of orthonormal eigenfunctions of the covariance operator $\Gamma$ of $X_t$, ordered descendingly by their corresponding positive eigenvalues. The FPCs form an orthonormal basis of $\overline{Im(\Gamma)}$, the subspace where the regression operator is defined according to Assumption \ref{as:identifiability}.\\
\indent In the following subsections, we discuss the two conventional approaches that are usually taken to estimate the simplified model (2). The general setup, discussed in Section \ref{sec:ill-posed}, allows for infinite rank $\Gamma$ but suffers from an ill-posed inverse problem whereas the restricted case, discussed in Section \ref{sec:well-posed}, requires a finite rank $\Gamma$ and implies a well-posed regression framework.
In Section \ref{sec:simpleFFR}, we introduce our novel method which serves as a middle ground between the alternatives and allows for an infinite rank $\Gamma$ but requires a finite-rank cross-covariance operator.

\subsection{Infinite-dimensional regressor with infinite factor structure} \label{sec:ill-posed}

Firstly, consider the case where the covariance operator $\Gamma$ has infinite rank, so that $\phi_l$ forms an infinite sequence of orthonormal eigenfunctions. The regression coefficient can be represented as
\begin{equation} \label{eq:betaKLbasis}
	\beta(r,s) = \sum_{l=1}^\infty b_{l}(r) \phi_l(s), \quad b_{l}(r) = \int_a^b \beta(r,s) \phi_l(s) \,ds,
\end{equation}
and the basis representation of the regressor, known as the Karhunen-Loève decomposition, is given by
\begin{equation} \label{eq:KL}
	X_{t}(s) = \sum_{l=1}^\infty x_{lt} \phi_l(s), \quad x_{lt} = \langle X_{t}, \phi_l \rangle.
\end{equation}
Equation \eqref{eq:KL} indicates the infinite factor structure of the regressor $X_t$, where the scores $x_{lt}$ can be interpreted as FPC factors, and the FPCs $\phi_l$ as the corresponding orthonormal loading functions. By combining \eqref{eq:KL} and \eqref{eq:betaKLbasis}, model \eqref{eq:modelsimple} becomes
\begin{equation} \label{eq:modelsimpleinfinite}
	Y_t(r) = \sum_{l=1}^\infty b_{l}(r) x_{lt} + u_t(r),
\end{equation}
which is a regression on infinitely many regressors $x_{lt}$.
Each coefficient $\beta_{l}(r)$ is identified by the moment condition
$$
E[x_{lt} Y_t(r)] = b_l(r) E[x_{lt}^2],
$$
and, by equation \eqref{eq:betaKLbasis}, the coefficient function has the solution 
$$\beta(r,s) = \sum_{l=1}^\infty \frac{E[x_{lt} Y_t(r)]}{E[x_{lt}^2]} \phi_l(s).$$
Estimating $\beta(r,s)$ is challenging because it relies on infinitely many moment conditions, which results in an ill-posed inverse problem that requires regularization (see \cite{Carrasco.2007}).\\
\indent The most common solution is spectral cut-off, where the infinite sum over the FPCs is truncated after $\tilde L$ components.
The resulting FPC estimator is given by 
\begin{align}
	\widehat \beta_{\text{FPC}}(r,s) = \sum_{l=1}^{\tilde L} \widehat b_l(r) \widehat \phi_l(s), \quad \widehat b_l(r) = \frac{\sum_{t=1}^T \widehat x_{lt}  Y_t(r) }{\sum_{t=1}^T \widehat x_{lt}^2}, \label{eq:FPC}
\end{align}
where $\widehat \phi_l$ are the sample FPCs and $\widehat x_{lt} = \langle X_t, \widehat \phi_l \rangle$ the sample FPC scores. Under an appropriate choice of the truncation parameter $\tilde L$ as an increasing function of the sample size, consistency results for $\widehat \beta_{\text{FPC}}(r,s)$ are established by \textcite{Crambes.2013}, \textcite{Hormann.2014}, and \textcite{Imaizumi.2018}, with non-parametric rates of convergence slower than $\sqrt T$. Other regularization approaches include double FPC truncation (see \cite{Yao.2005}) and Tikhonov regularization (see \cites{Benatia.2017}{Caponera.2022}). However, it is impossible to prove a central limit theorem for spectral cut-off and Tikhonov regularized estimators for $\beta(r,s)$ in the Hilbert-Schmidt topology (\cites{Mas.2007}{Babii.2020}).

\subsection{Finite-dimensional regressor with exact factor structure} \label{sec:well-posed}

The regression problem \eqref{eq:modelsimple} becomes well-posed if the covariance operator $\Gamma$ has finite rank.
This condition has been considered, for instance, in \textcite{Li.2013}, \textcite{Li.2022}, and \textcite{Hormann.2023}.
If $rank(\Gamma) = L < \infty$, then the Karhunen-Loève expansion \eqref{eq:KL} becomes
\begin{equation} \label{eq:Lfactormodel}
	X_{t}(s) = \sum_{l=1}^L x_{lt} \phi_l(s),
\end{equation}
which describes an exact factor model, where the loadings $\phi_l$ span a finite-dimensional factor space.
Equation \eqref{eq:modelsimpleinfinite} becomes
\begin{equation} \label{eq:simplefunctionalfactorregression}
	Y_t(r) = \sum_{l=1}^L b_{l}(r) x_{lt} + u_t(r),
\end{equation}
which is a well-posed regression of $Y_t(r)$ on finitely many FPC factors $x_{1t}, \ldots, x_{Lt}$. \textcite{Hormann.2023} show that the spectral cut-off estimator has parametric rates of convergence in this case, and methods from multivariate statistics can be used to conduct statistical inference.\\
\indent In the functional models commonly used in applied econometrics such as the Nelson-Siegel model for yield curves or inflation expectations (\cites{Diebold.2006}{Aruoba.2019}), it is standard to assume that the curves have a finite-dimensional underlying factor structure, especially when the curves are smooth. However, not all FPC scores may be relevant for explaining the linear regression relationship with $Y_t$, and the number of components $L$ could still be large relative to the sample size $T$, which still requires regularization to avoid overfitting. Additionally, the FPC basis may not be optimal for functional linear regression, as it indicates directions with the largest variance in the regressor, rather than those with the largest covariance with the dependent variable. As discussed in \textcite{Otto.2025}, it is more appropriate to identify factors that are directly correlated with the response variable $Y_t$.

\subsection{Infinite-dimensional regressor with approximate factor structure}\label{sec:simpleFFR}

We propose an approximate factor model that decomposes $X_t$ into a finite number of predictive factors plus a non-predictive, possibly infinite-dimensional error component.
Concretely, we extend equation \eqref{eq:Lfactormodel} by writing
\begin{align} \label{eq:approxfactormodel}
	X_{t}(s) = X_t^*(s) + \varepsilon_{t}(s),
\end{align}
where
$$
X_t^*(s) := \sum_{l=1}^K f_{lt} \psi_{l}(s) = F_{t}' \Psi(s).
$$
Here, $F_t = (f_{1t}, \ldots, f_{Kt})'$ is a vector of zero-mean latent factors, $\Psi(s) = (\psi_1(s), \ldots, \psi_K(s))'$ is a vector of orthonormal $H$-valued loading functions, and $\varepsilon_t(s)$ is an additional $H$-valued error term. 
Because the covariance operator of the error $\varepsilon_t(s)$ is unrestricted, \eqref{eq:approxfactormodel} defines an approximate factor model in the sense of 
\textcite{Chamberlain.1983}, which allows the covariance operator of $X_t$ to have infinite rank.

To identify the factor component $X_t^*(s)$ separately from $\varepsilon_{t}(s)$, we require that $\varepsilon_{t}(s)$ is non-predictive for $Y_t(r)$, i.e.,
\begin{align}
	Cov(Y_t(r), \varepsilon_{t}(s)) = 0, \quad Cov(f_{lt}, \varepsilon_{t}(s)) = 0, \quad Cov(u_t(r), \varepsilon_{t}(s))  = 0, \label{eq:nonpredictiveness}
\end{align}
for all $r,s \in [a,b]$,
while each factor $f_{lt}$ must be predictive in the sense that
\begin{align}
	\int_a^b Cov(Y_t(r), f_{1t})^2 \,dr \geq \ldots \geq \int_a^b Cov(Y_t(r), f_{Kt})^2 \,dr > 0. \label{eq:predictiveness}
\end{align}
Condition \eqref{eq:predictiveness} indicates that our factors $f_{lt}$ can be interpreted as predictive factors and therefore differ from the FPC factors $x_{lt}$, which have the property that
$$
	\int_a^b Cov(X_t(r), x_{1t})^2 \,dr \geq \int_a^b Cov(X_t(r), x_{2t})^2 \,dr \geq \ldots .
$$
By construction, the factor loadings $\psi_1, \ldots, \psi_K$ span the finite-dimensional subspace of $X_t$ that is actually relevant for predicting $Y_t$.
Using \eqref{eq:approxfactormodel}, we can rewrite the functional linear regression \eqref{eq:modelsimple} as 
\begin{align*}
	Y_t(r) &= \int_a^b \beta(r,s) X_t(s) \,ds + u_t(r) \\
	&= \int_a^b \beta(r,s) X_t^*(s) \,ds + u_t^*(r) \\
	&= \sum_{l=1}^K \beta_l(r) f_{lt} + u_t^*(r),
\end{align*}
where $u_t^*(r) = u_t(r) + \int_a^b \beta(r,s) \varepsilon_t(s) \,ds$ and $\beta_l(r) = \int_a^b \beta(r,s) \psi_l(s)$.
A direct consequence of condition \eqref{eq:nonpredictiveness} is
$$
E[u_t^*(r)f_{lt}] = 0 \quad \text{for all} \ r \in [a,b],
$$
which implies that $\beta_1(r), \ldots \beta_K(r)$ are the population coefficients in a regression of $Y_t(r)$ on $f_{1t}, \ldots, f_{Kt}$, which is a well-posed problem. From these coefficients alone, we can fully determine $\beta(r,s)$.
To gain some intuition for this result, note that we can always find a Gram-Schmidt orthogonalized factor structure 
	$$
	X_t(s)=\sum_{l=1}^K (f_{lt} + \langle \varepsilon_t , \psi_l \rangle) \psi_{l}(s) + \varepsilon_t(s) - \sum_{l=1}^K \langle \varepsilon_t , \psi_l \rangle \psi_{l}(s) \coloneqq \sum_{l=1}^K f_{lt}^* \psi_{l}(s) + \varepsilon_t^*(s).
	$$
Now consider any function $\psi_\perp$ orthogonal to $\{\psi_1,\dots,\psi_K\}$. Such a function captures the directions in $X_t$ that belong entirely to $\varepsilon_t^*$.
Together with Assumption \ref{as:identifiability} which constrains $\beta(r, \cdot)$ to live in the closed image of the covariance operator of $X_t$, and because $\varepsilon_t$ is non-predictive by \eqref{eq:nonpredictiveness}, we show that also  $\varepsilon_t^*$ is non-predictive and  no variation along $\psi_\perp$ can affect $Y_t$, forcing
$$
\int_a^b \beta(r,s) \psi_\perp(s) \,ds = 0.
$$
If this term were nonzero, $\beta$ would be extracting predictive information from $\varepsilon_t$, contradicting \eqref{eq:nonpredictiveness}.
Hence, $\beta(r,s)$ must lie entirely within the finite-dimensional space spanned by $\{\psi_1,\dots,\psi_K\}$. Therefore,
\begin{align} \label{eq:finiterankbeta}
	\beta(r,s) = \sum_{l=1}^K \beta_l(r) \psi_l(s).
\end{align}
For a formal proof of \eqref{eq:finiterankbeta}, see Lemma \ref{lem:ident} in the next section.\\
\indent The key object to identify the loading functions in our approximate factor structure setting is the cross covariance operator $C$, which is the integral operator with kernel function $c(r,s) = Cov(X_t(r), Y_t(s))$. The product with its adjoint operator is $D = CC^*$, which is the integral operator with kernel function $d(r,s) = \int_a^b c(r,q) c(s,q) \,dq$. Inserting the factor model \eqref{eq:approxfactormodel} into $c(r,s)$ and applying the non-predictiveness condition \eqref{eq:nonpredictiveness} implies
\begin{align*}
	c(r,s) = Cov(X_t(r), Y_t(s)) = (\Psi(r))'E[F_t Y_t(s)],
\end{align*}
and
\begin{equation}\label{eq:d}
	d(r,s) = \int_a^b c(r,q) c(s,q) \,dq = (\Psi(r))' \bigg( \int_a^b E[F_t Y_t(q)] E[Y_t(q) F_t]' \,dq \bigg) \Psi(s).
\end{equation}
Since $\int_a^b E[F_t Y_t(q)] E[Y_t(q) F_t]' \,dq$ is a full-rank $K \times K$ matrix by the predictiveness condition \eqref{eq:predictiveness}, we have 
$$
K = rank(C) = rank(D).
$$
Moreover, the loading functions $\psi_1(r), \ldots, \psi_K(r)$ must be linear combinations of the left-singular functions of $C$, which are the eigenfunctions of $D$. These linear combinations are only identified up to a rotation. To see this, note that $(\Psi(r))'F_t = (\Psi(r))'Q Q^{-1} F_t$ for any invertible $K \times K$ matrix $Q$. To fix the rotation, we impose a restriction on the matrix $\int_a^b E[F_t Y_t(q)] E[Y_t(q) F_t]' \,dq$ that appears in the definition of $d(r,s)$. We assume that
\begin{align}	\label{eq:rotation-simple}
	\int_a^b E[F_t Y_t(q)] E[Y_t(q) F_t]' \,dq = diag(\lambda_1, \ldots, \lambda_K), \quad \lambda_1 > \ldots > \lambda_K > 0,
\end{align}
which rotates the loading functions such that they are the left-singular functions of $C$, respectively the eigenfunctions of $D$, in descending order of singular values, i.e.
$$
d(r,s) = \sum_{l=1}^K \lambda_l \psi_l(r) \psi_l(s).
$$
Consequently, conditions \eqref{eq:nonpredictiveness} and \eqref{eq:predictiveness} suffice to identify the subspace spanned by the loadings $\psi_1, \ldots, \psi_K$, while \eqref{eq:rotation-simple} adds an extra restriction necessary to identify the loadings themselves up to a sign change.\\
\indent In conclusion, our approximate factor structure offers several key advantages. By focusing only on the predictive components of the functional regressor variables, our approach balances the trade-off between flexibility and tractability. Unlike the general framework in Section \ref{sec:ill-posed}, our model avoids the ill-posed inverse problem while still accommodating infinite-dimensional regressor spaces. Compared to the restricted case in Section \ref{sec:well-posed}, we relax the assumption that regressors must reside in a finite-dimensional subspace but even if finite-dimensionality is assumed to hold, our model still serves as a useful dimension reduction technique if only a small number of directions in the regressor are actually predictive. This balance enables both consistent estimation and the development of asymptotic theory for inference, which we will explore in the following sections.

%% file: Functional_factor_regression.tex
\section{Functional factor regression}\label{sec:FFR}

In this chapter, we extend our findings from Section \ref{sec:simpleFFR} to the general regression framework with multiple non-centered functional variables and additional scalar covariates as introduced in equation \eqref{eq:basemodel}, i.e.,
\begin{equation*}
	Y_t(r) = w_t'\alpha(r) + \sum_{j=1}^J \int_a^b \beta_j(r,s) X_{jt}(s) \,ds + u_t(r).
\end{equation*}
Now, every  $H$-valued regressor $X_{1t}, \ldots, X_{Jt}$ admits the approximate factor structure
\begin{equation}\label{eq:factor.structure}
	X_{jt}(s) = \mu_j(s) + \sum_{l=1}^{K_j} f_{ljt} \psi_{lj}(s) + \varepsilon_{jt}(s).
\end{equation}
$\mu_{j}(s)$ is the deterministic intercept function, $F_{jt} = (f_{1jt}, \ldots, f_{K_{j}jt})'$ is the vector of $K_j$ latent factors with $\Psi_j(s) = (\psi_{1j}(s), \ldots, \psi_{K_{j}j}(s))'$ being the vector of corresponding unknown deterministic loading functions and $\varepsilon_{jt}(s)$ is a potentially infinite-dimensional error component.\\
\indent In addition to Assumption \ref{as:exo} and \ref{as:identifiability}, the following set of restrictions is required to arrive at a fully identified regression model:

\begin{assumption}\label{as:identifiability2}
	\textbf{\textup{(Identification II).}}\\ For any $j=1, \ldots, J$, $t=1, \ldots, T$, and $r,s \in [a,b]$,
	\begin{itemize}
		\item[(a)] The error component of the factor structure satisfies \\ $E[Y_t(r)\varepsilon_{jt}(s)] = 0$, $E[F_{jt} \varepsilon_{jt}(s)] = 0$, $E[u_t(r) \varepsilon_{jt}(s)] = 0$, $E[w_t \varepsilon_{jt}(s)] = 0$, and $\int_a^b \beta_i(r,q) E[X_{it}(q) \varepsilon_{jt}(s)] \dd q = 0$ for all $i \neq j$.
		\item[(b)] The factor component satisfies
		$$
		\int_a^b E[F_{jt} Y_t(q)] E[Y_t(q) F_{jt}]' \,dq = diag(\lambda_{1j}, \ldots, \lambda_{K_j j}), \quad \lambda_{1j} > \ldots > \lambda_{K_j j} > 0.
		$$
		Moreover, $E[F_{jt}] = 0$ and $\psi_{1j}, \ldots, \psi_{K_j j}$ form an orthonormal system in $H$,
		\item[(c)] The vector $z_t = (w_t', F_{1t}', \ldots,  F_{Jt})'$ is covariance stationary and has a positive definite covariance matrix.
	\end{itemize}	
\end{assumption}
Assumptions \ref{as:identifiability2}(a) and (b) are the multiple regressors extension of the conditions \eqref{eq:nonpredictiveness} and \eqref{eq:predictiveness}. Part \ref{as:identifiability2}(a) ensures that each factor model’s functional error term has no predictive power for the dependent variable $Y_t$. Moreover, the error is uncorrelated with all other variables in the functional regression equation \eqref{eq:basemodel}. Part \ref{as:identifiability2}(b) guarantees that each factor exhibits non-zero correlation with the response variable, which makes the factors response-driven or predictive, in contrast to the FPC factors which only model variation in the regressors. Additionally, \ref{as:identifiability2}(b) serves as a normalization because the loadings and factors are generally only separable up to a rotation matrix. Lastly, Assumption \ref{as:identifiability2}(c) rules out perfect multicollinearity in a regression model.\\
\indent Following the same line of reasoning as in Section \ref{sec:simpleFFR}, Assumptions \ref{as:identifiability2}(a) and (b) imply that each functional regressor $X_{jt}$ with associated cross covariance operator $C_j$ and the product with its adjoint operator $D_j = C_j C_j^*$ admits the kernel representation
\begin{align*}
	d_j(r,s) &= (\Psi_j(r))' \left( \int_a^b E[F_{jt} Y_t(q)] E[Y_t(q) F_{jt}]' \,dq \right)  (\Psi_j(s)) \\
	&= (\Psi_j(r))' \text{diag}(\lambda_{1j}, \ldots, \lambda_{K_jj}) (\Psi_j(s)) \\
	&= \sum_{l=1}^{K_j} \lambda_{lj} \psi_{lj}(r) \psi_{lj}(s).
\end{align*}
Because the eigenequation of $D_j$ follows as 
$$
\int_{a}^{b}d_j\left( r,s\right)\psi_{lj}\left( s\right)\,ds= \sum_{k=1}^{K_j} \lambda_{kj} \psi_{kj}\left( r\right)\left\langle \psi_{kj}, \psi_{lj}\right\rangle = \lambda_{lj} \psi_{lj}\left( r\right)
$$ 
by Assumption \ref{as:identifiability2}(b), the pairs $\left( \lambda_{lj}, \psi_{lj}\right)$, $l=1,\ldots, K_j$ are identified as the eigenvalues and their corresponding eigenfunctions of $D_j$. This directly implies $rank(D_j)=K_j$. Note how the rank of $D_j$ might differ for each regressor $X_{jt}$, depending on the number of factors correlated with the regressand.\\
\indent Given the identified loading functions $\Psi_j(s)$, Lemma \ref{lem:ident} shows that each regression coefficient function admits a finite rank representation.
\begin{lemma}\label{lem:ident}
Under Assumptions \ref{as:exo}--\ref{as:identifiability2}, for any $j=1, \ldots, J$, $t=1, \ldots, T$, and $r,s \in [a,b]$,
\begin{itemize}
	\item[(a)] $\int_a^b \beta_j(r,s) \varepsilon_{jt}(s) \,ds = 0$ almost surely.
	\item[(b)] 
	$\int_a^b \beta_j(r,s) \psi_\perp(s)\,ds =0$
	for every
	$\psi_\perp\in Im(\Gamma_j)$ with $\langle\psi_\perp,\psi_{lj}\rangle=0$ for $l=1,\dots,K_j$.
	\item[(c)] Let $\beta_{lj}(r) = \int_a^b \beta_j(r,s) \psi_{lj}(s) \,ds$ and $B_j(r) = (\beta_{1j}(r), \ldots, \beta_{K_j j}(r))'$. Then,
	\begin{align} \label{eq:beta.ident}
		\beta_j(r,s) = \sum_{l=1}^{K_j} \beta_{lj}(r) \psi_{lj}(s) = (B_j(r))'\Psi_j(s).
	\end{align}
\end{itemize}	
\end{lemma}
Note that the loading functions $\Psi_j(s)$ are only identified up to a sign change because $m\psi_{lj}$ is still an eigenfunction if $\psi_{lj}$ is an eigenfunction, for any nonzero scalar $m$. The normalization $\left\| \psi_{lj} \right\| =1$ from Assumption \ref{as:identifiability2}(b) does not identify the sign either. However, the sign of the coefficient function $\beta_j(r, s)$ is still fully identified. If the sign of $\psi_{lj}(s)$ changes, the sign of the corresponding $\beta_{lj}(r)$  also flips such that  $\beta_j(r,s)$ remains unchanged.\\
\indent Until here, the last open issue concerns how to obtain $\beta_{lj}(r)$.
If we make use of the factor equation \eqref{eq:factor.structure}, the orthonormality of the loading functions, and the fact that $\varepsilon_{jt}$ lies in the null space of the regression operator by Lemma \ref{lem:ident}, the integral term in our regression equation can be written as
$$
\sum_{j=1}^J \int_a^b \beta_j(r,s) X_{jt}(s) \,ds =  \beta_{0}(r) + \sum_{j=1}^J \sum_{l=1}^{K_j} \beta_{lj}(r) f_{ljt},
$$
where $\beta_{0}(r) = \sum_{j=1}^J \int_a^b \beta_j(r,s) \mu_j(s) \,ds$.
Therefore, the fully functional regression model \eqref{eq:basemodel} becomes the functional factor regression
\begin{align}
	Y_t(r) &= w_t'\alpha(r) + \sum_{j=1}^J \int_a^b \beta_j(r,s) X_{jt}(s) \,ds + u_t(r) \notag \\
	&= w_t'\alpha^*(r) + \sum_{j=1}^J \sum_{l=1}^{K_j} \beta_{lj}(r) f_{ljt} + u_t(r) \notag \\
	&= z_t'B(r) + u_t(r), \label{eq:func.factor.regression}
\end{align}
which is a regression of $Y_t(r)$ on the $(N+\sum_{j=1}^J K_j)$-dimensional vector of regressors $z_t = (w_t', F_{1t}', \ldots, F_{Jt}')'$ with coefficient functions $B(r) = ((\alpha^*(r))',(B_1(r))', \ldots, (B_J(r))')'$, where $\alpha^*(r) = (\alpha_1^*(r), \alpha_2(r), \ldots, \alpha_N(r))'$ with $\alpha_1^*(r) = \alpha_1(r) + \beta_0(r)$.
Under Assumption \ref{as:identifiability2}(c), the coefficients are identified as
\begin{equation}\label{eq:OLS}
	B(r) = E\left[ z_t z_t'\right]^{-1} E\left[ z_t Y_t(r) \right].
\end{equation}
The factors admit a certain representation through projection coefficients, namely
\begin{align*}
	\langle X_{jt} - \mu_{j}, \psi_{lj} \rangle &= \sum_{l=1}^{K_j} f_{kjt} \langle \psi_{kj}, \psi_{lj} \rangle + \langle \varepsilon_{jt}, \psi_{lj} \rangle \\
	&= f_{ljt} + \langle \varepsilon_{jt}, \psi_{lj} \rangle.
\end{align*}
We see that the original factors $f_{ljt}$ from equation \eqref{eq:factor.structure} are latent and only partially identifiable because $\langle \varepsilon_{jt}, \psi_{lj} \rangle$ is unknown. However, the full identification of the factors is not required to identify all parameters of the model (see the discussion in \cite{Otto.2025}).
In fact, $f_{ljt}$ coincides with the projection coefficient $f_{ljt}^* := \langle X_{jt} - \mu_{j}, \psi_{lj} \rangle$ up to the noise term $\langle \varepsilon_{jt}, \psi_{lj} \rangle$, which is uncorrelated with $Y_t(r)$ by Assumption \ref{as:identifiability2}(a). Specifically, the latent factors $f_{ljt}$ and the projection coefficients $f_{ljt}^*$ can be used interchangeably in the functional factor regression \eqref{eq:func.factor.regression}. To see this, note that by \eqref{eq:beta.ident},
$$
	\sum_{j=1}^J \int_a^b \beta_j(r,s) (X_{jt}(s) - \mu_j(s)) \,ds = \sum_{j=1}^J \sum_{l=1}^{K_j} \beta_{lj}(r) \langle X_{jt} - \mu_j, \psi_{lj} \rangle,
$$
which implies that $B(r)$ are the regression coefficients obtained from a regression of $Y_t(r)$ on $w_t$ and $F_{jt}^* = (f_{1jt}^*, \ldots, f_{K_j jt}^*)'$ for $j=1, \ldots, J$. Therefore, the regression coefficient functions is identified using either $f_{ljt}$ or $f_{ljt}^*$. \\

\indent Before turning to the estimation of the model, we specify the necessary set of assumptions for this task.
\begin{assumption}\label{as:estimation}
	\textbf{\textup{(Estimation).}}\\ For any $j=1, \ldots, J$, $t=1, \ldots, T$, and $r \in [a,b]$,
	\begin{itemize}
		\item[(a)] The vector $(u_t, w_t, X_{1t}, \ldots, X_{Jt})_{t}$ is a weakly dependent $\alpha$-mixing process of size $-\nu/(\nu-2)$ for some $\nu \geq 4$ with $\sup_{t}E[(u_t(r))^\nu] < \infty$, $\sup_{t} E[(X_{jt}(r))^\nu] < \infty$, and $\sup_{t} E[(w_{it})^\nu] < \infty$.
		\item[(b)] $Y_t(r)$ and $X_{jt}(s)$ have differentiable sample paths, and $\alpha(r)$, $\beta_j(r,s)$, and $\Psi_j(s)$ are differentiable parameter functions.
	\end{itemize}	
\end{assumption}
Assumption \ref{as:estimation}(a) is sufficient to ensure consistent estimation of all parameters in our model. It only entails fairly mild moment restrictions on the regressors and regression errors, as well as $\alpha$-mixing at a rate which is common in the literature (see, e.g., \cite{White.2001}). Assumption \ref{as:estimation}(b) is a simple smoothness condition required to achieve a uniform central limit theorem for the functional regression coefficients.

%% file: Estimation.tex
\section{Estimation}\label{sec:Estimation}

As all parameters of the regression model are identified by now, we can estimate the population objects based on sample observations. Firstly, we need an asymptotically consistent moment estimator for the integral operator kernel $d_j(r,s)$. Based on that result, consistent estimates of the eigenpairs $(\lambda_{lj}, \psi_{lj})$, $l=1, \ldots, K_j$, follow straightforward. Secondly, the number of factors $K_j$, equal to the rank of $D_j$, has been treated as known until here. To estimate $K_j$, we develop a functional eigenvalue difference test based on the concepts introduced in \textcite{Wu.2018} for multivariate factor models. Lastly, we show that the estimators of the functional regression coefficients $\beta_j(r,s)$ are asymptotically normal with uniform convergence rates. All proofs are supplied in the Appendices \ref{app:Theorem1}--\ref{app:Theorem3}.

\subsection{Primitives}\label{subsec:Primitives}

Let the sample estimators for the first moment functions be
\begin{equation}\label{eq:sample.means}
	\overline{Y}(r) = \frac{1}{T} \sum_{t=1}^{T}Y_t(r), \quad \widehat{\mu}_{j}(s) = \frac{1}{T} \sum_{t=1}^{T}X_{jt}(s), \quad r,s\in [a,b].
\end{equation}
Then, the sample cross-covariance function is
\begin{equation}\label{eq:sample.crosscov}
	\widehat{c}_j(r,s) = \frac{1}{T} \sum_{t=1}^{T}(X_{jt}(r)-\widehat{\mu}_{j}(r))(Y_t(s)-\overline{Y}(s))  , \quad r,s\in [a,b],
\end{equation}
and the respective product of cross-covariance kernels, integrated over the de-meaned response variable $Y_t(q)$, follows as
\begin{equation}\label{eq:sample.d}
	\widehat{d}_j(r,s) = \int_{a}^{b} \widehat{c}_j(r,q)\widehat{c}_j(s,q) \,dq, \quad r,s\in [a,b].
\end{equation}
The integral operator $\widehat{D}_j$ with kernel function $d_j(r,s)$ has $T$ empirical eigenpairs $(\widehat{\lambda}_{lj}, \widehat{\psi}_{lj})$ with associated empirical factors $\widehat F_{jt} = (\widehat f_{1jt}, \ldots, \widehat f_{K_j j t})'$, $\widehat f_{ljt} = \langle X_{jt} - \widehat{\mu}_{j}, \widehat{\psi}_{lj} \rangle$.
In Theorem 1, we prove convergence of the sample primitives to their population counterparts. The convergence proof for $\widehat{D}_j$ builds upon consistent estimation of the mean function estimators and the cross-covariance operator $\widehat C_j$ with kernel function $\widehat{c}_j(r,s)$.

\begin{theorem}\label{theo:primitves}
	\textbf{\textup{(Primitives).}} By Assumption \ref{as:exo}--\ref{as:estimation} and as $T\rightarrow\infty$, for any $j=1,\ldots, J$,
	\begin{itemize}
		\item[(a)] $\| \widehat \mu_j - \mu_j \| = \mathcal{O}_P(T^{-1/2})$ and $\| \overline Y - E[Y_t] \| = \mathcal{O}_P(T^{-1/2});$
		\item[(b)] $\| \widehat{C}_j - C_j \|_{\mathcal{S}} = \mathcal{O}_P(T^{-1/2})$ and $\| \widehat{D}_j - D_j \|_{\mathcal{S}} = \mathcal{O}_P(T^{-1/2});$
		\item[(c)] $| \widehat{\lambda}_{lj} - \lambda_{lj} | = \mathcal{O}_P(T^{-1/2})$, for $l \geq 1$, where  $\lambda_{lj} \coloneqq 0 $ for $l > K_j;$
		\item[(d)] $\| s_{lj} \widehat{\psi}_{lj} - \psi_{lj} \| = \mathcal{O}_P(T^{-1/2})$, for $1\leq l \leq K_j$, where $s_{lj} \coloneqq \mathit{sign}(\langle \widehat{\psi}_{lj}, \psi_{lj}\rangle).$
	\end{itemize}
\end{theorem}

\subsection{Number of factors}\label{subsec:factor.number}

Next, we determine the number of common factors $K_j$. Popular methods in the multivariate statistics literature consider information criteria as in \textcite{Bai.2002} or maximize the ratio of subsequent eigenvalues as in \textcite{Ahn.2013}. Especially the latter approach is attractive for practitioners because it is computationally inexpensive and the respective scree plot allows for simple visual interpretations. However, \textcite{Xia.2015} formalized that such an estimator might suffer computational instability from $0/0$ type ratios. Referring back to our functional factor structure with the specific integral operator kernel $d_j(r,s)$, this problem amplifies. In detail, recall that $D_j$ has rank $K_j$, meaning all eigenvalues of $D_j$ beyond the $K_j$-th one are exactly zero. As the estimated eigenvalues $\widehat{\lambda}_{lj}$ converge to their population counterparts by Theorem \ref{theo:primitves}, the \textcite{Ahn.2013} type ratio estimator $\widehat{K}_j^{ER} = \widehat{\lambda}_{lj}/\widehat{\lambda}_{(l+1)j}$ is not consistent because the $(K_j+1)$-th eigenvalue is asymptotically zero and not uniformly bounded away from zero.\\
\indent To avoid this problem, we propose a functional equivalent of the eigenvalue difference estimator introduced by \textcite{Wu.2018}. Given the kernel function estimator $\widehat{d}_j\left(r,s \right)$, the proposed method converges to the true number of factors $K_j$ under no additional assumptions.\\
\indent As in the original paper, the idea is to find a monotonous function which converges to one for the first $K_j$ eigenvalues and to zero for all others. Formalized, we have
\begin{equation*}
	G(\widehat{\lambda}_{lj}) \rightarrow 
	\begin{cases}
		1, & \text{for $l=1,\ldots,K_j$}\\
		0, & \text{for $l=K_j+1, K_j+2, \ldots$}.
	\end{cases}
\end{equation*}
In order to incorporate the possibility for zero factors, i.e. all eigenvalues of $D_j$ are zero, we define the helper function
\begin{equation*}
	g_{lj} \coloneqq
	\begin{cases}
		1, & \text{for $l=0$}\\
		G(\widehat{\lambda}_{lj}), & \text{for $l=1,\ldots,K_j^{\text{max}}$}\\
		0, & \text{for $l=K_j^{\text{max}}+1$},
	\end{cases}
\end{equation*}
where $K_j^{\text{max}}$ is a user-specified positive constant. The eigenvalue difference estimator then follows as
\begin{equation}\label{eq:ed.estimator}
	\widehat{K}_j^{ED}=\argmax_{0\leq l \leq K_j^{\text{max}}} 	\left\lbrace g_{lj} - g_{(l+1)j}\right\rbrace.
\end{equation}
In the functional factor regression setting, the mock eigenvalue $g_{0j}=1$ is important because it serves as a simple model specification tool. If the estimator chooses $\widehat{K}_j^{ED}=0$ factors, the respective functional regressor $X_{jt}$ is implied to have no predictive power for the dependent variable and should be dropped from the regression. This feature is derived from the fact that we make use of the singular values of the cross-covariance operator between $X_{jt}$ and $Y_{t}$. In addition to the first mock eigenvalue for $l=0$, we also add $g_{lj}=0$ for $l=K_j^{\text{max}}+1$. If the eigenvalue difference estimator chooses $K_j^{\text{max}}$ factors, the practitioner is advised to redo their analysis with a larger number of possible eigenvalues as it is indicated that $K_j \geq K_j^{\text{max}}$.\\
\indent Lastly, we discuss the specific transformation function $G(\widehat{\lambda}_{lj})$ needed for the estimator. It takes the form
\begin{equation}\label{eq:arctan.trans}
	G(\widehat{\lambda}_{lj})=\frac{2}{\pi}\arctan\left(\frac{\gamma \ln(T) \widehat{\lambda}_{lj}}{\sqrt{T^{-1} \sum_{t=1}^T \|X_{jt} - \widehat{\mu}_j \|^2} \sqrt{T^{-1} \sum_{t=1}^T \|Y_t - \overline{Y} \|^2}} \right).
\end{equation}
The tuning parameter $\gamma$ is a freely chosen positive constant. We recommend to set its value according to some standard time series cross-validation procedure, such as expanding window cross-validation, that minimizes the mean squared error in a test set. Moreover, the transformation is indifferent to the specific scale of the empirical eigenvalues as we divide all estimates by the integrated sample standard deviation of  $X_{jt}$ and $Y_{t}$. It is clear that the inverse tangent function converges to $\pi / 2$ if its domain approaches infinity and it becomes zero for the input zero. Following this intuition, Theorem 2 proves consistency of the eigenvalue difference estimator for the specific function described in (\ref{eq:arctan.trans}).

\begin{theorem}\label{theo:number.factors}
	\textbf{\textup{(Number of factors).}} By Assumption \ref{as:exo}--\ref{as:estimation}, $K_{j}^\text{max} > K_j$, $\gamma > 0$, and for any $j=1,\ldots, J$, we have
	\begin{equation*}
		\lim_{T \rightarrow \infty} \text{Pr}\left( \widehat{K}_j^{ED}=K_j\right)=1.
	\end{equation*}
\end{theorem}

\subsection{Least squares estimator}\label{subsec:OLS}

Estimation of the functional factor regression model builds upon the least squares equation identified in (\ref{eq:OLS}).
Since the factors are latent unknown in practice, we define the vector of generated regressors
$$
	\widehat{z}_t \coloneqq (w_t', \widehat{F}_{1t}', \ldots, \widehat{F}_{Jt}')'
$$
with empirical factors $\widehat F_{jt} = (\widehat f_{1jt}, \ldots, \widehat f_{K_j j t})'$, $\widehat f_{ljt} = \langle X_{jt} - \widehat \mu_j, \widehat \psi_{lj} \rangle$.
The respective sample estimator for $B(r)$ follows as
\begin{equation}\label{eq:ols.estimator}
	\widehat{B}(r) = \left(\sum_{t=1}^{T}\widehat{z}_t \widehat{z}_t'\right)^{-1} \left(\sum_{t=1}^{T} \widehat{z}_t Y_{t}(r)\right), \quad r \in [a,b],
\end{equation}
where $\widehat{B}(r) = ((\widehat{\alpha}^*(r))',(\widehat{B}_1(r))', \ldots, (\widehat{B}_J(r))')'$ and $\widehat B_j(r) = (\widehat \beta_{1j}(r), \ldots, \widehat \beta_{K_j j}(r))'$.
The estimator for the $j$-th functional regression coefficient function as identified in Lemma \ref{lem:ident} is
\begin{equation}\label{eq:func.coefficient}
	\widehat{\beta}_j(r,s) =  \sum_{l=1}^{K_j} \widehat{\beta}_{lj}(r)  \widehat{\psi}_{lj}(s) = (\widehat{B}_j(r))'\widehat{\Psi}_j(s).
\end{equation}
For the scalar variables $i=2, \ldots, N$, the coefficient function estimator $\widehat \alpha_i(r)$ for $\alpha_i(r)$ is the $i$-th entry of $\widehat{B}(r)$, and the intercept estimator is reconstructed as $\widehat \alpha_1(r) = \widehat{\alpha}^*_1(r) - \sum_{j=1}^J \int_{a}^{b} \widehat \beta_j(r,s) \widehat \mu_j(s) \dd s$.

A challenge in deriving the asymptotic properties of these estimators lies in the fact that the generated regressor vector $\widehat{z}_t$ contains empirical factors instead of the true ones. The following theorem establishes asymptotic normality and provides a covariance function estimator that accounts for the underlying generated regressors problem.
\begin{theorem}\label{theo:ols.normality}
	\textbf{\textup{(Least squares).}} 
By Assumption \ref{as:exo}--\ref{as:estimation} and as $T \rightarrow \infty$, for any $j=1,\ldots, J$ and $r,s \in [a,b]$, we have
	\begin{align*}
		\frac{\sqrt T (\widehat \beta_j(r,s) - \beta_j(r,s))}{\sqrt{\widehat \Omega_j(r,s)}}  \overset{d}{\longrightarrow} \mathcal N(0, 1),
	\end{align*}
where the convergence holds uniformly for $\ r,s \in [a,b]$ with respect to the sup norm.
The covariance function estimator is given by
\begin{equation}\label{eq:cov.estimat}
	\widehat \Omega_j(r,s) = \frac{1}{T} \sum_{t=1}^T \Big(([\widehat Q^{-1}]_j \widehat z_t \widehat u_t(r))'\widehat \Psi_j(s)  + \widehat \omega_{jt}(r,s) \Big)^2,
\end{equation}
where all remaining terms are discussed in detail below.
\end{theorem}
The covariance function estimator in equation \eqref{eq:cov.estimat} includes several components that need further explanation. $\widehat{Q} = \frac{1}{T}\sum_{t=1}^T \widehat{z}_t\widehat{z}'_t$ represents the sample covariance matrix of the regressor vector, and $[\widehat{Q}^{-1}]_j$ is the submatrix of $\widehat{Q}^{-1}$ composed of the $K_j$ rows indexed from $N + \sum_{l=1}^{j-1} K_l + 1$ through $N + \sum_{l=1}^j K_l$ that belong to the $j$-th functional regressor. The residual term $\widehat{u}_t(r) = Y_t(r) - \widehat{z}'_t \widehat{B}(r)$ is the difference between the observed and fitted values.
Equation \eqref{eq:cov.estimat} represents the functional extension of the White-Huber-Eicker heteroskedasticity-consistent covariance estimator when the additional term $\widehat \omega_{jt}(r,s)$ is ignored.
This additional term is relevant because it captures the uncertainty arising from the estimation of factors and loadings, and it is essential for conducting valid statistical inference on $\beta_j(r, s)$. It is given by
\begin{align*}
	\widehat \omega_{jt}(r,s) = (\widehat \Psi_j(s))'[\widehat Q^{-1}]_j \sum_{k=1}^J (\overline z  \widehat F_{kt}' - \overline{z F_k'} \widehat G_{kt} ) \widehat B_k(r) + (\widehat \Psi_j(s))' \widehat G_{jt} \widehat B_j(r) + \widehat \varepsilon_{jt}(s) \widehat h_{jt}' \widehat B_j(r),
\end{align*}
where $\overline z = \frac{1}{T} \sum_{t=1}^T \widehat z_t$, $\overline{z F_k'} = \frac{1}{T} \sum_{t=1}^T \widehat z_t \widehat F_{kt}'$, $\widehat \gamma_{lj}(r) = \frac{1}{T} \sum_{t=1}^T \widehat f_{ljt} (Y_t(r) - \overline Y(r))$, $\widehat y_{ljt} = \langle Y_t - \overline Y, \widehat \gamma_{lj} \rangle$, $\overline{f_{mj} y_{lj}} = \frac{1}{T} \sum_{t=1}^T \widehat f_{mjt} \widehat y_{ljt}$, 
$\widehat \varepsilon_{jt}(s) = X_{jt}(s) - \widehat \mu_j(s) - \widehat F_{jt}' \widehat \Psi_j(s)$,
 $\widehat G_{jt}$ is the $K_j \times K_j$ matrix with $(l,m)$-entry
$$
	[\widehat G_{jt}]_{lm} = (\widehat \lambda_{lj} - \widehat \lambda_{mj})^{-1}(\widehat f_{mjt} \widehat y_{ljt} - \overline{f_{mj} y_{lj}} + \widehat f_{ljt} \widehat y_{mjt} - \overline{f_{lj} y_{mj}}) 1_{\{l\neq m\}},
$$ 
and $\widehat h_{jt}$ are the vectors of length $K_j$ with $l$-th entries $\widehat h_{ljt} = \widehat \lambda_{lj}^{-1} \widehat y_{ljt}$. Further details are given in the appendix.\\
\indent For each fixed pair $(r,s)$, we define the heteroskedasticity-consistent and corrected standard errors as
$$
se_j(r,s) = \frac{\sqrt{\widehat \Omega_j(r,s)}}{\sqrt T}.
$$
Given any significance level $\alpha$ and the $(1-\frac{\alpha}{2})$-quantile $z_{(1-\frac{\alpha}{2})}$ of the the standard normal distribution, these standard errors allow us to formulate a confidence region for $\beta_j(r,s)$:
$$
I_j(r,s) = \bigg[\widehat \beta_j(r,s) - z_{(1-\frac{\alpha}{2})} se_j(r,s); \widehat \beta_j(r,s) + z_{(1-\frac{\alpha}{2})} se_j(r,s)\bigg].
$$
As a direct consequence of Theorem \ref{theo:ols.normality}, these intervals are asymptotically valid. 
Specifically, for any fixed $r,s \in [a,b]$,
$$
\lim_{T \to \infty} P\Big(\beta_j(r,s) \in I_j(r,s) \Big) = 1-\alpha.
$$

\subsection{Practical implementation}\label{subsec:implementation}
After the methodological technicalities and asymptotics, this section serves as a short summary to help practitioners estimate a functional factor regression. The \textsf{R} package "ffr"\footnote{Available from \href{https://github.com/luiswn/ffr}{https://github.com/luiswn/ffr}} accompanying this paper provides a user friendly application of the following steps.\\[7pt]
\textbf{Step 1: Transform observed data to functions.} Normally, empirical functional data is only available stored in high-dimensional vectors. Therefore, construct the response variable and all functional regressors using some standard method such as basis expansion. Common bases choices are natural splines, smoothing splines or B-splines (see \textcite{Kokoszka.2021} or any other introductory textbook for more detailed information).\\[7pt]
\textbf{Step 2: Estimate the primitives.} For each functional regressor $j=1,\ldots,J$, compute the sample means $\overline{Y}(r)$ and $\widehat{\mu}_{j}(s)$ according to (\ref{eq:sample.means}), then construct the sample cross-covariance $\widehat{c}_j\left(r,s \right)$ from (\ref{eq:sample.crosscov}) in order to estimate the integral kernel $\widehat{d}_j\left(r,s \right)$ as described in (\ref{eq:sample.d}). Lastly, choose some large enough $K_{j}^\text{max}$ and compute the eigenpairs $\lbrace(\widehat{\lambda}_{lj}, \widehat{\psi}_{lj})\rbrace_{l=1}^{K_{j}^\text{max}}$, as well the corresponding factors $\widehat{f}_{ljt}=\langle X_{jt}-\widehat{\mu}_{j}, \widehat{\psi}_{lj} \rangle, l=1, \ldots, K_{j}^\text{max}$.\\[7pt]
\textbf{Step 3: Estimate the number of factors.} Set the tuning parameter $\gamma$ to some positive value or apply any cross-validation technique of your choice. For each functional regressor $j=1,\ldots,J$, use the estimated eigenvalues $\lbrace \widehat{\lambda}_{lj}\rbrace_{l=1}^{K_{j}^{\text{max}}}$ to determine the number of factors $\widehat{K}_j^{ED}$ according to (\ref{eq:ed.estimator}) and (\ref{eq:arctan.trans}).\\[7pt]
\textbf{Step 4: Estimate the functional regression.} Store an intercept dummy, all relevant factors $\lbrace\widehat{f}_{l1t}\rbrace_{l=1}^{\widehat{K}_1^{ED}}, \ldots, \lbrace\widehat{f}_{lJt}\rbrace_{l=1}^{\widehat{K}_J^{ED}}$ and any other scalar valued exogenous regressors $w_t$ in a vector $\widehat{z}_t$, then estimate the regression model for a sufficiently large number of points on the curve $Y_t(r)$ according to (\ref{eq:ols.estimator}). Finally, for each functional regressor $j=1,\ldots,J$, use the relevant regression coefficients $\widehat{B}_j(r)$ and the respective eigenfunctions $\lbrace \widehat{\psi}_{lj}\rbrace_{l=1}^{\widehat{K}_j^{ED}}$ to recover the functional coefficient function $\widehat{\beta}_j(r,s)$ according to (\ref{eq:func.coefficient}). By calculating each coefficient function's covariance estimator as in \eqref{eq:cov.estimat}, statistical inference procedures like hypothesis testing or constructing confidence intervals follow straightforwardly. 

%% file: Simulation.tex
\section{Simulation}\label{sec:Simulation}

In the following, we verify the favorable finite sample properties of our functional regression estimator and confirm the asymptotic normality results needed for confidence bands and hypothesis tests. In our generic data generating process (DGP), each regressor follows a factor structure defined as 
\begin{equation*}
	X_{jt}(s) = \sum_{l=1}^{K}f_{ljt}v_{l}(s) + \sum_{l=K+1}^{3K}\varepsilon_{ljt}v_{l}(s), \quad s \in [0,1],
\end{equation*}
where $v_{1}(s)=1, v_{(2z)}(s)=\sqrt{2}\sin(2z\pi s), v_{(2z+1)}(s)=\sqrt{2}\cos(2z\pi s)$ is the Fourier basis, the factors are independently generated according to $f_{ljt} \mathtt{\sim} \mathcal{N}(0,1)$, the errors are $\varepsilon_{ljt}\mathtt{\sim} \mathcal{N}(0,1)$ distributed and the number of relevant factors is set to $K=3$. On the other hand, the regression equation follows as
\begin{equation*}
	Y_t(r) = \int_{0}^{1} \beta_1(r,s)X_{1t}(s) \,ds + \int_{0}^{1} \beta_2(r,s)X_{2t}(s) \,ds + \sum_{i=1}^{2K} u_{it}\rho_i(r), \quad r,s \in [0,1].
\end{equation*}
The bivariate coefficient functions are defined as $\beta_j(r,s)=\sum_{l=1}^{K}v_{l}(r)\beta_{lj} v_{l}(s)$, where $\beta_{lj}$ is deterministically chosen in a way that ensures $rank(\beta_j(r,s))=K$ (see Appendix \ref{app:beta_spec} for the exact specification). Lastly, $\rho_i(r)=\binom{I+1}{i}r^i(1-r)^{I+1-i}, i=1,\ldots,I$, are the Bernstein basis polynomials. We differentiate between homoskedastic regression errors generated from $u_{it} \mathtt{\sim} \mathcal{N}(0,1)$ (DGP1), and conditionally heteroskedastic errors defined as $u_{it} \mathtt{\sim} \mathcal{N}(0,f_{1t}^2)$ (DGP2).\\
\indent For both data generating processes, the average bias and the pointwise average coverage rate of 95\% confidence intervals for the regression coefficient function $\beta_1(r,s)$ are reported. Firstly, the bias of our functional factor regression (FFR) estimator is compared to the theoretically constructed case where the factors and loadings are known and the only error stems from the regression step. Secondly, we present three types of coverage rates which differ by the covariance function estimator used to construct the intervals. In the constructed case with known factors and loadings, the heteroskedasticity consistent (HC) estimator $ T^{-1}\sum_{t=1}^T ((Q^{-1} z_t \widehat u_t(r))'\Psi(s))^2$ is expected to produce bands with asymptotically correct coverage. Next, we report the coverage share with HC covariance estimator but the factors and loadings now have to be estimated themselves beforehand. This case differs from our functional factor regression covariance estimator $\widehat \Omega(r,s)$ from equation \eqref{eq:cov.estimat} as it misses the correction term $\widehat \omega_{t}(r,s)$ which is necessary to account for the generated regressor issue. Therefore, we expect the uncorrected covariance estimator to result in incorrect coverage results while the FFR approach should yield asymptotically accurate rates.\\
\begin{table}\centering
	\caption{Bias and point-wise average coverage rate of 95\% CIs for $\beta_1(r,s)$. All results are based on a discrete grid of 200 equidistant points for $r,s \in [0,1]$ and 5000 Monte Carlo simulations.}
	\begin{tabular}{ll wc{2cm} wc{2cm} wc{2cm} wc{2cm}}				
		\hline
		& & $T = 50$ & $T = 100$ & $T = 500$ & $T = 1000$\\
		\toprule
		\midrule
		\multicolumn{1}{l}{\multirow{7}{*}{DGP 1: Homo}} & &\multicolumn{4}{c}{Bias} \\
		\cmidrule{3-6}
		&True param. & 0.001 & 0.000 & 0.000 & 0.000 \\
		&FFR. & 0.011 & 0.005 & 0.001 & 0.000 \\
		& & \multicolumn{4}{c}{Coverage rate} \\
		\cmidrule{3-6}
		&True param. & 0.906 & 0.928 & 0.943 & 0.944 \\
		&Uncorrected & 0.456 & 0.419 & 0.346 & 0.327 \\
		&FFR & 0.858 & 0.894 & 0.936 & 0.943 \\
		\midrule
		\multicolumn{1}{l}{\multirow{7}{*}{DGP 2: Hetero}} & &\multicolumn{4}{c}{Bias} \\
		\cmidrule{3-6}
		&True param. & 0.004 & 0.000 & 0.000 & 0.000 \\
		&FFR & 0.016 & 0.012 & 0.000 & 0.000 \\		
		& & \multicolumn{4}{c}{Coverage rate} \\
		\cmidrule{3-6}
		&True param. & 0.881 & 0.915 & 0.940 & 0.942 \\
		&Uncorrected & 0.516 & 0.526 & 0.559 & 0.570 \\
		&FFR & 0.858 & 0.896 & 0.937 & 0.943 \\
		\midrule
		\bottomrule
	\end{tabular}
	\label{table:simu.estimator}
\end{table}
\indent The simulation outcome is displayed in Table \ref{table:simu.estimator}. To start with the bias results, it becomes apparent that our FFR estimator is consistent as it exhibits a low bias even in small samples which vanishes completely as $T$ increases. This holds for both data generating processes. Regarding the coverage rates, it firstly should be noted how the confidence intervals constructed from the uncorrected covariance are generally incorrect even for large $T$ which is in line with the generated regressors literature. On the other hand, our FFR coverage rates as described in Theorem \ref{theo:ols.normality} asymptotically converge to the true value 0.95. These rates are of comparable size as the ones obtained from the unfeasible estimator with known factors and loadings and no first-stage error.\\
\begin{figure}[t]
	\centering
	\includegraphics[width=1.0\textwidth]{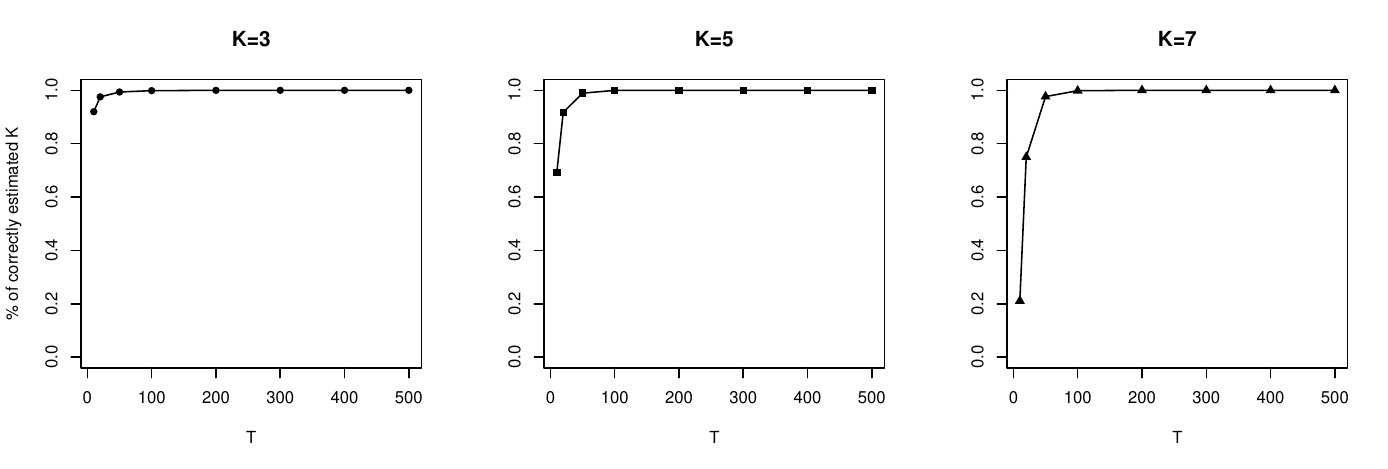}
	\caption{Simulation performance of the functional eigenvalue difference test applied to $X_{1t}(s)$ with 5000 Monte Carlo repetitions for different numbers of relevant factors $K$. The y-axis shows the share of correctly estimated cases during the 5000 simulations. The tuning parameter is set to $\gamma=1$.}
	\label{fig:simu.numb.factors}
\end{figure}
\indent Lastly, we examine the finite sample properties of the functional eigenvalue difference test introduced in Chapter \ref{subsec:factor.number} by determining the number relevant factors for $X_{1t}(s)$. The data generating process remains the same as above, we only change the number of factors $K$ between $3, 5$ and $7$. Figure \ref{fig:simu.numb.factors} shows that the estimator yields robustly correct results for all cases, even in settings with relatively small samples. The estimator lacks some precision when the number of observations drops below $50$. Moreover, it is noticeable that this issue amplifies with a larger number of factors $K$. 

%% file: Application_to_electricity_price_curve_forecasting.tex
\section{Application to electricity price curve modeling}\label{sec:Application}

In this section, we apply the functional factor regression method to model and forecast electricity spot prices in three different markets. For this task, understanding the institutional power market framework first is necessary. The European and North American wholesale short-term power trading is set-up as a one day ahead price auction. The 24 hourly spot prices for the next day are simultaneously settled when the gates close at 12 p.m.\ noon. Furthermore, substantial relevant information for this price setting mechanism is widely accessible to all market participants. The European Network of Transmission System Operators (ENTSO-E) exemplifies this type of comprehensive data source. The association provided with a legal mandate by the European Union requires its local transmission system operators (TSOs) to publish one day ahead load (i.e. demand) forecasts at least two hours before the gate closure time\footnote{\url{https://transparency.entsoe.eu/}}. Hence, it seems like a feasible goal to retrace the publicly available data which serves as a lower bound of the information set that wholesale traders make use of to buy and sell electricity\footnote{For an extensive discussion of the European electricity spot market design, see \textcite{Bichler.2021}}.\\
\indent On the other hand, there are clear economic motives for this modeling and forecasting task. Firstly, electricity still can not be stored efficiently. Therefore, demand spikes result in potentially untempered spot price fluctuations (\cite{Liebl.2013}). Secondly, increasing shares of renewable energy sources in the European electricity mix make the power generation more volatile which directly impacts energy price volatility (\cite{Lago.2021}). Modeling these relationships is therefore a promising avenue to achieve accurate price forecasts. Precise predictions, in turn, create substantial economic value across the market: individual generators can optimize their production schedules to maximize profits during high-price periods, large industrial consumers can shift energy-intensive operations to lower-price intervals to reduce costs, and traders can identify profitable arbitrage opportunities between day-ahead and real-time markets, enhancing overall market liquidity and efficiency.\\
\begin{figure}[t]
	\centering
	\includegraphics[scale=0.75]{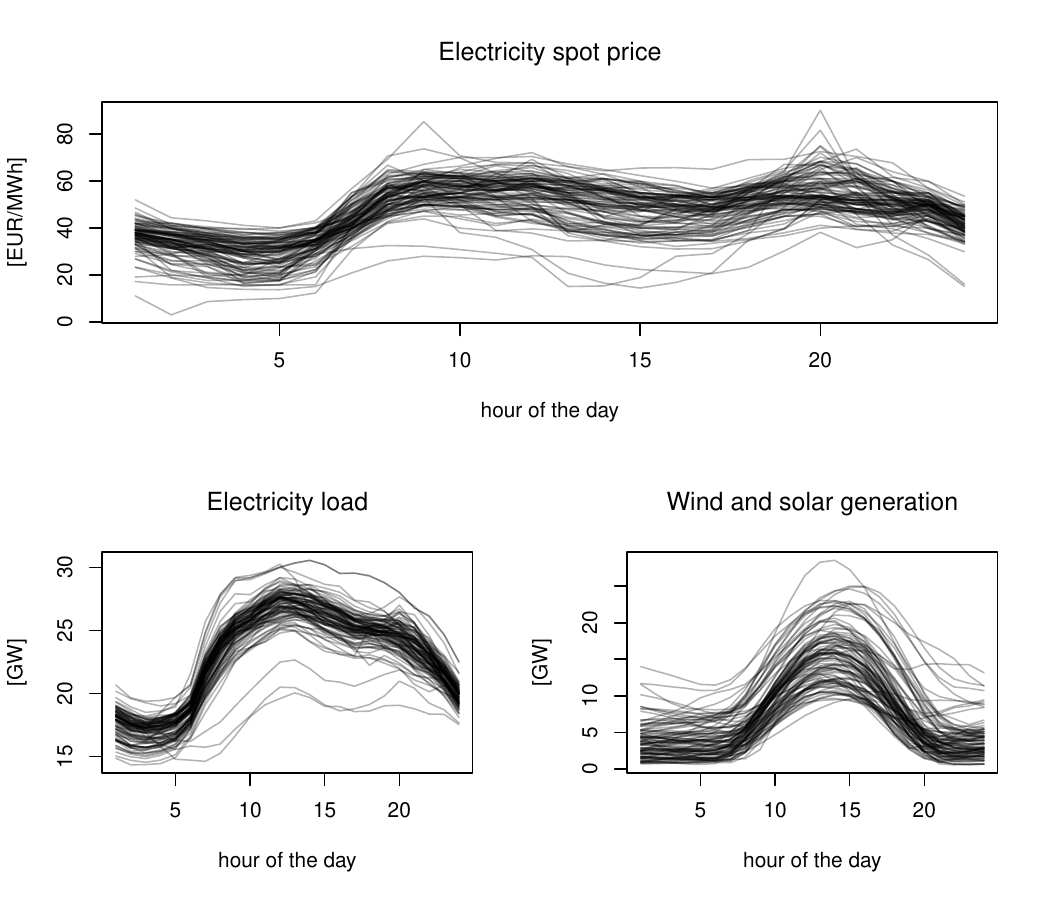}
	\caption{Energy market variables from the German EPEX-DE data set provided by \textcite{Lago.2021}. Each plot consists of 100 consecutive daily sample curves.}
	\label{fig:func.elec.data}
\end{figure}
\indent Following from the simultaneous auction based price setting, electricity spot prices are usually considered 24-dimensional multivariate time series (\cite{Uniejewski.2016}). However, neighboring data points are highly correlated and plots reveal that inter-day price data can be considered functional observations with regard to some intra-day time unit (e.g. hours, minutes, seconds, etc.). Figure \ref{fig:func.elec.data} shows 100 sample functions for three different energy market related variables to highlight this point. Hence, we model spot prices using our functional factor regression approach. The goal is not only to forecast but specifically reveal and interpret conditional correlation structures in this market through statistical inference. While our consistently estimated model parameters in combination with the novel asymptotic normality result allow for this task, alternative estimation techniques like regularization methods which are necessary to handle the high-dimensional data are often hardly interpretable and require additional post-estimation steps for valid inference (see, for example, \textcite{Lee.2016}).\\
\indent The data set we consider is provided by \textcite{Lago.2021}. The "open-access benchmark" energy data set that entails relevant regressor variables for five different power markets. Focusing on the German electricity market first, the data spans 2184 days (i.e. six 364-day "years") between 2012 and 2017. In our forecasting exercise, the first four years are assigned to model training and the remaining two years are used for forecast evaluation. Besides the hourly spot prices obtained from the European Power Exchange (EPEX), two exogenous regressors observed at the same frequency are available. Firstly, the day ahead load forecast in the TSO Amprion zone which mainly covers the whole of Western Germany. Load is a measure of the power absorbed by all installations connected to the transmission or distribution network. In economic terms, it can be understood as electricity consumption or demand. As described above, this day-ahead demand forecast is publicly available on the ENTSO-E transparency platform two hours before each daily auction. Secondly, the data set provides day ahead wind and solar generation forecasts in the three largest TSO zones. The variable is a supply-side measure and freely available on the system operators' websites.\\
\indent We set up the following fully-functional regression model in order to study German electricity spot prices:
\begin{align}\label{eq:elec.FAFR}
	\begin{split}
	P_t(r) = {}& w_t'\alpha(r) + \int_{0}^{1} \varphi_{1}(r,s)P_{t-1}(s) \,ds + \int_{0}^{1} \varphi_{2}(r,s)P_{t-2}\left(s\right) \,ds + \int_{0}^{1} \varphi_{5}(r,s)P_{t-5}(s) \,ds\\
	& + \int_{0}^{1} \beta_{1,1}(r,s)L_{t}(s) \,ds + \int_{0}^{1} \beta_{1,2}(r,s)L_{t-1}(s) \,ds + \int_{0}^{1} \beta_{1,5}(r,s)L_{t-5}(s) \,ds \\
	& + \int_{0}^{1} \beta_{2,1}(r,s)G_{t}(s) \,ds + \int_{0}^{1} \beta_{2,2}(r,s)G_{t-1}(s) \,ds + \int_{0}^{1} \beta_{2,5}(r,s)G_{t-5}(s) \,ds + u_t(r). \\
	\end{split}
\end{align}
$P_t\left(r\right)$ describes the daily price curve, $L_{t}\left(s\right)$ is the curve of load forecasts for day $t$ available one day before, $G_{t}\left(s\right)$ is the curve of wind and solar generation forecasts for day $t$ available one day before, and $w_t$ stores the intercept, as well as one-hot encoded weekday dummies. Because only weekdays are considered in all subsequent analyses, the chosen regressor dynamics are one day lags, two days lags and one week (i.e five days) lags. The intra-day time domains $r$ and $s$ are scaled to lie on the unit interval $[0,1]$. We restrain from normalizing the data in order to provide straight-forward interpretation of the functional regression parameters.\\ 
\begin{table}[!t]
	\caption{Results from the functional eigenvalue difference test introduced in section \ref{subsec:factor.number} for the nine regressors specified in equation (\ref{eq:elec.FAFR}). The optimal tuning parameter $\gamma=93$ was found through cross-validation.}
	\label{ta:epf.K}
	\centering
	\begin{tabular}{l|ccccccccc}
		\toprule
		\midrule
		\textbf{Regressor} & $P_{t-1}$ & $P_{t-2}$ & $P_{t-5}$ & $L_{t}$ & $L_{t-1}$ & $L_{t-5}$ & $G_{t}$ & $G_{t-1}$ & $G_{t-5}$ \\
		\midrule
		$\mathbf{\widehat{K}_{ED}}$ & 5 & 4 & 4 & 5 & 4 & 4 & 5 & 4 & 3 \\
		\bottomrule
	\end{tabular}
\end{table}
\indent The regression model (\ref{eq:elec.FAFR}) is estimated as explained in chapter \ref{sec:Estimation}. Table \ref{ta:epf.K} shows the eigenvalue difference test results for the number of factors $K_j$ for each of the nine functional regressors. By splitting the data 60/40 and performing an expanding window cross-validation with regard to the one day ahead forecasting performance, we find an optimal tuning parameter of $\gamma=93$.\\
\begin{figure}[!t]
	\centering
	\begin{subfigure}{0.49\textwidth}
		\centering
		\includegraphics[width=\textwidth]{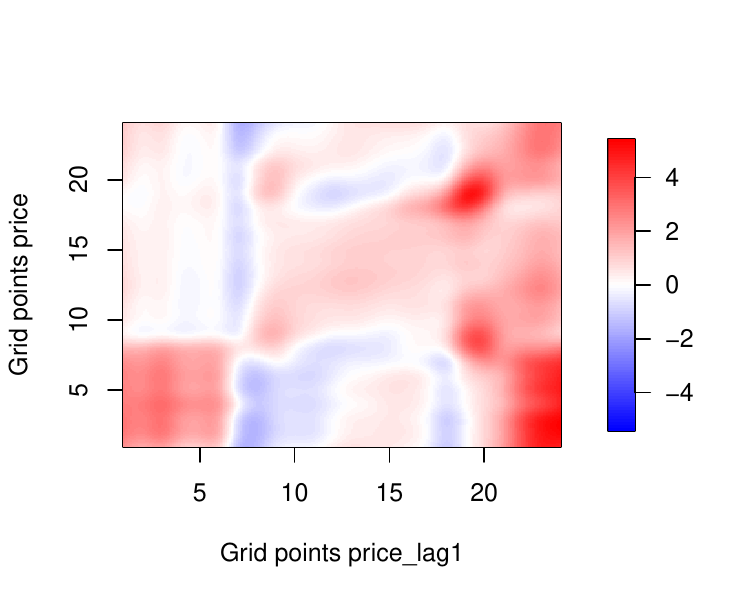}
		\caption{$\widehat{\varphi}_1(r,s)$: impact of last working day's price $P_{t-1}(s)$ on today's price $P_{t}(r)$.}
		\label{fig:phi.func}
	\end{subfigure}
	\begin{subfigure}{0.49\textwidth}
		\centering
		\includegraphics[width=\textwidth]{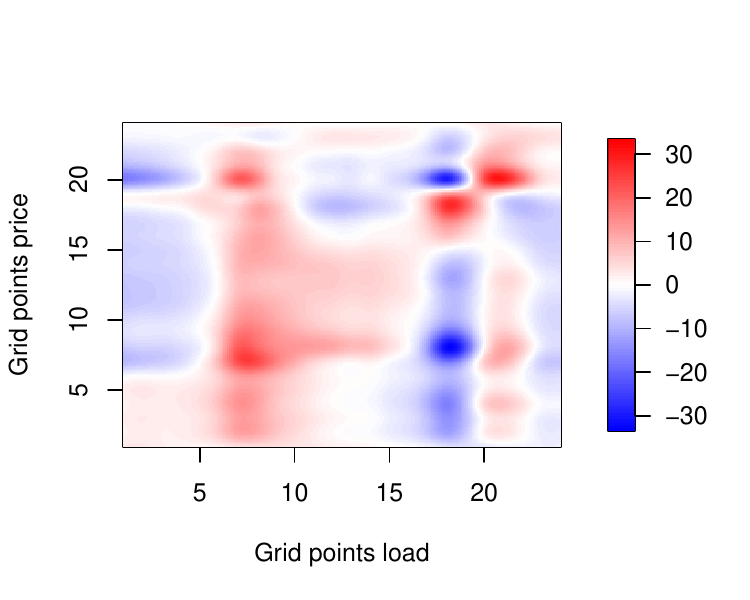}
		\caption{$\widehat{\beta}_{1,1}(r,s)$: impact of day-ahead load forecast $L_t(s)$ on today's price $P_{t}(r)$.}
		\label{fig:beta.func1}
	\end{subfigure}
	\begin{subfigure}{0.49\textwidth}
		\centering
		\includegraphics[width=\textwidth]{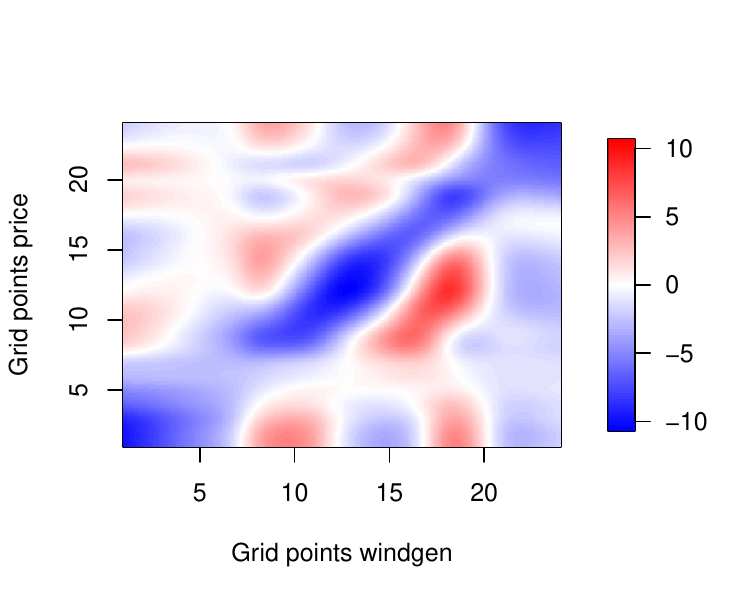}
		\caption{$\widehat{\beta}_{2,1}(r,s)$: impact of day-ahead wind and solar generation forecast $G_t(s)$ on today's price $P_{t}(r)$.}
		\label{fig:beta.func2}
	\end{subfigure}
	\caption{Estimated functional regression coefficients of model (\ref{eq:elec.FAFR}). The number of factors chosen to estimate the functional parameters in figure \ref{fig:phi.func}--\ref{fig:beta.func2} are shown in Table \ref{ta:epf.K}. The underlying German power market data spans all working days between 2012 and 2017.}
	\label{fig:coef.funcs}
\end{figure}
\indent Because the regression coefficients are uni- and bivariate functions, we require a visual interpretation. Figure \ref{fig:coef.funcs} shows heatmap plots of the three functional slope coefficients associated with $P_{t-1}$, $L_t$ and $G_t$. The first plot visualizes the impact of yesterday's spot price at hour $s$ on today's price at hour $r$ (Figure \ref{fig:phi.func}). To begin with, the positive diagonal values indicate price persistence as higher one-day-before prices lead to price increases at the same hour today. At the peak, a 1 Euro per MWh increase in yesterday's 8 p.m.\ price leads to an average 4.92 Euro per MWh increase in today's 7 p.m.\ price, ceteris paribus. Secondly, the coefficient plot reveals three separately interpretable regions of positive values on the surface. Looking at the red area in the bottom left of the plot, we see that yesterday's prices between midnight and roughly 7 a.m.\ mostly impact today's nighttime prices. The second field is marked by the four peaks at around 8 a.m.\ and 8 p.m.\ yesterday and today. Similarly to the first field, this surface shows how working hour electricity prices can mainly be explained by yesterday's working hour prices. While the 8 p.m.\ peak is expected, it is striking that the last day's 8 p.m.\ evening price is also a strong signal for the next day's 8 a.m.\ morning price. Lastly, the coefficient plot shows that the previous day's last hourly prices heavily influences all prices of the following day which is in line with the "end-of-day effect" described in \textcite{Maciejowska.2016} and \textcite{Ziel.2016}. Since our model already controls for electricity load forecasts and renewable energy generation, the observed price persistence patterns likely reflect market dynamics beyond the merit order effect. The temporal dependencies can stem from several potential mechanisms. Primarily, market participants might use the recent price history as anchoring points for their current bidding strategies. In addition, operational constraints like minimum up/down times and ramping limitations of power plants likely lead to temporal dependencies. The distinct day and night pattern might reflect different market regimes like liquidity differences or changes in the composition of market participants. To validate these findings statistically, Figure \ref{fig:phi.p} presents pointwise p-values from a two-tailed t-test on difference to zero based on our new inferential method. The dark red areas within the solid contour lines indicate p-values smaller than 0.01. The test result confirms the significance of three price patterns: same-hour effects, nighttime price links, and end-of-day influences. Hence, our paper delivers support for the end-of-day effect hypothesis beyond the simple interpretation of unconditional correlation structures.\\
\indent The day-ahead load forecast coefficient function (Figure \ref{fig:beta.func1}) shows how expected demand affects prices. The positive values along the diagonal and its adjacent hours align with economic theory as higher demand forecasts lead to higher prices. Quantitatively, a 1 GW increase in the 8 a.m.\ load forecast corresponds to an average 21.64 Euro per MWh increase in the contemporaneous spot price, ceteris paribus. While the strongly negative relationship between afternoon load forecasts and night/morning prices appears difficult to interpret, the corresponding p-value heatmap (Figure \ref{fig:beta1.p}) indicates that this effect lacks statistical significance. In fact, our inference result reveals that load forecasts overall play a surprisingly minor role in explaining electricity prices when conditioning on renewable energy generation forecasts and lagged prices.\\
\indent Lastly, the supply-side effects are captured by the day-ahead wind and solar generation forecast coefficient function in Figure \ref{fig:beta.func2}. In line with theory, the diagonal shows strongly negative values, indicating that higher renewable forecasts reduce prices in the corresponding hour. The directly adjacent off-diagonals, however, show substantial positive influences, suggesting the presence of ramping costs. This pattern emerges as conventional producers must recover the additional costs incurred from cycling their power plants in response to fluctuating renewable generation. When high renewable output is forecasted, conventional generators face a choice between reducing output and incurring ramping costs, shutting down entirely and incurring shutdown and later startup costs, or running at minimum load with a potential loss during peak renewable hours. To recover these costs, generators likely increase their bids in hours surrounding the high renewable periods, leading to higher prices in these adjacent hours. When examining the respective p-value plot in Figure \ref{fig:beta2.p}, we can conclude that the variable is highly relevant for modeling electricity prices as large regions share significant p-values smaller than 0.01. The complete set of bivariate coefficient functions and their pointwise p-values can be found in Appendix \ref{app:func.reg.coefs} and \ref{app:func.p.vals}. Interactive 3D surface figures and additional model visualizations, including statistical inference and goodness-of-fit analyses, are available at the supplementary website: \url{https://luiswn.github.io/ffr-visualizations/}.\\
\begin{figure}[!t]
	\centering
	\begin{subfigure}{0.32\textwidth}
		\centering
		\includegraphics[width=\textwidth]{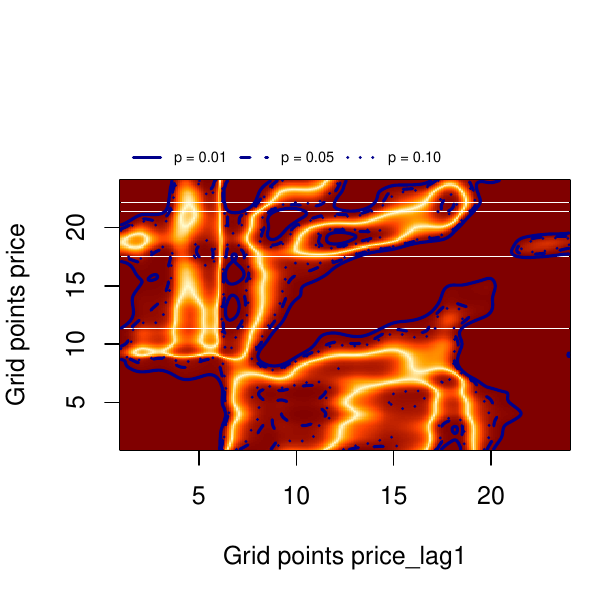}
		\caption{p-values for $\widehat{\varphi}_1(r,s)$}
		\label{fig:phi.p}
	\end{subfigure}
	\begin{subfigure}{0.32\textwidth}
		\centering
		\includegraphics[width=\textwidth]{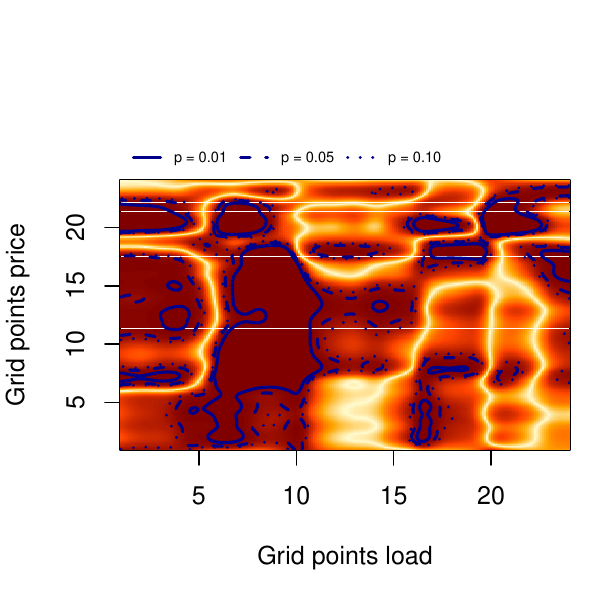}
		\caption{p-values for $\widehat{\beta}_{1,1}(r,s)$}
		\label{fig:beta1.p}
	\end{subfigure}
	\begin{subfigure}{0.32\textwidth}
		\centering
		\includegraphics[width=\textwidth]{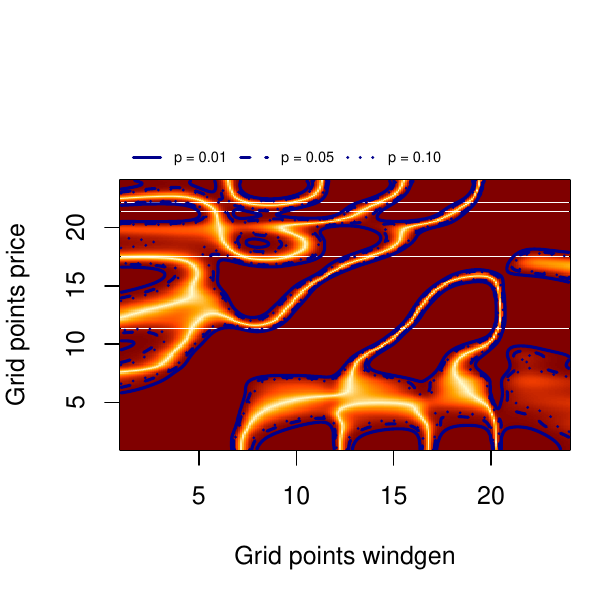}
		\caption{p-values for $\widehat{\beta}_{2,1}(r,s)$}
		\label{fig:beta2.p}
	\end{subfigure}
	\caption{Pointwise p-values from a two-tailed t-test on difference to zero. Dark red color indicates smaller values and the contour lines as defined in the legends indicate significance according to standard alpha levels of 0.01, 0.05 and 0.1.}
	\label{fig:p.funcs}
\end{figure}
\indent The second part of this application chapter deals with the out-of-sample performance of our model when forecasting electricity price curves one day ahead. The relevant test set consists of two years of working days which corresponds to 520 data points. Besides the German energy data, we forecast prices on the European power market of the Nordic countries Nord Pool (NP) and the U.S. American Pennsylvania–New Jersey–Maryland (PJM) market. All three data sets are equivalent in length as they are taken from the same open-access benchmark framework provided by \textcite{Lago.2021}\footnote{The exogenous regressors of NP and PJM differ slightly from the German ones discussed above. Consult the \textcite{Lago.2021} paper for detailed information concerning the definitions.}.\\
\indent Two non-functional benchmark models are implemented to forecast hourly spot prices. The first one is a parameter-rich LASSO estimated regression, originally suggested in the context of energy market forecasts by \textcite{Uniejewski.2016}. The idea is to run 24 regressions for all hourly price time series on all hours of lagged-dependent and lagged-exogenous regressors. In our specific set-up, we use the same variables as in (\ref{eq:elec.FAFR}) but treat them as multivariate time series. That leaves the LASSO model with $220$ regressors plus intercept. In theory, the regularization term deals with overfitting problems in parameter-rich models. We tune the regularization parameter for computational feasibility 4 times using 10-fold cross-validation and estimate the model with the glmnet approach by \textcite{Friedman.2010}. The LASSO generally performs well at prediction tasks as it can process high-dimensional data efficiently compared to, for example, least squares estimated regressions. However, if we are not willing to put further sparsity assumptions on the data generating process, the estimator does not necessarily possess oracle properties and the coefficients are hardly interpretable (\cites{Fu.2000}{Zou.2006}). The second benchmark model is fundamentally different to the LASSO as it is based on an OLS estimated parsimonious autoregressive structure without exogenous regressors. This so-called expert model introduced by \textcite{Ziel.2018} combines lagged (one day, two days, one week) prices for the same hour, last day's minimum and maximum price, last day's very last hourly price and all weekday dummy variables. We run this regression 24 times for all hours of the day. In case a regressor appears more than one time as it is always the case for the last hour model, the duplicate variables are removed. Estimating the expert model is computationally inexpensive and the regression coefficients are clearly interpretable. Lastly, we add a third naive forecasting method which works by always taking last weeks price at the same hour as the prediction.\\
\begin{table}[t!]
	\caption{Comparison between the one day ahead rolling window electricity price forecasts of four different models (FFR, LASSO, Expert and Naive) on three different markets (DE, NP and PJM). rMAE, MAE and RMSE are measures for prediction accuracy.}
	\label{ta:epf.accuracy}
	\centering
	\begin{tabular}{ll|cccc}
		\toprule
		\midrule
		& & \textbf{FFR} & \textbf{LASSO} & \textbf{Expert} & \textbf{Naive} \\
		\midrule
		\multirow{3}{*}{DE} & rMAE & 0.47 & 0.45 & 0.69 & 1.00 \\
		& MAE & 4.20 & 4.04 & 6.17 & 8.88 \\
		& RMSE & 6.46 & 6.36 & 10.11 & 14.55 \\
		\midrule
		\multirow{3}{*}{NP} & rMAE & 0.55 & 0.48 & 0.68 & 1.00 \\
		& MAE & 2.26 & 1.97 & 2.77 & 4.08 \\
		& RMSE & 4.30 & 3.91 & 5.09 & 7.02 \\
		\midrule
		\multirow{3}{*}{PJM} & rMAE & 0.71 & 0.67 & 0.72 & 1.00 \\
		& MAE & 4.54 & 4.30 & 4.62 & 6.38 \\
		& RMSE & 6.96 & 6.89 & 7.20 & 10.76 \\
		\bottomrule
	\end{tabular}
\end{table}
\indent Regarding performance measures, we report the mean absolute error (MAE), the MAE relative to the naive MAE (rMAE) and the root mean squared error (RMSE). Table \ref{ta:epf.accuracy} shows the rolling window forecast results for the three markets. The expert and the naive model yield reasonable predictions considering their simplicity, yet they underperform compared to our functional factor regression and the LASSO for all measures. The direct comparison between FFR and LASSO highlights that our model is roughly on the same level as the well established machine learning approach in terms of forecasting accuracy.

%% file: Conclusion.tex
\section{Conclusion}

This paper introduces the functional factor regression, a novel approach for modeling time-dependent curve data that addresses limitations in existing function-on-function regression frameworks. By imposing factor structures on the regressors and assuming that the predictive information in functional variables is confined to a finite-dimensional subspace, our method is able to provide both the flexibility of infinite-dimensional functional models and the statistical inferential capabilities of traditional multiple regression.\\
\indent A key challenge in conventional unrestricted functional regression is the infinite-rank regression operator, which introduces an ill-posed problem requiring nonparametric estimation. Our model addresses this by ensuring a well-posed framework, allowing for least squares estimation with parametric convergence rates. The innovation lies in assuming that a finite number of factors can effectively represent the functional regression relationship, thereby simplifying the regression operator to finite-rank. Parameters are estimated using the eigencomponents of the integral operator $D$, which is the product of the cross-covariance operator with its adjoint, and the required number of factors is determined through an eigenvalue-difference test. Hence, we are able to develop a novel central limit theorem for the regression parameters in a fully functional model, enabling the construction of valid confidence bands and hypothesis tests that account for the uncertainty in estimated factors and loadings.\\
\indent Our empirical application to electricity price curve modeling demonstrates the practical value of our approach. The estimated bivariate slope coefficient functions reveal economically interpretable patterns while pointwise hypotheses tests make our claims statistically verifiable. Firstly, we find evidence for the existence of a pronounced end-of-day effect in lagged electricity prices. Secondly, the functional model reveals how load forecasts play only a minor role in explaining energy prices when conditioning on further regressors. The influence of wind and solar energy generation forecasts on prices, on the other hand, appears to be substantial and statistically significant. In terms of out-of-sample forecasting performance, our functional factor regression is competitive with machine learning models like LASSO.\\
\indent Future research could extend this framework by addressing further econometric issues, such as endogeneity problems within the functional variables, or by exploring regression models with more complex data structures, such as volatility surfaces or climate data in the form of geographic surfaces. The accompanying \textsf{R} package "ffr" makes our methodology accessible to practitioners and researchers across various disciplines where functional data analysis plays an increasingly important role.

\section*{Supporting Information}

An accompanying R package is available from \url{https://github.com/luiswn/ffr}.\\
A website serving as an online appendix to the application is available from \url{https://luiswn.github.io/ffr-visualizations/}.

%% file: Appendix.tex
\appendix
\section{Appendix}\label{App}
\renewcommand{\thefigure}{A\arabic{figure}}
\setcounter{figure}{0}

\subsection{Proof of Lemma \ref{lem:ident}}
\begin{proof}
\textbf{Part (a):} By inserting the factor representation $X_{jt}(s) = \sum_{l=1}^{K_j} f_{ljt} \psi_{lj}(s) + \varepsilon_{jt}(s)$ into the regression equation we obtain
\begin{align} \label{eq:lem1aux}
	Y_t(r) =  w_t'\alpha(r) + \sum_{\substack{k=1 \\ k \neq j}}^J \langle \beta_k(r,\cdot), X_{kt} \rangle + \sum_{l=1}^{K_j} f_{ljt} \langle \beta_j(r,\cdot), \psi_{lj} \rangle + \langle \beta_j(r,\cdot), \varepsilon_{jt} \rangle +  u_t(r).
\end{align}
Multiplying \eqref{eq:lem1aux} by 
$\langle\beta_j(r,\cdot),\varepsilon_{jt}\rangle$,
taking expectations, and using
$$
E[\varepsilon_{jt}(s)w_t]=0, \quad E[\varepsilon_{jt}(s)\langle\beta_k(r,\cdot),X_{kt}\rangle]=0, \quad E[\varepsilon_{jt}(s)f_{ljt}]=0, \quad E[\varepsilon_{jt}(s)u_t(r)]=0,
$$
we get
$$
E[Y_t(r) \langle\beta_j(r,\cdot),\varepsilon_{jt}\rangle]
=E[\langle\beta_j(r,\cdot),\varepsilon_{jt}\rangle^2].
$$
The non-predictiveness condition $E[Y(r) \varepsilon_{jt}(s)] = 0$ implies $E[Y_t(r) \langle \beta_j(r, \cdot), \varepsilon_{jt} \rangle] = 0$. Hence,
$$
E[\langle \beta_j(r,\cdot), \varepsilon_{jt} \rangle^2] = 0,
$$
and the result follows because a random variable whose second moment is zero is zero almost surely. \\
\noindent
\textbf{Part (b):}
Because $\psi_\perp \in Im(\Gamma_j)$ and $Im(\Gamma_j) = \{ g \in H: g = \Gamma_j(h) \ \text{for some} \ h \in H \}$, there exists $h \in H$ with $\psi_\perp = \Gamma_j(h)$.
The kernel function of $\Gamma_j$ is
$$
Cov(X_{jt}(r), X_{jt}(s)) = \sum_{l,m=1}^{K_j} E[f_{ljt}f_{mjt}] \psi_{lj}(r) \psi_{mj}(s) + E[\varepsilon_{jt}(r) \varepsilon_{jt}(s)]
$$
because $E[\varepsilon_{jt}(s)f_{ljt}] = 0$.
Therefore,
$$
\psi_\perp(r) = \Gamma_j(h)(r) = \sum_{l,m=1}^{K_j} E[f_{ljt}f_{mjt}] \psi_{lj}(r) \langle \psi_{mj}, h \rangle + E[\varepsilon_{jt}(r) \langle \varepsilon_{jt}, h \rangle] = E[\varepsilon_{jt}(r) \langle \varepsilon_{jt}, h \rangle],
$$
where the last step holds by the condition $\langle \psi_\perp, \psi_{lj} \rangle = 0$ for all $l=1, \ldots, K_j$.
Therefore, by the Cauchy-Schwarz inequality and part (a),
\begin{align*}
	\bigg| \int_a^b \beta_j(r,s) \psi_\perp(s) \,ds \bigg|^2 &= E[\langle \beta_j(r,\cdot), \varepsilon_{jt} \rangle \langle \varepsilon_{jt}, h \rangle]^2 \\
	&\leq E[\langle \beta_j(r,\cdot), \varepsilon_{jt} \rangle^2] E[\langle \varepsilon_{jt}, h \rangle^2] \\
	&= 0
\end{align*}
because $E\langle\varepsilon_{jt},h\rangle^2<\infty$.
The result follows by taking the square root.
\\
\noindent
\textbf{Part (c):} Let $\{w_{mj}\}$ be an orthonormal basis of the closed subspace 
$$\overline{Im(\Gamma_j)} \cap span\{\psi_{1j}, \ldots, \psi_{K_j j}\}^\perp.$$
Since $\beta_j(r,\cdot) \in Im(\Gamma_j)$ by the identifiability condition, we can decompose $\beta_j(r,\cdot)$ into two orthogonal components:
\begin{align} \label{eq:betadecomp}
	\beta_j(r,s) = \sum_{l=1}^{K_j} \beta_{lj}(r) \psi_{lj}(s) + \sum_{m=1}^\infty c_{mj}(r) w_{mj}(s),
\end{align}
where $\beta_{lj}(r) = \int_a^b \beta_j(r,s) \psi_{lj}(s) \,ds$ and $c_{mj}(r) = \int_a^b \beta_j(r,s) w_{mj}(s) \,ds$.
Part (b) with $\psi_\perp = w_{mj}$ implies $c_{mj}(r) = 0$ for every $m$, so the remainder term in \eqref{eq:betadecomp} vanishes and the result follows.
\end{proof}

\subsection{Proof of Theorem \ref{theo:primitves}}\label{app:Theorem1}
All elements in Theorem \ref{theo:primitves} follow by the sum of Lemma \ref{lem:mu}--\ref{lem:eigenf}.

\begin{lemma}\label{lem:mu}
	Let $V_t$ be a $L^2([a,b])$-valued random process with $\mu(r) = E[V_t(r)]$. Furthermore, $\sup_t E[(V_t(r))^\nu] < \infty$ for some $\nu > 2$ such that $V_t$ is $\alpha$-mixing of size $-\nu/(\nu-2)$. Let $\widehat \mu(r) = \frac{1}{T} \sum_{t=1}^T V_t(r)$. Then, 
	$$\|\widehat \mu - \mu\| = \mathcal O_P(T^{-1/2}).$$
\end{lemma}

\begin{proof}
	Note that $V_t$ is $\alpha$-mixing of size $-\nu/(\nu-2)$ if the $\alpha$-mixing sequence $\alpha_h$ of $V_t$ satisfies $\sum_{h=1}^\infty \alpha_h^{(\nu-2)/\nu} < \infty$ (see, e.g., \cite{Davidson.2021}, Section 15.1).
	Functions of $\alpha$-mixing processes are $\alpha$-mixing of the same size, which implies that $V_t(r)$ is $\alpha$-mixing for any $r \in [a,b]$.
	The mixing inequality for $\alpha$-mixing processes (see, e.g., Corollary 15.3 in \cite{Davidson.2021}) implies
	\begin{equation}\label{eq:mix}
		|Cov(V_t(r), V_{t-h}(r))| \leq 6 \alpha_h^{(p-2)/p} \big(E[(V_t(r))^p]E[(V_{t-h}(r))^p]\big)^{1/p}
	\end{equation}
	for any $p \geq 2$.
	We can decompose
	\begin{align*}
		(\widehat \mu(r) - \mu(r))^2
		&= \bigg( \frac{1}{T} \sum_{t=1}^T V_t(r) - \mu(r) \bigg)^2 \\
		&= \frac{1}{T^2} \sum_{t=1}^T (V_t(r) - \mu(r))^2 + \frac{2}{T^2} \sum_{t=2}^T \sum_{h=1}^{t-1} (V_t(r) - \mu(r))(V_{t-h}(r) - \mu(r))
	\end{align*}
	and obtain
	\begin{align*}
		E[\|\widehat \mu - \mu\|^2] = \frac{1}{T^2} \sum_{t=1}^T \int_a^b  Var[V_t(r)] \dd r + \frac{2}{T^2} \sum_{t=2}^T \sum_{h=1}^{t-1} \int_a^b Cov(V_t(r), V_{t-h}(r)) \dd r.
	\end{align*}
	The first term is $\mathcal O(T^{-1})$ since $\nu \geq 2$, and, with the mixing inequality with $p = \nu$, the second term satisfies
	\begin{align*}
		&\bigg|\frac{2}{T^2} \sum_{t=2}^T \sum_{h=1}^{t-1} \int_a^b Cov(V_t(r), V_{t-h}(r)) \dd r \bigg| \\
		&\leq \frac{12}{T^2} \sum_{t=2}^T \sum_{h=1}^{t-1} \alpha_h^{(\nu-2)/\nu} \int_a^b \big(E[(V_t(r))^\nu]E[(V_{t-h}(r))^\nu]\big)^{1/\nu} \dd r,
	\end{align*}
	which is $\mathcal O(T^{-1})$ as well since $\sum_{h=1}^\infty \alpha_h^{(\nu-2)/\nu} < \infty$ and the fact that the $\nu$-th moments are bounded. The assertion follows by Markov's inequality.
\end{proof}

\begin{lemma}\label{lem:cov}
Let $V_t$ and $P_t$ be $L^2([a,b])$-valued random processes with $\mu_V(r) = E[V_t(r)]$, $\mu_P(r) = E[P_t(r)]$, and $\sigma(r,s) = Cov(V_t(r), P_t(s))$.  Furthermore, $\sup_t E[(V_t(r))^\nu] < \infty$ and $\sup_t E[(P_t(r))^\nu] < \infty$ for some $\nu \geq 4$ such that the joint process $(V_t, P_t)$ is $\alpha$-mixing of size $-\nu/(\nu-2)$. Let $\widehat \mu_V(r) = \frac{1}{T} \sum_{t=1}^T V_t(r)$, $\widehat \mu_P(r) = \frac{1}{T} \sum_{t=1}^T P_t(r)$,
and $\widehat \sigma(r,s) = \frac{1}{T} \sum_{t=1}^T (V_t(r) - \widehat \mu_V(r)) (P_t(s) - \widehat \mu_P(s))$. Let $S$ and $\widehat S$ be the integral operators with kernel functions $\sigma(r,s)$ and $\widehat \sigma(r,s)$, respectively. Then, 
$$\|\widehat S - S\|_\mathcal S =  \mathcal O_P(T^{-1/2}).$$
\end{lemma}

\begin{proof}
	Let $\widetilde S$ be the integral operator with kernel function $\widetilde \sigma(r,s) := \frac{1}{T} \sum_{t=1}^T (V_t(r) - \mu_V(r))(P_t(s) - \mu_P(s))$. 
	We can decompose
	$$
	\widehat \sigma(r,s) = \widetilde \sigma(r,s) - (\mu_V(r) - \widehat \mu_V(r))(\mu_P(s) - \widehat \mu_V(s)),
	$$
	which implies that $\|\widehat S - \widetilde S\|_\mathcal S = \|\mu_V - \widehat \mu_V \| \|\mu_P - \widehat \mu_P \| = \mathcal O_P(T^{-1})$ by the Lemma \ref{lem:mu}.
	By the triangle inequality, it remains to show that $\|\widetilde S - S\|_\mathcal S = \mathcal O_P(T^{-1/2})$.
	Let us define the centered random variables $V_t^*(r) := V_t(r) - \mu_V(r)$ and $P_t^*(r) := P_t(r) - \mu_P(r)$. We have
	\begin{align*}
		&(\widetilde \sigma(r,s) - \sigma(r,s))^2 = \bigg( \frac{1}{T} \sum_{t=1}^T V_t^*(r)P_t^*(s) - E[V_t^*(r) P_t^*(s)] \bigg)^2 \\
		&= \frac{1}{T^2} \sum_{t=1}^T (V_t^*(r) P_t^*(s) - E[V_t^*(r) P_t^*(s)])^2 \\
		& \quad
		+ \frac{2}{T^2} \sum_{t=2}^T \sum_{h=1}^{t-1} (V_t^*(r) P_t^*(s) - E[V_t^*(r) P_t^*(s)])(V_{t-h}^*(r) P_{t-h}^*(s) - E[V_{t-h}^*(r) P_{t-h}^*(s)]),
	\end{align*}
	and
	\begin{align*}
		E[\|\widetilde S - S \|_{\mathcal S}^2] &= \frac{1}{T^2} \sum_{t=1}^T \int_a^b \int_a^b Var[V_t^*(r) P_t^*(s)] \dd r \dd s \\
		& \quad 
		+ \frac{2}{T^2} \sum_{t=2}^T \sum_{h=1}^{t-1} \int_a^b \int_a^b Cov(V_t^*(r) P_t^*(s), V_{t-h}^*(r) P_{t-h}^*(s)) \dd r \dd s.
	\end{align*}
	The first term is $\mathcal O(T^{-1})$ since $Var[V_t^*(r) P_t^*(s)] < \infty$. This follows from the Cauchy Schwarz inequality and the fact that we have bounded fourth moments. For the second term, we apply the mixing inequality from equation (\ref{eq:mix}) with $p = \nu$, and obtain
	\begin{align*}
		&\bigg|\frac{2}{T^2} \sum_{t=2}^T \sum_{h=1}^{t-1} \int_a^b \int_a^b Cov(V_t^*(r) P_t^*(s), V_{t-h}^*(r) P_{t-h}^*(s)) \dd r \dd s \bigg| \\
		&\leq \frac{2}{T^2} \sum_{t=2}^T \sum_{h=1}^{t-1} \alpha_h^{(\nu-2)/\nu} \int_a^b \int_a^b \big(E[(V_t^*(r) P_t^*(s))^{\nu}]E[(V_{t-h}^*(r) P_{t-h}^*(s))^{\nu}]\big)^{1/\nu} \dd r \dd s.
	\end{align*}
	Notice that $\int_a^b \int_a^b \big(E[(V_t^*(r) P_t^*(s))^{\nu}]E[(V_{t-h}^*(r) P_{t-h}^*(s))^{\nu}]\big)^{1/\nu} \dd r \dd s < \infty$ by the Cauchy Schwarz inequality and the fact that $\nu$-th moments are bounded, and $\sum_{h=1}^{\infty} \alpha_h^{(\nu-2)/\nu} < \infty$ by the $\alpha$-mixing property, which implies that the second term of $E[\|\widetilde S - S \|_{\mathcal S}^2]$ is $\mathcal O(T^{-1})$ as well.
	The assertion follows by Markov's inequality.
\end{proof}

\begin{lemma}\label{lem:cross-cov}
	Let $Y_t$ and $X_{jt}$ be defined as in equation (\ref{eq:basemodel}) and (\ref{eq:factor.structure}). $C_j$ is the auxiliary cross covariance integral operator with kernel function $c_j\left(r,s\right)  \coloneqq \mathit{Cov}(X_{jt}( r),Y_t( s) ).$ The respective sample counterparts $\widehat{C}_j$ and $\widehat{c}_j\left(r,s\right)$  are defined as in equation \eqref{eq:sample.crosscov}. By Assumption \ref{as:estimation}, for any $j=1,\ldots,J$,
	\begin{equation*}
		\| \hat{C}_j - C_j \|_{\mathcal{S}} = \mathcal{O}_P(T^{-1/2})
	\end{equation*}
\end{lemma}

\begin{proof}
	Because functions of $\alpha$-mixing processes are $\alpha$-mixing of the same size, $Y_t(r)$ is $\alpha$-mixing for any $r \in [a,b]$ by Assumption \ref{as:estimation}. Hence, the assertion follows by Lemma \ref{lem:cov}.
\end{proof}

\begin{lemma}\label{lem:D}
	 Let $D_j$ and $\widehat{D}_j$ be the respective integral operators of the kernels $d_j$ and $\widehat{d}_j$ as defined in equations (\ref{eq:d}) and (\ref{eq:sample.d}). By Assumptions \ref{as:exo}--\ref{as:estimation}, for any $j=1,\ldots,J$,
	\begin{equation*}
		\| \widehat{D}_j - D_j \|_{\mathcal{S}} = \mathcal{O}_P(T^{-1/2})
	\end{equation*}
\end{lemma}

\begin{proof}
	Consider the decomposition
	\begin{equation*}
		\begin{split}
			&\int_{a}^{b} \widehat{c}_j\left(r,q\right) \widehat{c}_j\left(q,s\right) - c_j\left(r,q\right) c_j\left(q,s\right)\,dq \\
			& = \int_{a}^{b}\left(\widehat{c}_j\left(r,q\right)- c_j\left(r,q\right)\right) \widehat{c}_j\left(q,s\right) + \left(\widehat{c}_j\left(s,q\right)- c_j\left(s,q\right)\right) c_j\left(q,r\right)\,dq \\
			& = a_3\left(r,s\right) + a_4\left(r,s\right) + a_5\left(r,s\right),
		\end{split}	
	\end{equation*}
	where
	\begin{equation*}
		\begin{split}
			& a_3\left(r,s\right) = \int_{a}^{b}\left(\widehat{c}_j\left(r,q\right)- c_j\left(r,q\right)\right)\left(\widehat{c}_j\left(s,q\right)- c_j\left(s,q\right)\right)\,dq,\\
			& a_4\left(r,s\right) = \int_{a}^{b}\left(\widehat{c}_j\left(r,q\right)- c_j\left(r,q\right)\right) c_j\left(q,s\right)\,dq,\\
			& a_5\left(r,s\right) = \int_{a}^{b}\left(\widehat{c}_j\left(s,q\right)- c_j\left(s,q\right)\right) c_j\left(q,r\right)\,dq.
		\end{split}	
	\end{equation*}
	The triangle inequality states that $\| \widehat{D} - D \|_{\mathcal{S}} \leq \left\| A_3 \right\|_{\mathcal{S}} + \left\| A_4 \right\|_{\mathcal{S}} + \left\| A_5 \right\|_{\mathcal{S}}.$ Starting with the first term,
	\begin{equation*}
		\begin{split}
			\left\| A_3 \right\|_{\mathcal{S}}^2 & = \int_{a}^{b} \int_{a}^{b} a_3\left(r,s\right)^2 \,dr \,ds\\
			& = \int_{a}^{b} \int_{a}^{b}\left(\int_{a}^{b}\left(\widehat{c}_j\left(r,q\right)- c_j\left(r,q\right)\right)\left(\widehat{c}_j\left(s,q\right)- c_j\left(s,q\right)\right)\,dq\right)^2 \,dr \,ds\\
			& \leq \int_{a}^{b} \int_{a}^{b}\left(\int_{a}^{b}\left(\widehat{c}_j\left(r,q\right)- c_j\left(r,q\right)\right)^2 \,dq \int_{a}^{b} \left(\widehat{c}_j\left(s,q\right)- c_j\left(s,q\right)\right)^2\,dq\right) \,dr \,ds\\
			& = \int_{a}^{b} \int_{a}^{b}\left(\widehat{c}_j\left(r,q\right)- c_j\left(r,q\right)\right)^2 \,dq \,dr \int_{a}^{b} \int_{a}^{b}\left(\widehat{c}_j\left(s,q\right)- c_j\left(s,q\right)\right)^2 \,dq \,ds\\
			& = \| \widehat{C}_j - C_j \|_{\mathcal{S}}^2 \| \widehat{C}_j - C_j \|_{\mathcal{S}}^2 = \mathcal{O}_P\left(T^{-2}\right),
		\end{split}	
	\end{equation*}
	where the last equality follows from Lemma \ref{lem:cross-cov}. Moving on to the second term,
	\begin{equation*}
		\begin{split}
			\left\| A_4 \right\|_{\mathcal{S}}^2 & = \int_{a}^{b} \int_{a}^{b} a_4\left(r,s\right)^2 \,dr \,ds\\
			& = \int_{a}^{b} \int_{a}^{b}\left(\int_{a}^{b}\left(\widehat{c}_j\left(r,q\right)- c_j\left(r,q\right)\right) c_j\left(q,s\right)\,dq\right)^2 \,dr \,ds\\
			& \leq \int_{a}^{b} \int_{a}^{b}\left(\int_{a}^{b}\left(\widehat{c}_j\left(r,q\right)- c_j\left(r,q\right)\right)^2 \,dq \int_{a}^{b} c_j\left(q,s\right)^2\,dq\right) \,dr \,ds\\
			& = \int_{a}^{b} \int_{a}^{b}\left(\widehat{c}_j\left(r,q\right)- c_j\left(r,q\right)\right)^2 \,dq \,dr \int_{a}^{b} \int_{a}^{b} c_j\left(q,s\right)^2 \,dq \,ds\\
			& = \| \widehat{C}_j - C_j \|_{\mathcal{S}}^2 \left\| C_j \right\|_{\mathcal{S}}^2 = \mathcal{O}_P\left(T^{-1}\right).
		\end{split}	
	\end{equation*}
	Similar to the term $\left\| A_1 \right\|_{\mathcal{S}}^2$, the auxiliary covariance function is continuous and defined on a closed bounded interval. Therefore, the boundedness theorem states that there exists a positive real constant such that $\sup_{s,q \in \left[ 0,1\right]} \left| c_j\left(s,q\right)\right| = K < \infty$. In combination with Lemma \ref{lem:cross-cov}, the last equality holds.\\
	Because $\left\| A_4 \right\|_{\mathcal{S}}^2 = \left\| A_5 \right\|_{\mathcal{S}}^2$, the proof follows analogously and we conclude that $\| \widehat{D} - D \|_{\mathcal{S}}^2 = \mathcal{O}_P\left(T^{-1}\right).$
\end{proof}

\begin{lemma}\label{lem:eigenv}
	The eigenvalues $\lambda_{lj}$ of the operator $D_j$ can be estimated by their empirical counterparts $\widehat{\lambda}_{lj}$ of $\widehat{D}_j$. For each $l \geq 1$ and any $j=1,\ldots,J$,
	\begin{equation*}
		\left|\widehat{\lambda}_{lj} - \lambda_{lj} \right| = \mathcal{O}_P\left(T^{-1/2}\right).
	\end{equation*}
\end{lemma}

\begin{proof}
	Lemma 2.2 in \textcite{Horvath.2012} implies that
	\begin{equation*}
		\left|\widehat{\lambda}_{lj} - \lambda_{lj} \right| \leq \left\| \widehat{D}_j - D_j \right\|_{\mathcal{S}},
	\end{equation*}
	for each $l \geq 1$. The assertion then results as a consequence of Lemma \ref{lem:D}.
\end{proof}

\begin{lemma}\label{lem:eigenf}
	The eigenfunctions $\psi_{lj}\left( s\right)$ of the operator $D_j$ can be estimated by their empirical counterparts $\widehat{\psi}_{lj}\left( s\right) $ of $\widehat{D}_j$. For any $j=1,\ldots,J$,
	\begin{equation*}
		\max_{1 \leq l \leq K_j} \left\| s_{lj} \widehat{\psi}_{lj} - \psi_{lj} \right\|  = \mathcal{O}_P\left(T^{-1/2}\right).
	\end{equation*}
\end{lemma}

\begin{proof}
	Lemma 2.3 in \textcite{Horvath.2012} implies that
	\begin{equation*}
			\max_{1 \leq l \leq K_j} \left\| s_{lj} \widehat{\psi}_{lj} - \psi_{lj} \right\| \leq \frac{2\sqrt{2}}{\alpha}\left\| \widehat{D}_j - D_j \right\|_{\mathcal{S}},
	\end{equation*}
	where $\alpha = \min\left\lbrace \lambda_1-\lambda_2, \lambda_2-\lambda_3, \ldots, \lambda_{K_j-1}-\lambda_{K_j}, \lambda_{K_j}\right\rbrace.$ The assertion then results as a consequence of Lemma \ref{lem:D}.
\end{proof}

\subsection{Proof of Theorem \ref{theo:number.factors}}\label{app:Theorem2}

\begin{proof}
	Define $c \coloneqq \sqrt{T^{-1} \sum_{t=1}^T \|X_{jt} - \widehat{\mu}_j \|^2} \sqrt{T^{-1} \sum_{t=1}^T \|Y_t - \overline{Y} \|^2}$, where $ c = \mathcal{O}_P(1)$. Then for $1 \leq l \leq K_j$ and any $j=1,\ldots,J$,
	\begin{equation*}
		\plim_{T \rightarrow \infty}\gamma \ln(T)\frac{\widehat{\lambda}_{lj}}{c} = \plim_{T \rightarrow \infty}\gamma \ln(T)\left(\frac{\lambda_{lj}}{c} + \mathcal{O}_P\left(T^{-1/2}\right)\right) = \infty,
	\end{equation*}
	because $\ln(T)\mathcal{O}_P\left(T^{-1/2}\right)= o_P(1)$ and $\lambda_{lj}/c>0$ by Theorem \ref{theo:primitves}, Assumption \ref{as:identifiability2}(b) and $c > 0$. Then
	\begin{equation*}
		\plim_{T \rightarrow \infty} \arctan(\gamma \ln(T)\frac{\widehat{\lambda}_{lj}}{c}) = \frac{\pi}{2},
	\end{equation*}
	leading to
	\begin{equation*}
		\plim_{T \rightarrow \infty} G(\widehat{\lambda}_{lj}) = \frac{2}{\pi} \frac{\pi}{2} = 1 \quad \text{for} \ 1 \leq l \leq K_j.
	\end{equation*}
	For $K_j+1 \leq l \leq K_{j}^\text{max}$ and any $j=1,\ldots,J$,
	\begin{equation*}
		\plim_{T \rightarrow \infty}\gamma \ln(T)\frac{\widehat{\lambda}_{lj}}{c} = \plim_{T \rightarrow \infty}\gamma \ln(T)\left(\frac{\lambda_{lj}}{c} + \mathcal{O}_P\left(T^{-1/2}\right)\right) = 0,
	\end{equation*}
	because $\ln(T)\mathcal{O}_P\left(T^{-1/2}\right)=o_P(1)$ and $\lambda_{lj}/c=0$ by Theorem \ref{theo:primitves}, Assumption \ref{as:identifiability2}(b) and $c > 0$. Therefore,
	\begin{equation*}
		\plim_{T \rightarrow \infty} G(\widehat{\lambda}_{lj}) = 0 \quad \text{for} \ K_j+1 \leq l \leq K_{j}^\text{max}.
	\end{equation*}
	As defined in Section \ref{subsec:factor.number},
	\begin{equation*}
		g_{lj} \coloneqq
		\begin{cases}
			1, & \text{for $l=0$}\\
			G(\widehat{\lambda}_{lj}), & \text{for $l=1,\ldots,K_j^{\text{max}}$}\\
			0, & \text{for $l=K_j^{\text{max}}+1$}.
		\end{cases}
	\end{equation*}
	Clearly, $\widehat{K}_j^{ED}=\argmax_{0\leq l \leq K_j^{\text{max}}} 	\left\lbrace g_{lj} - g_{(l+1)j}\right\rbrace$ then implies
$$
	\plim_{T\to \infty} \widehat K_j^{ED} = K_j.
$$	
	Since $\widehat K_j^{ED}$ is a discrete random variable it follows that
	\begin{equation*}
		\lim_{T \rightarrow \infty} \text{Pr}\left( \widehat{K}_j^{ED}=K_j\right)=1.
	\end{equation*}
\end{proof}

\subsection{Proof of Theorem \ref{theo:ols.normality}}\label{app:Theorem3}

\begin{proof}
	A challenge for the notation of this proof is that the signs of the loadings $\psi_{lj}(s)$ and the coefficients $\beta_{jt}(r) = \int_a^b \beta_j(r,s) \psi_{lj}(s) \dd s$ are unidentified whereas the signs of their products $\beta_{lj}(r) \psi_{lj}(s)$ and of the true coefficient function $\beta_j(r,s) = \sum_{l=1}^{K_j} \beta_{lj}(r) \psi_{lj}(s)$ are identified.

Therefore, we condition our notation on the selected signs of the estimated loadings and define the sign matrices $S_j = diag(s_{1j}, \ldots, s_{K_j j})$ and $S = diag(I_N, S_1, \ldots, S_J)$ with $s_{lj} = sign(\langle \widehat \psi_{lj}, \psi_{lj} \rangle)$ for $j=1, \ldots, J$ and $l=1, \ldots, K_j$.
This matrix has the property that $S = S'$ and $SS = I_M$, where $M = N + \sum_{j=1}^J K_j$.

To represent the $j$-th coefficient surface and its estimator conveniently, let $P_j \in \mathbb{R}^{K_j \times M}$ be the selection matrix defined by
$$
	(P_j)_{k,i} = \begin{cases} 1 & \text{if} \ i = N + \sum_{l=1}^{j-1} K_l + k, \\
	0 & \text{otherwise} \end{cases}
$$
to extract the relevant entries of $B(r)$ and $\widehat B(r)$. Then, $P_j B(r) = B_j(r)$, $P_j \widehat B(r) = \widehat B_j(r)$, and $P_j \widehat Q^{-1} = [\widehat Q^{-1}]_j$.
The selection matrix has the property $SP_j' = P_j'S_j$.
Then, the coefficient functions are
\begin{align*}
	\beta_j(r,s) &= (B_j(r))'\Psi_j(s) = (S B(r))'P_j'S_j \Psi_j(s), \\
	\widehat{\beta}_{j}(r,s) &= (\widehat B_j(r))'\widehat \Psi_j(s) = (\widehat B(r))'P_j'\widehat \Psi_j(s),
\end{align*}
with
$$
	\widehat{\beta}_{j}(r,s) - \beta_j(r,s) = (\widehat B(r) - S B(r))' P_j' \widehat \Psi_j(s) + (S B(r))' P_j' (\widehat \Psi_j(s) - S_j \Psi_j(s)).
$$
The least squares estimator is defined as
$$
	\widehat B(r) = \widehat Q^{-1} \bigg( \frac{1}{T} \sum_{t=1}^T \widehat z_t Y_t(r) \bigg), \quad \widehat Q = \frac{1}{T} \sum_{t=1}^T \widehat z_t \widehat z_t',
$$
which gives
\begin{align} \label{eq:thm3Atilde}
	\sqrt T (\widehat B(r) - S B(r) ) = \widehat Q^{-1} \widetilde A(r), \quad \widetilde A(r) = \frac{1}{\sqrt T} \sum_{t=1}^T \widehat z_t (Y_t(r) - \widehat z_t' S B(r)).
\end{align}
Then,
\begin{align*}
	\sqrt T ( \widehat{\beta}_{j}(r,s) - \beta_j(r,s) ) 
	&= \sqrt T(\widehat B(r) - S B(r))' P_j' \widehat \Psi_j(s) + \sqrt T(S B(r))' P_j' (\widehat \Psi_j(s) - S_j \Psi_j(s)) \\
	&= (P_j \widehat Q^{-1} \widetilde A(r))' \widehat \Psi_j(s) + \sqrt T(\widehat \Psi_j(s) - S_j \Psi_j(s))' S_j B_j(r).
\end{align*}
Lemma \ref{lem:aux17} and the consistency of $\widehat \Psi_j(s)$ by Theorem \ref{theo:primitves} imply
\begin{align}	\label{eq:thm3aux1}
(P_j \widehat Q^{-1} \widetilde A(r))' \widehat \Psi_j(s)
	= \frac{1}{\sqrt T} \sum_{t=1}^T \Big( (P_j Q^{-1} z_t^* u_t(r))'\Psi_j(s)  + \omega^{(B)}_{jt}(r,s) \Big) + o_P(1),
\end{align}
where $z_t^* = (w_t', F_{1t}^*, \ldots, F_{Jt}^*)'$ with $F_{jt}^* = (f_{1jt}^*, \ldots, f_{K_j j t}^*)'$, $f_{ljt}^* = \langle X_{jt} - \mu_j, \psi_{lj} \rangle$, and $Q = E[z_t^* (z_t^*)']$.
Note that $z_t^* = z_t + e_t$ with $e_t = (\langle \varepsilon_{jt}, \psi_{1j} \rangle, \ldots, \langle \varepsilon_{jt}, \psi_{K_j j} \rangle)'$ and $Q = E[z_t z_t'] + E[e_t e_t']$. Because $E[z_t z_t']$ is positive definite by Assumption \ref{as:identifiability2}(c), $Q$ must be positive definite as well, and $Q^{-1}$ exists.
The additional term in the expression above is
$$
	\omega^{(B)}_{jt}(r,s) = (\Psi_j(s))'P_j Q^{-1} \sum_{k=1}^J \bigg( E[z_t^*]  (F_{kt}^*)' - E[z_t^* (F_{kt}^*)'] G_{kt} \bigg) B_k(r),
$$
where $G_{kt}$ is the $K_j \times K_j$ matrix with the $(l,m)$-entries
\begin{align} \label{eq:thm3G}
	[G_{kt}]_{lm} = (\lambda_{lk} - \lambda_{mk})^{-1}(f_{mkt}^* y_{lkt} - E[f_{mkt}^* y_{lkt}] + f_{lkt}^* y_{mkt} - E[f_{lkt}^* y_{mkt}]) 1_{\{l\neq m\}}.
\end{align}
Lemma \ref{lem:aux22} implies
\begin{align} \label{eq:thm3aux2}
	\sqrt T(\widehat \Psi_j(s) - S_j \Psi_j(s))' S_j B_j(r) = \frac{1}{\sqrt T} \sum_{t=1}^T \omega^{(\Psi)}_{jt}(r,s) + o_P(1),
\end{align}
where
\begin{align*}
	\omega^{(\Psi)}_{jt}(r,s) = (\Psi_j(s))' G_{jt} B_j(r) + \varepsilon_{jt}^*(s) h_{jt}'B_j(r)
\end{align*}
with $\varepsilon_{jt}^*(s) = X_{jt}(s) - \mu_j(s) - (F_{jt}^*)'\Psi_j(s)$, and $h_{jt}$ is the vector of length $K_j$ with entries $h_{ljt} = \lambda_{lj}^{-1}\langle \gamma_{lj}, Y_t - E[Y_t]\rangle$, where $\gamma_{lj}(r) = E[f_{ljt}^*(Y_t(r) - E[Y_t(r)])]$.
Therefore,
\begin{align*}
	\sqrt T ( \widehat{\beta}_{j}(r,s) - \beta_j(r,s) ) 
	&=  \frac{1}{\sqrt T} \sum_{t=1}^T \Big( (z_t^* u_t(r))'\Psi_j(s)  + \omega^{(B)}_{jt}(r,s) + \omega^{(\Psi)}_{jt}(r,s)  \Big) + o_P(1).
\end{align*}
Note that $(z_t^* u_t(r))'\Psi_j(s)  + \omega^{(B)}_{jt}(r,s) + \omega^{(\Psi)}_{jt}(r,s)$ is a zero mean and alpha mixing sequence of size $-\nu/(\nu-2)$.
Therefore, by the central limit theorem for alpha-mixing sequences (e.g., Theorem 5.20 in \cite{White.2001}), we have
\begin{align} \label{eq:thm3aux3}
	\sqrt T ( \widehat{\beta}_{j}(r,s) - \beta_j(r,s) ) \overset{d}{\rightarrow} \mathcal N(0,\Omega_j(r,s))
\end{align}
for any fixed $r,s \in [a,b]$, 
where $\Omega_j(r,s) = Var[(z_t^* u_t(r))'\Psi_j(s)  + \omega^{(B)}_{jt}(r,s) + \omega^{(\Psi)}_{jt}(r,s)]$.
To show that $\widehat \Omega_j(r,s)$ is a consistent estimator for $\Omega_j(r,s)$, we apply Lemmas \ref{lem:aux18} and \ref{lem:aux23}, which give
\begin{align*}
	&(P_j \widehat Q^{-1} \widetilde A(r))' \widehat \Psi_j(s) + \sqrt T(\widehat \Psi_j(s) - S_j \Psi_j(s))' S_j B_j(r) \\
	&= \frac{1}{\sqrt T} \sum_{t=1}^T \Big( (P_j \widehat Q^{-1} \widehat z_t \widehat u_t(r))'\widehat \Psi_j(s)  + \widehat \omega_{jt}(r,s) \Big)  + o_P(1)
\end{align*}
Therefore, by \eqref{eq:thm3aux1} and \eqref{eq:thm3aux2},
\begin{align*}
	&\frac{1}{\sqrt T}\sum_{t=1}^T \Big( ([\widehat Q^{-1}]_j \widehat z_t \widehat u_t(r))'\widehat \Psi_j(s)  + \widehat \omega_{jt}(r,s) \Big) \\
	&\quad - \frac{1}{\sqrt T}\sum_{t=1}^T \Big( (P_j Q^{-1} z_t^* u_t(r))'\Psi_j(s) + \omega^{(B)}_{jt}(r,s) + \omega^{(\Psi)}_{jt}(r,s) \Big) = o_P(1),
\end{align*}
and
\begin{align*}
	\widehat \Omega_j(r,s) &= \frac{1}{T} \sum_{t=1}^T \Big( ([\widehat Q^{-1}]_j \widehat z_t \widehat u_t(r))'\widehat \Psi_j(s)  + \widehat \omega_{jt}(r,s) \Big)^2 \\
	&\overset{p}{\to} E\Big[( (z_t^* u_t(r))'\Psi_j(s)  + \omega^{(B)}_{jt}(r,s) + \omega^{(\Psi)}_{jt}(r,s)  )^2\Big] = \Omega_j(r,s)
\end{align*}
by the law of large numbers for alpha-mixing processes (e.g., Theorem 3.47 in \cite{White.2001}). Therefore, by \eqref{eq:thm3aux3} and Slutsky's theorem,
\begin{align} \label{eq:thm3aux4}
	\frac{\sqrt T ( \widehat{\beta}_{j}(r,s) - \beta_j(r,s) )}{\sqrt{\widehat \Omega_j(r,s)}} \overset{d}{\rightarrow} \mathcal N(0,1) \quad \text{pointwise for} \ r,s \in [a,b].
\end{align}
To show that \eqref{eq:thm3aux4} also holds uniformly in $r,s \in [a,b]$ in the space of continuous functions $\mathcal C[a,b] \times \mathcal C[a,b]$ equipped with the sup-norm, we employ Theorem 7.5 in \textcite{Billingsley.1999}, which requires that two conditions must be satisfied.
First, finite-dimensional convergence of \eqref{eq:thm3aux4} holds because, by the Cramer-Wold device and the central limit theorem for alpha-mixing sequences, the vector
$$
\Bigg(\frac{\sqrt T ( \widehat{\beta}_{j}(r_1,s_1) - \beta_j(r_1,s_1) )}{\sqrt{\widehat \Omega_j(r_1,s_1)}}, \ldots, \frac{\sqrt T ( \widehat{\beta}_{j}(r_k,s_k) - \beta_j(r_k,s_k) )}{\sqrt{\widehat \Omega_j(r_k,s_k)}}\Bigg)
$$
converges to a Gaussian vector for any $r_1, \ldots, r_k \in [a,b]$ and $s_1, \ldots, s_k \in [a,b]$.
The second condition is asymptotic uniform stochastic equicontinuity of the left-hand side of \eqref{eq:thm3aux4}.
This follows from Assumption \ref{as:estimation}(b).
Since $Y_t(r)$ and $X_{jt}(s)$ have differentiable sample paths, and $B(r)$ and $\Psi_j(s)$ are differentiable, 
$$
	\frac{\sqrt T ( \widehat{\beta}_{j}(r,s) - \beta_j(r,s) )}{\sqrt{\widehat \Omega_j(r,s)}}
$$
also has differentiable and hence Lipschitz-continuous sample paths, and the sufficient condition for asymptotic uniform stochastic equicontinuity stated in Theorem 22.10 in \textcite{Davidson.2021} is fulfilled.
Hence \eqref{eq:thm3aux4} holds uniformly in the space $\mathcal C[a,b] \times \mathcal C[a,b]$.
\end{proof}

\subsubsection{Auxiliary Lemmas for the proof of Theorem \ref{theo:ols.normality}}

In this subsection we present some auxiliary lemmas required for the proof of Theorem \ref{theo:ols.normality}.
Throughout the proofs of these lemmas, we will use the following definitions and useful facts:
\begin{enumerate}
\item For all real numbers $c_1, c_2$ we have $(c_1 + c_2)^2 \leq 2(c_1^2 + c_2^2)$
	\item The Euclidean vector norm for a $n \times 1$ vector $a$ is $\|a\|_2 = (\sum_{i=1}^n a_i^2)$. The Frobenius matrix norm for a $n \times k$ matrix $A$ is $\|A\|_F = (\sum_{i=1}^n \sum_{j=1}^k A_{ij}^2)^{1/2}$. We have $\|a b'\|_F \leq \|a\|_2 \|b\|_2$, and, by the triangle and Cauchy-Schwarz inequalities,
	$$
	\Big\| \frac{1}{T} \sum_{t=1}^T a_t b_t' \Big\|_F 
	\leq  \frac{1}{T} \sum_{t=1}^T \|a_t\|_2 \|b_t\|_2 
	\leq \Big( \frac{1}{T} \sum_{t=1}^T \|a_t\|_2^2 \Big)^{1/2} \Big( \frac{1}{T} \sum_{t=1}^T \|b_t\|_2^2 \Big)^{1/2}.
	$$
	Furthermore, $\|Aa\|_2 \leq \|A\|_F \|a\|_2$ and $|a'b| \leq \|a\|_2 \|b\|_2$.
	\item The Hilbert-Schmidt norm and the $L^2$-norm are compatible, i.e., for an operator $\mathcal T$ and a function $x$ we have $\|\mathcal T x\| \leq \|\mathcal T\|_{\mathcal S} \|x\|$.
	Furthermore, $\langle \mathcal T(x), y \rangle = \langle x, \mathcal T^*(y) \rangle$ for all $x,y \in H$, where $\mathcal T^*$ is the adjoint operator.
	\item Define $z_t^* = (w_t', (F_{1t}^*)', \ldots, (F_{Jt}^*)')'$ with $F_{jt}^* = (f_{1jt}^*, \ldots, f_{K_j j t}^*)'$ and $f_{ljt}^* = \langle X_{jt} - \mu_j, \psi_{lj} \rangle$
	Define the sign matrices $S_j = diag(s_{1j}, \ldots, s_{K_j j})$ for $j=1, \ldots, J$ with $s_{lj} = sign(\langle \widehat \psi_{lj}, \psi_{lj} \rangle)$ for $l=1, \ldots, K_j$, and  $S = diag(I_N, S_1, \ldots, S_J)$ with the property $S = S' = S^{-1}$.
	Furthermore, define $\varepsilon_{jt}^*(s) = X_{jt}(s) - \mu_j(s) - (F_{jt}^*)'\Psi_j(s)$, $Q = E[z_t^*(z_t^*)']$, $Q^* = \frac{1}{T} \sum_{t=1}^T z_t^*(z_t^*)'$, $y_{ljt} = \langle Y_{t} - E[Y_t], \psi_{lj} \rangle$, $g_{lmjt} = s_{mj} s_{lj}( (f_{mjt}^* y_{ljt} - E[f_{mjt}^* y_{ljt}]) + (f_{ljt}^* y_{mjt} - E[f_{ljt}^* y_{mjt}]))$, $\widehat g_{lmjt} = \widehat f_{mjt} \widehat y_{ljt} - \overline{ f_{mj} y_{lj}} + \widehat f_{ljt} \widehat y_{mjt} - \overline{f_{lj} y_{mj}}$, and $\gamma_{lj}(r) = E[f_{ljt}^* (Y_t(r) - E[Y_t(r)])]$.
	\item We will frequently use the facts that all random variables have bounded forth moments by Assumption \ref{as:estimation}(a) and that the primitives are $\sqrt T$-consistent by Theorem \ref{theo:primitves}. Further, we frequently apply the Cauchy-Schwarz and triangle inequalities.
	\item We will frequently use the law of large numbers for alpha-mixing sequences (Theorem 3.47 in \cite{White.2001}) and the central limit theorem for alpha mixing sequences (Theorem 5.20 in \cite{White.2001}).
	\item All results in this subsection follow under Assumptions \ref{as:exo}--\ref{as:estimation}, for any $j=1, \ldots, J$, $l=1, \ldots, K_j$, $t \in \mathbb Z$, $r,s \in [a,b]$, as $T \to \infty$.
\end{enumerate}

\begin{lemma} \label{lem:aux1}
 $\|\widehat z_t - S z_t^* \|_2 = \mathcal O_P(T^{-1/2})$.
\end{lemma}
\begin{proof}
By Cauchy-Schwarz, the $\sqrt T$-consistency of the primitives, and bounded fourth moments,
	\begin{align*}
		\|\widehat z_t - S z_t^* \|_2^2 
		&= \sum_{j=1}^J \sum_{l=1}^{K_j} \Big( \langle X_{jt} - \widehat \mu_j, \widehat \psi_{lj} \rangle - \langle X_{jt} - \mu_j, s_{lj} \psi_{lj} \rangle \Big)^2 \\
	&= 	\sum_{j=1}^J \sum_{l=1}^{K_j} \Big( \langle \mu_j - \widehat \mu_j, \widehat \psi_{lj} \rangle + \langle X_{jt} - \mu_j, \widehat \psi_{lj} - s_{lj} \psi_{lj} \rangle \Big)^2 \\
	&\leq 2 \sum_{j=1}^J \sum_{l=1}^{K_j} \| \mu_j - \widehat \mu_j \|^2 + \|X_{jt} - \mu_j \|^2 \|\widehat \psi_{lj} - s_{lj} \psi_{lj} \|^2 = \mathcal O_P(T^{-1}).
	\end{align*}
\end{proof}

\begin{lemma} \label{lem:aux2}
$\|\widehat Q - SQ^*S \|_F = \mathcal O_P(T^{-1/2})$.
\end{lemma}
\begin{proof}
	By the triangle inequality, the relation of Frobenius and Euclidean norm, Lemma \ref{lem:aux1}, and bounded fourth moments,
\begin{align*}
	&\|\widehat Q - SQ^*S \|_F \\
	 &= \bigg\| \frac{1}{T} \sum_{t=1}^T  \widehat z_t \widehat z_t' - (S z_t^*)(S z_t^*)' \bigg\|_F  \\
	&= \bigg\| \frac{1}{T} \sum_{t=1}^T  (\widehat z_t - S z_t^*)' \widehat z_t' + (S z_t^*)(\widehat z_t - S z_t^*)' \bigg\|_F \\
    &\leq \bigg\| \frac{1}{T} \sum_{t=1}^T  (\widehat z_t - S z_t^*)' \widehat z_t' \bigg\|_F + \bigg\| \frac{1}{T} \sum_{t=1}^T  (S z_t^*)(\widehat z_t - S z_t^*)' \bigg\|_F \\
    &= \bigg( \frac{1}{T} \sum_{t=1}^T \|\widehat z_t - S z_t^*\|_2^2 \bigg)^{1/2} \bigg( \frac{1}{T} \sum_{t=1}^T \|\widehat z_t\|_2^2 \bigg)^{1/2} + \bigg( \frac{1}{T} \sum_{t=1}^T \|z_t^*\|_2^2 \bigg)^{1/2} \bigg( \frac{1}{T} \sum_{t=1}^T \|\widehat z_t - S z_t^*\|_2^2 \bigg)^{1/2} \\
    &= \mathcal O_P(T^{-1/2}).
\end{align*}	
\end{proof}

\begin{lemma} \label{lem:aux3}
$\|Q^* - Q\|_F = o_P(1)$.
\end{lemma}
\begin{proof}
Since $z_t^*$ alpha-mixing of size $-\nu/(\nu-2)$, the result follows from the law of large numbers for alpha-mixing sequences.
\end{proof}

\begin{lemma} \label{lem:aux4}
	$\|\widehat Q - SQS\|_F = o_P(1)$ and $\|\widehat Q^{-1} - SQ^{-1}S\|_F = o_P(1)$.
\end{lemma}
\begin{proof}
Follows by Lemmas \ref{lem:aux2}--\ref{lem:aux3}, the triangle inequality, the fact that $S = S^{-1}$ and the continuous mapping theorem.
\end{proof}

\begin{lemma} \label{lem:aux5}
	$\|\widehat B(r)- SB(r)\|_2 = \mathcal O_P(T^{-1/2})$.
\end{lemma}
\begin{proof}
Since $\beta_j(r) = (B_j(r))'\Psi_j(s)$, we have $\int_a^b \beta_j(r,s) X_{jt}(s) \dd s = (B_j(r))'F_{jt}^*$ and the regression equation can be represented as $Y_t(r) = (S z_t^*)'S B(r) + u_t(r)$.
Therefore,
$$
	\widehat Q \widehat B(r) = \frac{1}{T} \sum_{t=1}^T \widehat z_t Y_t(r) 
	= \widehat Q S B(r) + \frac{1}{T} \sum_{t=1}^T \widehat z_t (S z_t^* - \widehat z_t)' S B(r)  + \frac{1}{T} \sum_{t=1}^T \widehat z_t u_t(r)
$$
and
$$
	\widehat Q^{-1} (\widehat B(r) - S B(r)) =  \frac{1}{T} \sum_{t=1}^T \widehat z_t (S z_t^* - \widehat z_t)' S B(r) + \frac{1}{T} \sum_{t=1}^T (S z_t^* - \widehat z_t) u_t(r) + \frac{1}{T} \sum_{t=1}^T S z_t^* u_t(r).
$$
By Lemma \ref{lem:aux1} and bounded fourth moments,
\begin{align*}
	&\bigg\|\frac{1}{T} \sum_{t=1}^T \widehat z_t (S z_t^* - \widehat z_t)' S B(r)\bigg\|_2 \\
	&\leq \bigg( \frac{1}{T} \sum_{t=1}^T \|\widehat z_t\|_2^2 \bigg)^{1/2} \bigg( \frac{1}{T} \sum_{t=1}^T \|S z_t^* - \widehat z_t\|_2^2 \bigg)^{1/2} \|B(r)\|_2 = \mathcal O_P(T^{-1/2})
\end{align*}
and 
$$
	\bigg\|\frac{1}{T} \sum_{t=1}^T (S z_t^* - \widehat z_t) u_t(r)\bigg\|_2 \leq \bigg( \frac{1}{T} \sum_{t=1}^T \|S z_t^* - \widehat z_t\|_2^2 \bigg)^{1/2} \bigg( \frac{1}{T} \sum_{t=1}^T |u_t(r)|^2 \bigg)^{1/2}  = \mathcal O_P(T^{-1/2}).
$$
The result follows from Lemma \ref{lem:aux4} and
$$
	\bigg\| \frac{1}{T} \sum_{t=1}^T S z_t^* u_t(r) \bigg\|_2 = \mathcal O_P(T^{-1/2}),
$$
which holds because $z_t^* u_t(r)$ is a zero mean martingale difference sequence for all $r \in [a,b]$ by Assumption \ref{as:exo}.
\end{proof}

\begin{lemma} \label{lem:aux6}
$\|\sum_{t=1}^T (\widehat z_t - S z_t^*) u_t(r)\|_2 = \mathcal O_P(1)$.
\end{lemma}
\begin{proof}
Similarly to the proof of Lemma \ref{lem:aux1},
	\begin{align*}
	&\bigg\|\sum_{t=1}^T (\widehat z_t - S z_t^*) u_t(r) \bigg\|_2^2 \\
	&= \sum_{j=1}^J \sum_{l=1}^{K_j} \bigg(\sum_{t=1}^T \Big( \langle \mu_j - \widehat \mu_j, \widehat \psi_{lj} \rangle + \langle X_{jt} - \mu_j, \widehat \psi_{lj} - s_{lj} \psi_{lj} \rangle \Big) u_t(r) \bigg)^2 \\
	&\leq 2 \sum_{j=1}^J \sum_{l=1}^{K_j} \bigg( \sum_{t=1}^T  \langle \mu_j - \widehat \mu_j, \widehat \psi_{lj} \rangle u_t(r) \bigg)^2 + \bigg( \sum_{t=1}^T  \langle X_{jt} - \mu_j, \widehat \psi_{lj} - s_{lj} \psi_{lj} \rangle u_t(r) \bigg)^2 \\
	&\leq 2 \sum_{j=1}^J \sum_{l=1}^{K_j} \| \mu_j - \widehat \mu_j \|^2 \bigg(\sum_{t=1}^T u_t(r)\bigg)^2
+ \|\widehat \psi_{lj} - s_{lj} \psi_{lj} \|^2 \bigg\| \sum_{t=1}^T (X_{jt} - \mu_j) u_t(r) \bigg\|^2	
	 = \mathcal O_P(1)
	\end{align*}
by the consistency of the primitives, bounded fourth moments, Cauchy-Schwarz, and because $u_t(r)$ and $(X_{jt}(s) - \mu_j(s)) u_t(r)$ are zero mean martingale difference sequences for all $r,s \in [a,b]$ by Assumption \ref{as:exo}.
\end{proof}

\begin{lemma} \label{lem:aux7}
$\langle \widehat \mu_j - \mu_j, \widehat \psi_{lj} \rangle = \frac{1}{T} \sum_{t=1}^T s_{lj} f_{ljt}^* + \mathcal O_P(T^{-1})$.
\end{lemma}
\begin{proof}
We have $\langle \widehat \mu_j - \mu_j, \widehat \psi_{lj} \rangle 
	= \frac{1}{T} \sum_{t=1}^T \langle X_{jt} - \mu_j, \widehat \psi_{lj} \rangle$ and $f_{ljt}^* = \langle X_{jt} - \mu_j,  \psi_{lj} \rangle$, which implies
\begin{align*}
	\langle \widehat \mu_j - \mu_j, \widehat \psi_{lj} \rangle 
	- \frac{1}{T} \sum_{t=1}^T s_{lj} f_{ljt}^* 
	&= \frac{1}{T} \sum_{t=1}^T \langle X_{jt} - \mu_j,  \widehat \psi_{lj} - s_{lj} \psi_{lj}\rangle \\
	&= \Big\langle \frac{1}{T} \sum_{t=1}^T (X_{jt} - \mu_j), \psi_{lj} - s_{lj} \psi_{lj} \Big\rangle \\
	&\leq \bigg\| \frac{1}{T} \sum_{t=1}^T (X_{jt} - \mu_j) \bigg\|\|\widehat \psi_{lj} - s_{lj} \psi_{lj}\|,
\end{align*}
because $\|\widehat \psi_{lj} - s_{lj} \psi_{lj}\| = \mathcal O_P(T^{-1/2})$ and $\| \frac{1}{\sqrt T} \sum_{t=1}^T (X_t - \mu)\| = \mathcal O_P(1)$ by the central limit theorem for alpha-mixing sequences.
\end{proof}

\begin{lemma} \label{lem:aux8}
$\langle s_{lj} \psi_{lj}, \widehat \psi_{lj} - s_{lj} \psi_{lj}  \rangle = \mathcal O_P(T^{-1})$.
\end{lemma}
\begin{proof}
We have
\begin{align*}
	&2 \langle s_{lj} \psi_{lj},  \widehat \psi_{lj} - s_{lj} \psi_{lj} \rangle + \|\widehat \psi_{lj} - s_{lj} \psi_{lj}\|^2  \\
	&= \langle 2 s_{lj} \psi_{lj}, \widehat \psi_{lj} - s_{lj} \psi_{lj} \rangle + \langle \widehat \psi_{lj} - s_{lj}  \psi_{lj}, \widehat \psi_{lj} - s_{lj} \psi_{lj} \rangle \\
	&= \langle \widehat \psi_{lj} + s_{lj}  \psi_{lj}, \widehat \psi_{lj} - s_{lj}  \psi_{lj} \rangle \\
	&= \langle \widehat \psi_{lj}, \widehat \psi_{lj} \rangle
	- \langle \widehat \psi_{lj}, s_{lj} \psi_{lj} \rangle
	+ \langle s_{lj} \psi_{lj}, \widehat \psi_{lj} \rangle
	- \langle s_{lj} \psi_{lj}, s_{lj} \psi_{lj} \rangle = 0
\end{align*}
by the orthonormality of the eigenfunctions.
Therefore,
\begin{align*}
	2 \langle s_{lj} \psi_{lj}, \widehat \psi_{lj} - s_{lj} \psi_{lj} \rangle = - \|\widehat \psi_{lj} - s_{lj}\psi_{lj}\|^2 = \mathcal O_P(T^{-1})
\end{align*}
by the $\sqrt T$-consistency of the primitives.
This trick is inspired by the approach in \textcite{Kokoszka.2013}.
\end{proof}

\begin{lemma} \label{lem:aux9}
$$
\langle \psi_{mj}, (\widehat D_j - D_j)(\psi_{lj}) \rangle = 
\frac{1}{T} \sum_{t=1}^T \Big( (f_{mjt}^* y_{ljt} - E[f_{mjt}^* y_{ljt}]) + (f_{ljt}^* y_{mjt} - E[f_{ljt}^* y_{mjt}]) \Big) + \mathcal O_P(T^{-1}).
$$
\end{lemma}
\begin{proof}
	We have $\widehat D_j = \widehat C_j \widehat C_j^*$ and $D_j = C_j C_j^*$. Therefore,
\begin{align*}
	\widehat D_j - D_j &= (\widehat C_j - C_j)(\widehat C_j^* - C_j^*) + (\widehat C_j - C_j) C_j^* + C_j(\widehat C_j^* - C_j^*) \\
	&= (\widehat C_j - C_j) C_j^* + C_j(\widehat C_j^* - C_j^*) + \mathcal O_P(T^{-1})
\end{align*}
because $\|(\widehat C_j - C_j)(\widehat C_j^* - C_j^*)\|_{\mathcal S} \leq \|\widehat C_j - C_j\|_{\mathcal S}^2 = \mathcal O_P(T^{-1})$.
Furthermore, let $\widetilde C_j$ be the sample cross-covariance operator with true mean, which is the integral operator with the kernel function $\widetilde c_j(r,s) = \frac{1}{T} \sum_{t=1}^T (X_{jt}(r) - \mu_j(r))(Y_t(s) - E[Y_t(s)])$.
Then, 
$$\widehat c_j(r,s) - \widetilde c_j(r,s) = (\widehat \mu_j(r) - \mu_j(r))(E[Y_t(s)] - \overline Y(s)),$$
and $\|\widehat C_j - \widetilde C_j\|_{\mathcal S} = \mathcal O_P(T^{-1})$ by the $\sqrt T$-consistency of the primitives, which implies
$$
\widehat D_j - D_j =  (\widetilde C_j - C_j) C_j^* + C_j(\widetilde C_j^* - C_j^*) + \mathcal O_P(T^{-1}).
$$
Therefore,
\begin{align*}
	\langle \psi_{mj}, (\widehat D_j - D_j)(\psi_{lj}) \rangle 
	&= \langle \psi_{mj}, (\widetilde C_j - C_j) C_j^*(\psi_{lj}) \rangle + \langle \psi_{mj}, C_j(\widetilde C_j^* - C_j^*)(\psi_{lj}) \rangle + \mathcal O_P(T^{-1}) \\
	&= \langle (\widetilde C_j^* - C_j^*)(\psi_{mj}),  C_j^*(\psi_{lj}) \rangle + \langle C_j^*(\psi_{mj}), (\widetilde C_j^* - C_j^*)(\psi_{lj}) \rangle \\
	&=  \langle \psi_{mj}, (\widetilde C_j - C_j)(\gamma_{lj}) \rangle + \langle \psi_{lj}, (\widetilde C_j - C_j)(\gamma_{mj}) \rangle + \mathcal O_P(T^{-1}),
\end{align*}
where the last step follows with $\gamma_{lj}(q) = E[f_{ljt}^* Y_t(q)]$ because $\psi_{1j}, \ldots, \psi_{K_j j}$ are the left-singular functions of $C_j$ with $c_j(r,s) = (\Psi_j(s))'E[F_{jt}^* Y_t(r)]$ and $C_j^*(\psi_{lj})(q) = E[f_{ljt}^* Y_t(q)]$.
The kernel function of the integral operator $\widetilde C_j - C_j$ is
$$
	\widetilde c_j(r,s) - c_j(r,s) = \frac{1}{T} \sum_{t=1}^T w_{jt}(r,s)
$$
with 
$$
w_{jt}(r,s) = (X_{jt}(r) - \mu_j(r))(Y_t(s) - E[Y_t(s)]) - E[(X_{jt}(r) - \mu_j(r))(Y_t(s) - E[Y_t(s)])].
$$
Therefore, with $f_{mjt}^* = \langle X_{jt} - \mu_j, \psi_{mj} \rangle$ and $y_{ljt} = \langle Y_t - E[Y_t], \psi_{lj} \rangle$,
$$
	\langle \psi_{mj}, (\widetilde C_j - C_j)(\gamma_{lj}) \rangle = \frac{1}{T} \sum_{t=1}^T \Big( f_{mjt}^* y_{ljt} - E[f_{mjt}^* y_{ljt}] \Big),
$$
and the result follows.
\end{proof}

\begin{lemma} \label{lem:aux10}
For $m \neq l$, 
$$\langle s_{mj} \psi_{mj}, \widehat \psi_{lj} - s_{lj} \psi_{lj}  \rangle = \frac{1}{T} \sum_{t=1}^T \frac{ g_{lmjt}}{\lambda_{lj} - \lambda_{mj}} + \mathcal O_P(T^{-1}).$$
\end{lemma}
\begin{proof}
The eigenpair $(\lambda_{mj}, \psi_{mj})$ is defined through the equation $D_j(\psi_{mj})(r) = \lambda_{mj} \psi_{mj}(r)$, where the operator $D_j$ is symmetric with the property that $\langle D_j(x), y \rangle = \langle x, D_j(y) \rangle$ for all $x,y \in H$.
Therefore, 
$$
\lambda_{mj} \langle \psi_{mj}, \widehat \psi_{lj}\rangle
	= \langle D_j(\psi_{mj}), \widehat \psi_{lj}  \rangle
	= \langle \psi_{mj}, D_j(\widehat \psi_{lj}) \rangle.
$$
The sample eigenpair $(\widehat \lambda_{lj}, \widehat \psi_{lj})$ is defined through the equation $\widehat D_j(\widehat \psi_{lj})(r) = \widehat \lambda_{lj} \widehat \psi_{lj}(r)$.
Then,
\begin{align*}
	&(\widehat \lambda_{lj} - \lambda_{mj}) \langle s_{mj} \psi_{mj}, \widehat \psi_{lj} - s_{lj} \psi_{lj}  \rangle \\
	&= (\widehat \lambda_{lj} - \lambda_{mj}) \langle s_{mj} \psi_{mj}, \widehat \psi_{lj} \rangle \\
	&= \langle s_{mj} \psi_{mj}, \widehat \lambda_{lj} \widehat \psi_{lj} \rangle - \lambda_{mj} \langle s_{mj} \psi_{mj}, \widehat \psi_{lj} \rangle \\
	&= \langle s_{mj} \psi_{mj}, \widehat D_j(\widehat \psi_{lj}) \rangle - \langle s_{mj} \psi_{mj}, D_j(\widehat \psi_{lj}) \rangle \\
	&= \langle s_{mj} \psi_{mj}, (\widehat D_j - D_j)(\widehat \psi_{lj}) \rangle \\
	&= \langle s_{mj} \psi_{mj}, (\widehat D_j - D_j)(s_{lj} \psi_{lj}) \rangle  + \langle s_{mj} \psi_{mj}, (\widehat D_j - D_j)(\widehat \psi_{lj} - s_{lj} \psi_{lj}) \rangle,
\end{align*}
where
$$
|\langle s_{mj} \psi_{mj}, (\widehat D_j - D_j)(\widehat \psi_{lj} - s_{lj} \psi_{lj}) \rangle| \leq \|(\widehat D_j - D_j)(\widehat \psi_{lj} - s_{lj} \psi_{lj})\| \leq \|\widehat D_j - D_j\|_{\mathcal S} \|\widehat \psi_{lj} - s_{lj} \psi_{lj}\|
$$
which is $\mathcal O_P(T^{-1})$ by the consistency of the primitives.
Therefore, with Lemma \ref{lem:aux9},
\begin{align*}
	&(\widehat \lambda_{lj} - \lambda_{mj}) \langle s_{mj} \psi_{mj}, \widehat \psi_{lj} - s_{lj} \psi_{lj}  \rangle \\
	&= \langle s_{mj} \psi_{mj}, (\widehat D_j - D_j)(s_{lj} \psi_{lj}) \rangle  + \mathcal O_P(T^{-1}) \\
	&= \frac{1}{T} \sum_{t=1}^T s_{mj} s_{lj}\Big( (f_{mjt}^* y_{ljt} - E[f_{mjt}^* y_{ljt}]) + (f_{ljt}^* y_{mjt} - E[f_{ljt}^* y_{mjt}]) \Big) + \mathcal O_P(T^{-1}).
\end{align*}
This term is $\mathcal O_P(T^{-1/2})$ because $(f_{mjt}^* y_{ljt} - E[f_{mjt}^* y_{ljt}]) + (f_{ljt}^* y_{mjt} - E[f_{ljt}^* y_{mjt}])$ is mean zero $\alpha$-mixing of size $-\nu/(\nu-2)$ and satisfies the conditions for the central limit theorem for alpha-mixing sequences. Therefore, with $|\widehat \lambda_{lj} - \lambda_{lj}| = \mathcal O_P(T^{-1/2})$, 
$$
(\widehat \lambda_{lj} - \lambda_{lj}) \langle s_{mj} \psi_{mj}, \widehat \psi_{lj} - s_{lj} \psi_{lj} \rangle = \mathcal O_P(T^{-1}),
$$
and we get
\begin{align*}
	&(\lambda_{lj} - \lambda_{mj}) \langle s_{mj} \psi_{mj}, \widehat \psi_{lj} - s_{lj} \psi_{lj}  \rangle \\
	&= \frac{1}{T} \sum_{t=1}^T s_{mj} s_{lj}\Big( (f_{mjt}^* y_{ljt} - E[f_{mjt}^* y_{ljt}]) + (f_{ljt}^* y_{mjt} - E[f_{ljt}^* y_{mjt}]) \Big) + \mathcal O_P(T^{-1}).
\end{align*}
Then, the result follows.
\end{proof}

\begin{lemma} \label{lem:aux11}
	$\sum_{t=1}^T \widehat z_t \langle \varepsilon_{jt}^*, s_{lj} \psi_{lj} - \widehat \psi_{lj} \rangle = \mathcal O_P(1)$.
\end{lemma}
\begin{proof}
With $\int_a^b \beta_j(r,q) E[X_{jt}(q) \varepsilon_{jt}(s)] \dd q = 0$ it follows that $E[f_{mjt}^* \varepsilon_{jt}(s)] = 0$. Further, because $f_{mjt}^* = f_{mjt} + \langle \varepsilon_{jt}, \psi_{mj} \rangle$ and $\varepsilon_{jt}^*(s) = \varepsilon_{jt}(s) - \sum_{m=1}^{K_j} \langle \varepsilon_{jt}, \psi_{mj} \rangle \psi_{mj}(s)$ we have $E[f_{mjt}^* \varepsilon_{jt}^*(s)] = 0$.
Therefore, $z_t^* \varepsilon_{jt}^*(s)$ is a zero mean alpha-mixing sequence of order $-\nu/(\nu-2)$ with 
$$
		 \int_a^b \bigg\| \frac{1}{\sqrt T} \sum_{t=1}^T z_t^* \varepsilon_{jt}^*(s) \bigg\|_2^2 \dd s = \mathcal O_P(1),
$$
which follows by the central limit theorem for alpha-mixing sequences.
Furthermore, with Lemma \ref{lem:aux1},
$$
		 \int_a^b \bigg\| \frac{1}{\sqrt T} \sum_{t=1}^T \widehat z_t \varepsilon_{jt}^*(s) \bigg\|_2^2 \dd s = \mathcal O_P(1).
$$
Consequently, by the $\sqrt T$-consistency of the primitives,
\begin{align*}
	\bigg\| \sum_{t=1}^T \widehat z_t \langle \varepsilon_{jt}^*, s_{lj} \psi_{lj} - \widehat \psi_{lj} \rangle \bigg\|_2^2 \leq \|s_{lj} \psi_{lj} - \widehat \psi_{lj} \rangle\|^2 \int_a^b \bigg\| \sum_{t=1}^T \widehat z_t \varepsilon_{jt}^*(s) \bigg\|_2^2 \dd s = \mathcal O_P(1).
\end{align*}
\end{proof}

\begin{lemma} \label{lem:aux12}
$\frac{1}{\sqrt T} \sum_{t=1}^T \widehat z_t \langle \widehat \mu_j - \mu_j, \widehat \psi_{lj} \rangle  = \frac{1}{\sqrt T} \sum_{t=1}^T S E[z_t^*] s_{lj} f_{ljt}^* + o_P(1)$.
\end{lemma}
\begin{proof}
We have
\begin{align*}
	\frac{1}{\sqrt T} \sum_{t=1}^T \widehat z_t \langle \widehat \mu_j - \mu_j, \widehat \psi_{lj} \rangle 
	&= \bigg( \frac{1}{T} \sum_{t=1}^T \widehat z_t \bigg) \bigg( \frac{1}{\sqrt T} \sum_{t=1}^T s_{lj} f_{ljt}^* \bigg) + \mathcal O_P(T^{-1/2}) \\
	&=\bigg( \frac{1}{T} \sum_{t=1}^T S z_t^* \bigg) \bigg( \frac{1}{\sqrt T} \sum_{t=1}^T s_{lj} f_{ljt}^* \bigg) + \mathcal O_P(T^{-1/2}) \\
	&= \frac{1}{\sqrt T} \sum_{t=1}^T E[S z_t^*] s_{lj} f_{ljt}^* + o_P(1) 
\end{align*}
where we applied Lemma \ref{lem:aux7} in the first line, Lemma \ref{lem:aux1} in the second line, and the law of large numbers for alpha-mixing sequences in the third line.
\end{proof}

\begin{lemma} \label{lem:aux13}
$$\frac{1}{\sqrt T} \sum_{t=1}^T \widehat z_t \langle X_{jt} - \mu_j, s_{lj} \psi_{lj} - \widehat \psi_{lj} \rangle = - \frac{1}{\sqrt T} \sum_{t=1}^T \sum_{\substack{m=1 \\ m\neq l}}^{K_j} S E[z_t^* f_{mjt}^*]  \frac{ s_{mj} g_{lmjt}}{\lambda_{lj} - \lambda_{mj}} + o_P(1).$$
\end{lemma}
\begin{proof}
The $j$-th regressor function admits the representation $X_{jt}(s) - \mu_j(s) = (F_{lj}^*)'(\Psi_j(s)) + \varepsilon_{jt}^*(s)$, where $\varepsilon_{jt}^*(s) = X_{jt}(s) - \mu_j(s) - (F_{lj}^*)'(\Psi_j(s))$.
Then,
\begin{align*}
&\frac{1}{\sqrt T} \sum_{t=1}^T \widehat z_t \langle X_{jt} - \mu_j, s_{lj} \psi_{lj} - \widehat \psi_{lj} \rangle \\
&= - \frac{1}{\sqrt T} \sum_{t=1}^T \sum_{m=1}^{K_j} \widehat z_t f_{mjt}^* s_{mj} \langle s_{mj} \psi_{mj}, \widehat \psi_{lj} - s_{lj} \psi_{lj} \rangle + \frac{1}{\sqrt T} \sum_{t=1}^T \widehat z_t \langle \varepsilon_{jt}^*, s_{lj} \psi_{lj} - \widehat \psi_{lj} \rangle \\
&= - \frac{1}{\sqrt T} \sum_{t=1}^T \sum_{\substack{m=1 \\ m\neq l}}^{K_j} \widehat z_t f_{mjt}^* s_{mj} \bigg( \frac{1}{T} \sum_{h=1}^T \frac{ g_{lmjh}}{\lambda_{lj} - \lambda_{mj}}  \bigg) + \mathcal O_P(T^{-1/2})
\end{align*}
where the last step follows from Lemmas \ref{lem:aux8}, \ref{lem:aux10}, and \ref{lem:aux11}.
Furthermore,
$$
\frac{1}{T} \sum_{t=1}^T \widehat z_t f_{mjt}^* = \frac{1}{T} \sum_{t=1}^T S z_t^* f_{mjt}^* + \frac{1}{T} \sum_{t=1}^T (\widehat z_t - S z_t^*) f_{mjt}^* = S E[z_t^* f_{mjt}^*] + o_P(1)
$$
by Lemma \ref{lem:aux1} and the law of large numbers for alpha-mixing sequences, and the result follows.
\end{proof}

\begin{lemma} \label{lem:aux14}
	$|\widehat u_t(r) - u_t(r)| = \mathcal O_P(T^{-1/2})$.
\end{lemma}
\begin{proof}
With $u_t(r) = Y_t(r) - (S z_t^*)'SB(r)$ and $\widehat u_t(r) = Y_t(r) - \widehat z_t'\widehat B(r)$ we have
\begin{align*}
	|\widehat u_t(r) - u_t(r)| &= |(S z_t^*)'SB(r) - \widehat z_t'\widehat B(r)| \\
	&= |(Sz_t^* - \widehat z_t)'SB(r) + \widehat z_t'(SB(r) - \widehat B(r))|  \\
	&\leq \|Sz_t^* - \widehat z_t\|_2 \|B(r)\|_2 + \|\widehat z_t\|_2 \|SB(r) - \widehat B(r)\|_2 = \mathcal O_P(T^{-1/2})
\end{align*}
by the triangle inequality, bounded fourth moments, and Lemmas \ref{lem:aux1} and \ref{lem:aux5}.
\end{proof}

\begin{lemma} \label{lem:aux15}
	$|\widehat f_{ljt} - s_{lj} f_{ljt}^*| = \mathcal O_P(T^{-1/2})$.
\end{lemma}
\begin{proof}
With Cauchy-Schwarz, bounded fourth moments, and the $\sqrt T$-consistency of the primitives,
\begin{align*}
	|\widehat f_{ljt} - s_{lj} f_{ljt}^*|
	&= |\langle X_{jt} - \widehat \mu_j, \widehat \psi_{lj} \rangle - \langle X_{jt} - \mu_j, s_{lj} \psi_{lj} \rangle| \\
&= |\langle \mu_j - \widehat \mu_j, \widehat \psi_{lj} \rangle + \langle X_{jt} - \mu_j,\widehat \psi_{lj} - s_{lj} \psi_{lj} \rangle| \\
&\leq \|\mu_j - \widehat \mu_j \| + \|X_{jt} - \mu_j\| \|\widehat \psi_{lj} - s_{lj} \psi_{lj} \| = \mathcal O_P(T^{-1/2}).
\end{align*}
\end{proof}

\begin{lemma} \label{lem:aux16}
	$|\widehat g_{lmjt} - g_{lmjt}| = o_P(1)$.
\end{lemma}
\begin{proof}
Analogously to the proof of Lemma \ref{lem:aux15},
$$
|\widehat y_{ljt} - s_{lj} y_{ljt}| \leq \|E[Y_t] - \overline Y\| + \|Y_t - E[Y_t]\| \|\widehat \psi_{lj} - s_{lj} \psi_{lj} \| = \mathcal O_P(T^{-1/2}),
$$
and, therefore,
\begin{align*}
	\widehat f_{mjt} \widehat y_{ljt} - s_{mj} s_{lj} f_{mjt}^* y_{ljt}
	&= (\widehat f_{mjt} - s_{mj}  f_{mjt}^*)\widehat y_{ljt} + s_{mj} f_{mjt}^*(\widehat y_{ljt} - s_{lj} y_{ljt}) = \mathcal O_P(T^{-1/2})
\end{align*}
by Lemma \ref{lem:aux15}, Cauchy-Schwarz, and bounded fourth moments.
By similar arguments,
\begin{align*}
	\overline{f_{mj} y_{lj}} = \frac{1}{T} \sum_{t=1}^T \widehat f_{mjt} \widehat y_{ljt} = \frac{1}{T} \sum_{t=1}^T s_{mj} f_{mjt}^* s_{lj} y_{ljt} + \mathcal O_P(T^{-1/2}),
\end{align*}
and by the law of large numbers for alpha-mixing sequences,
$$
	\overline{f_{mj} y_{lj}} = s_{mj} s_{lj} E[f_{mjt}^* y_{ljt}] + o_P(1).
$$
Therefore,
\begin{align*}
	\widehat g_{lmjt} - g_{lmjt} &= (\widehat f_{mjt} \widehat y_{ljt} - s_{mj} s_{lj} f_{mjt}^* y_{ljt}) - (\overline{f_{mj} y_{lj}} - s_{mj} s_{lj} E[f_{mjt}^* y_{ljt}]) \\
	&\quad + (\widehat f_{ljt} \widehat y_{mjt} - s_{mj} s_{lj} f_{mjt}^* y_{mjt}) - (\overline{f_{lj} y_{mj}} - s_{mj} s_{lj} E[f_{ljt}^* y_{mjt}]) \\
	&= o_P(1).
\end{align*}
\end{proof}

\begin{lemma} \label{lem:aux17}
$$
\widetilde A(r) = \frac{1}{\sqrt T} \sum_{t=1}^T S z_t^* u_t(r) + \sum_{j=1}^J \bigg( S E[z_t^*] (F_{jt}^*)'S_j - S E[z_t^* (F_{jt}^*)'] G_{jt} \bigg) B_j(r) + o_P(1),
$$
and $\widetilde A(r) = \mathcal O_P(1)$,
where $\widetilde A(r)$ is defined in \eqref{eq:thm3Atilde} and $G_{jt}$ is defined in \eqref{eq:thm3G}
\end{lemma}
\begin{proof}
Since $\beta_j(r) = (B_j(r))'\Psi_j(s)$, we have $\int_a^b \beta_j(r,s) X_{jt}(s) \dd s = (B_j(r))'F_{jt}^*$ and the regression equation can be represented as $Y_t(r) = (S z_t^*)'S B(r) + u_t(r)$.
Therefore,
\begin{align*}
	&\widetilde A(r) =  \frac{1}{\sqrt T} \sum_{t=1}^T \widehat z_t u_t(r) + \frac{1}{\sqrt T} \sum_{t=1}^T \widehat z_t (S z_t^* - \widehat z_t)' S B(r) \\
	&= \frac{1}{\sqrt T} \sum_{t=1}^T \widehat z_t u_t(r) + \frac{1}{\sqrt T} \sum_{t=1}^T \widehat z_t \sum_{j=1}^J \sum_{l=1}^{K_j} \Big( \langle \widehat \mu_j - \mu_j, \widehat \psi_{lj} \rangle + \langle X_{jt} - \mu_j, s_{lj} \psi_{lj} - \widehat \psi_{lj} \rangle \Big) s_{lj} \beta_{lj}(r) \\
	&= \frac{1}{\sqrt T} \sum_{t=1}^T \bigg( S z_t^* u_t(r) + \sum_{j=1}^J \sum_{l=1}^{K_j} \Big( S E[z_t^*] s_{lj} f_{ljt}^* - \sum_{\substack{m=1 \\ m\neq l}}^{K_j} S E[z_t^* f_{mjt}^*]s_{mj}  \frac{ g_{lmjt}}{\lambda_{lj} - \lambda_{mj}} \Big) s_{lj} \beta_{lj}(r) \bigg)  \\
	& \quad + o_P(1) 
\end{align*}
where we applied Lemmas \ref{lem:aux6}, \ref{lem:aux12}, and \ref{lem:aux13}. 
We have $\widetilde A(r) = \mathcal O_P(1)$ by the central limit theorem for alpha-mixing sequences because $z_t^* u_t(r)$, $f_{ljt}^*$, and $g_{lmjt}$ have mean zero and are alpha-mixing of size $-\nu/(\nu-2)$.
\end{proof}

\begin{lemma} \label{lem:aux18}
$$
	\widetilde A(r) = \frac{1}{\sqrt T} \sum_{t=1}^T \bigg( \widehat z_t \widehat u_t(r) + \sum_{j=1}^J  \overline{z} \widehat F_{jt}' \widehat B_j(r) - \overline{z F_j'} \widehat G_{jt} \widehat B_j(r) \bigg) + o_P(1)
$$
\end{lemma}
\begin{proof}
This follows from Lemma \ref{lem:aux17} together with the $\sqrt T$-consistency of the primitives, the facts that
$$
	\overline{z} - E[\widehat z_t] = o_P(1), \quad
	\overline{z F_{j}'} - E[\widehat z_t \widehat F_{jt}'] = o_P(1)
$$
by the law of large numbers for alpha-mixing sequences, and Lemmas \ref{lem:aux1}, \ref{lem:aux5}, \ref{lem:aux14}, \ref{lem:aux15}, \ref{lem:aux16}.
\end{proof}

\begin{lemma} \label{lem:aux19}
Let $\{v_{ij}\}_{i \in \mathbb N}$ be an orthonormal basis of $span(\psi_{1j}, \ldots, \psi_{K_j j})^\perp$. Then,
$$
	\sqrt T \langle \widehat \psi_{lj}, v_{ij} \rangle  = \frac{1}{\sqrt T} \sum_{t=1}^T \frac{1}{\lambda_{lj}} \langle \widehat \varepsilon_{jt}, v_{ij} \rangle \langle s_{lj} \gamma_{lj}, Y_t - E[Y_t] \rangle + o_P(1).
$$
\end{lemma}
\begin{proof}
With $\widehat D_j = \widehat C_j \widehat C_j^*$ we have
$$
\langle \widehat \psi_{lj}, v_{ij} \rangle = \frac{1}{\widehat \lambda_{lj}} \langle \widehat D_j(\widehat \psi_{lj}), v_{ij} \rangle 
= \frac{1}{\widehat \lambda_{lj}} \langle \widehat C_j^*(\widehat \psi_{lj}), \widehat C_j^*(v_{ij}) \rangle,
$$
where
$$
\widehat C_j^*(v_{mj})(r) = \frac{1}{T} \sum_{t=1}^T \langle \widehat \varepsilon_{jt}, v_{ij} \rangle (Y_t(r) - \overline Y(r)), \quad \widehat \varepsilon_{jt}(s) = X_{jt}(s) - \widehat \mu_j(s) - \widehat F_{jt}'\widehat \Psi_j(s),
$$
and, by Lemma \ref{lem:aux15} and the consistency of $\overline Y$,
\begin{align*}
	\widehat C_j^*(\widehat \psi_{lj})(r) &= \frac{1}{T} \sum_{t=1}^T \widehat f_{ljt} (Y_t(r) - \overline Y(r)) \\
	&= \frac{1}{T} \sum_{t=1}^T s_{lj} f_{ljt}^*  (Y_t(r) - E[Y_t(r)]) + \mathcal O_P(T^{-1/2}) \\
	&= s_{lj} \gamma_{lj}(r) + o_P(1)
\end{align*}
where the last step follows from the law of large numbers for alpha-mixing sequences implying
$$
	\frac{1}{T} \sum_{t=1}^T f_{ljt}^*  (Y_t(r) - E[Y_t(r)]) \overset{p}{\to} \gamma_{lj}(r) = E[f_{ljt}^* (Y_t(r) - E[Y_t(r)])].
$$
Therefore,
$$
	\sqrt T \langle \widehat \psi_{lj}, v_{ij} \rangle  = \frac{1}{\sqrt T} \sum_{t=1}^T \frac{1}{\widehat \lambda_{lj}} \langle \widehat \varepsilon_{jt}, v_{ij} \rangle \langle s_{lj} \gamma_{lj}, Y_t - \overline Y \rangle + o_P(1),
$$
and the results follows by the $\sqrt T$-consistency of the primitives, Cauchy-Schwarz, and bounded fourth moments.
\end{proof}

\begin{lemma} \label{lem:aux20}
$|\widehat \varepsilon_{jt}(s) - \varepsilon_{jt}^*(s)| = \mathcal O_P(T^{-1/2})$.
\end{lemma}
\begin{proof}
	This follows by the $\sqrt T$-consistency of the primitives, Lemma \ref{lem:aux15}, Cauchy-Schwarz, and bounded fourth moments.
\end{proof}

\begin{lemma} \label{lem:aux21}
$|\widehat \gamma_{lj}(s) - s_{lj} \gamma_{lj}(s)| = o_P(1)$.
\end{lemma}
\begin{proof}
	By the $\sqrt T$-consistency of the primitives, Lemma \ref{lem:aux15}, Cauchy-Schwarz, and bounded fourth moments,
	$$
		\widehat \gamma_{lj}(s) = \frac{1}{T} \sum_{t=1}^T s_{lj} f_{ljt}^* (Y_t(r) - E[Y_t(r)] ) + \mathcal O_P(T^{-1/2}),
	$$
	and the result follows by the law of large numbers for alpha-mixing sequences.
\end{proof}

\begin{lemma} \label{lem:aux22}
\begin{align*}
	&\sqrt T (\widehat \psi_{lj}(s) - s_{lj} \psi_{lj}(s)) \\
	&= \frac{1}{\sqrt T} \sum_{t=1}^T \bigg( \sum_{\substack{m=1\\ m \neq l}}^{K_j} \frac{ g_{lmjt}}{\lambda_{lj} - \lambda_{mj}} s_{mj} \psi_{mj}(s) +  \frac{1}{\lambda_{lj}} \varepsilon_{jt}^*(s) \langle s_{lj} \gamma_{lj}, Y_t - E[Y_t] \rangle \rangle \bigg) + o_P(1)
\end{align*}
\end{lemma}

\begin{proof}
Let $\{v_{ij}\}_{i \in \mathbb N}$ be an orthonormal basis of $span(\psi_{1j}, \ldots, \psi_{K_j j})^\perp$. Then, $\{\widehat \psi_{lj}\}_{l=1}^{K_j} \cup \{v_{il}\}_{l \in \mathbb N}$ forms an orthonormal basis of $L^2[a,b]$. Using the  representation of $\sqrt T (\widehat \psi_{lj}(s) - s_{lj} \psi_{lj}(s))$ with respect to this basis, we get
\begin{align*}
	&\sqrt T (\widehat \psi_{lj}(s) - s_{lj} \psi_{lj}(s)) \\
	&= \sqrt T \sum_{m=1}^{K_j} \langle \widehat \psi_{lj} - s_{lj} \psi_{lj}, \psi_{mj} \rangle \psi_{mj}(s) + \sqrt T \sum_{i=1}^\infty \langle \widehat \psi_{lj} - s_{lj} \psi_{lj}, v_{ij} \rangle v_{ij}(s) \\
	&= \sqrt T \sum_{\substack{m=1\\ m \neq l}}^{K_j} \langle \widehat \psi_{lj} - s_{lj} \psi_{lj}, \psi_{mj} \rangle \psi_{mj}(s) + \sqrt T \sum_{i=1}^\infty \langle \widehat \psi_{lj}, v_{ij} \rangle v_{ij}(s) + \mathcal O_P(T^{-1/2}) \\
	&= \frac{1}{\sqrt T} \sum_{t=1}^T \bigg( \sum_{\substack{m=1\\ m \neq l}}^{K_j} \frac{ s_{mj} g_{lmjt}}{\lambda_{lj} - \lambda_{mj}} s_{mj} \psi_{mj}(s) +   \sum_{i=1}^\infty \frac{1}{\lambda_{lj}} \langle \widehat \varepsilon_{jt}, v_{ij} \rangle \langle s_{lj} \gamma_{lj}, Y_t - E[Y_t] \rangle v_{ij}(s)  \bigg)+ o_P(1) \\
	&= \frac{1}{\sqrt T} \sum_{t=1}^T \bigg( \sum_{\substack{m=1\\ m \neq l}}^{K_j} \frac{  g_{lmjt}}{\lambda_{lj} - \lambda_{mj}} s_{mj} \psi_{mj}(s) +   \frac{1}{\lambda_{lj}} \widehat \varepsilon_{jt}(s) \langle s_{lj} \gamma_{lj}, Y_t - E[Y_t] \rangle \rangle \bigg) + o_P(1).
\end{align*}
Here, we applied Lemma \ref{lem:aux8} in the third line, Lemmas \ref{lem:aux10} and \ref{lem:aux19} in the fourth line, and the fact that $\langle \widehat \varepsilon_{jt}, \widehat \psi_{lj} \rangle = 0$ by construction implying $\widehat \varepsilon_{jt}(s) = \sum_{i=1}^\infty \langle \widehat \varepsilon_{jt}, v_{ij} \rangle v_{ij}(s)$ in the fifth line.
Finally, the result follows by Lemma \ref{lem:aux20}. 
\end{proof}

\begin{lemma} \label{lem:aux23}
\begin{align*}
	&\sqrt T (\widehat \psi_{lj}(s) - s_{lj} \psi_{lj}(s)) \\
	&= \frac{1}{\sqrt T} \sum_{t=1}^T \bigg( \sum_{\substack{m=1\\ m \neq l}}^{K_j} \frac{ \widehat g_{lmjt}}{\widehat \lambda_{lj} - \widehat \lambda_{mj}} \widehat \psi_{mj}(s) +  \frac{1}{\widehat \lambda_{lj}} \widehat \varepsilon_{jt}(s) \langle \widehat \gamma_{lj}, Y_t - \overline{Y} \rangle \rangle \bigg) + o_P(1)
\end{align*}
\end{lemma}
\begin{proof}
	This follows from Lemma \ref{lem:aux22} and by the $\sqrt T$-consistency of the primitives, Cauchy-Schwarz, bounded fourth moments, and Lemmas \ref{lem:aux16}, \ref{lem:aux20}, and  \ref{lem:aux21}.
\end{proof}

\subsection{Specification of the simulated coefficient functions}\label{app:beta_spec}
The bivariate coefficient functions in the simulation chapter can be expressed in matrix notation as
\begin{equation*}
	\beta_j(r,s)=(V(r))'\mathcal{B}_j V(s),
\end{equation*}
where $V(r) = (v_1(r), v_2(r), \ldots, v_K(r))'$ is a $K \times 1$ vector of Fourier basis functions evaluated at $r$, and $\mathcal{B}_j$ is a $K \times K$ matrix.

In our implementation, we have two slope functions $j=\left\lbrace 1,2\right\rbrace $ with $K=3$ for both of them and we set

$$
\mathcal{B}_1=
\begin{bmatrix}
	-0.03 & 0.09 & 0.15 \\
	0.11 & -0.94 & 0.26 \\
	-0.30 & -0.17 & 0.21
\end{bmatrix}; \quad
\mathcal{B}_2=
\begin{bmatrix}
	-0.41 & -0.29 & -0.28 \\
	-0.42 & -0.26 & 0.09 \\
	-0.73 & -0.12 & 0.15
\end{bmatrix}.
$$
These specific choices of $\mathcal{B}_1$ and $\mathcal{B}_2 $ ensure that $rank(\beta_j(r,s))=K$.

\subsection{Functional regression coefficients and confidence regions}\label{app:func.reg.coefs}
\begin{figure}[H]
	\centering
	\begin{subfigure}[b]{\textwidth}
		\centering
		\includegraphics[width=\textwidth]{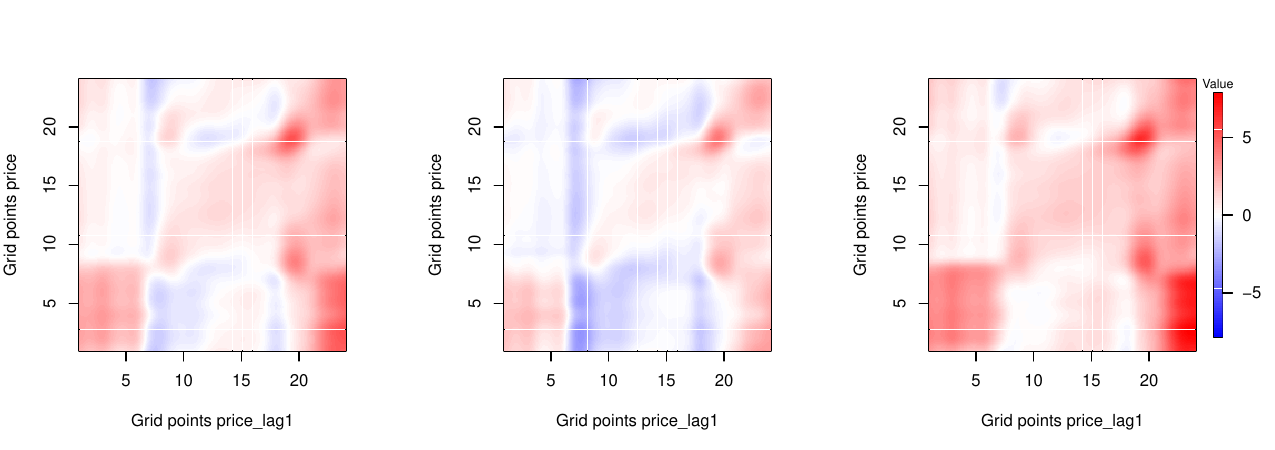}
		\caption{$\widehat{\varphi}_1(r,s)$: impact of $P_{t-1}(s)$ on $P_{t}(r)$.}
	\end{subfigure}
	
	\begin{subfigure}[b]{\textwidth}
		\centering
		\includegraphics[width=\textwidth]{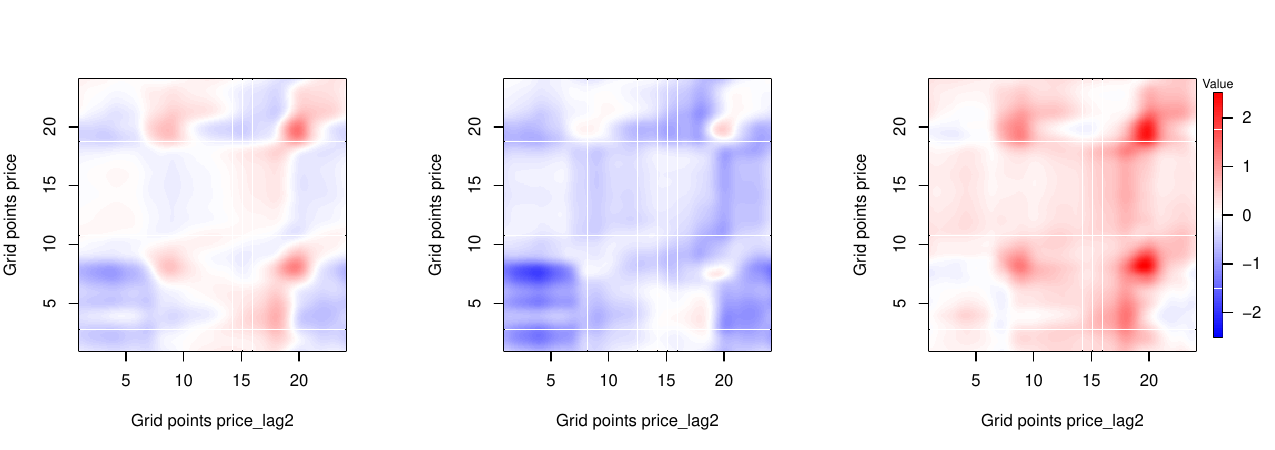}
		\caption{$\widehat{\varphi}_2(r,s)$: impact of $P_{t-2}(s)$ on $P_{t}(r)$.}
	\end{subfigure}
	
	\begin{subfigure}[b]{\textwidth}
		\centering
		\includegraphics[width=\textwidth]{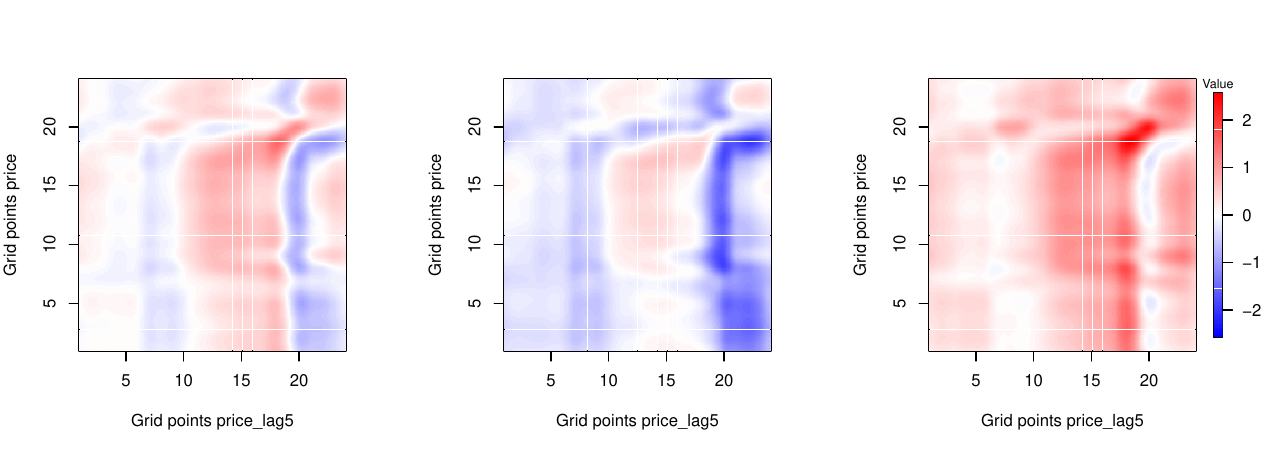}
		\caption{$\widehat{\varphi}_5(r,s)$: impact of $P_{t-5}(s)$ on $P_{t}(r)$.}
	\end{subfigure}
		\caption{First column panels show the estimated functional regression coefficients of model (\ref{eq:elec.FAFR}). Second and third column panels display the lower and upper confidence regions of the $95\%$ confidence interval. The estimated number of factors used in this analysis are shown in Table \ref{ta:epf.K}. The underlying German power market data spans all working days between 2012 and 2017.}
\end{figure}

\begin{figure}[H]\ContinuedFloat
	\centering
	\begin{subfigure}[b]{\textwidth}
		\centering
		\includegraphics[width=\textwidth]{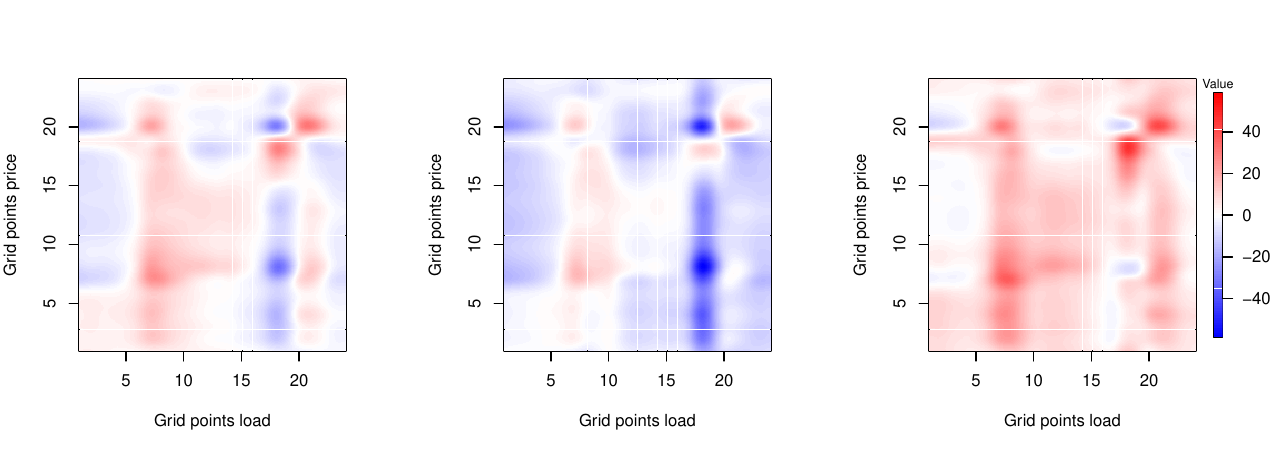}
		\caption{$\widehat{\beta}_{1,1}(r,s)$: impact of $L_{t}(s)$ on $P_{t}(r)$.}
	\end{subfigure}
	
	\begin{subfigure}[b]{\textwidth}
		\centering
		\includegraphics[width=\textwidth]{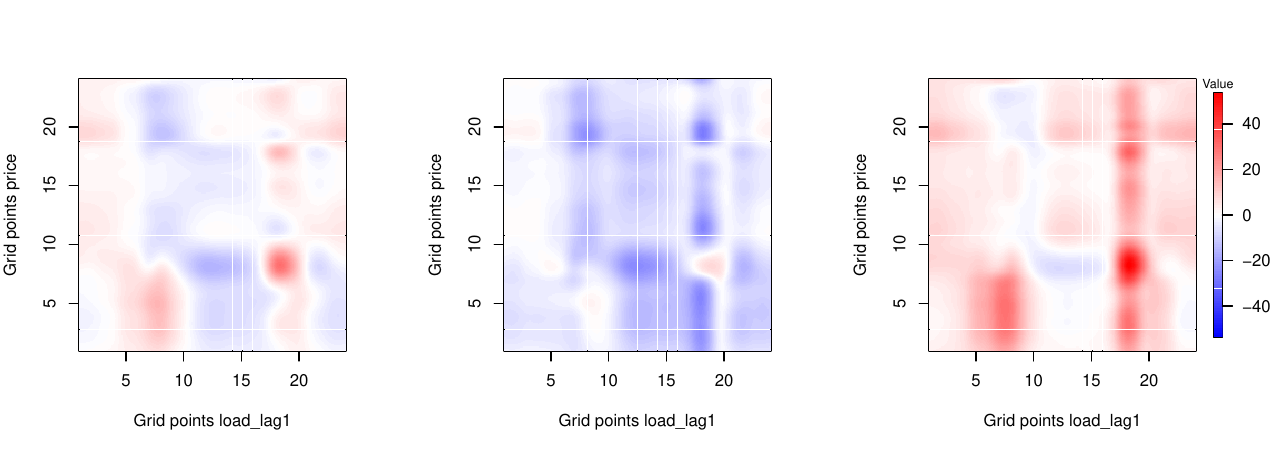}
		\caption{$\widehat{\beta}_{1,2}(r,s)$: impact of $L_{t-1}(s)$ on $P_{t}(r)$.}
	\end{subfigure}
	
	\begin{subfigure}[b]{\textwidth}
		\centering
		\includegraphics[width=\textwidth]{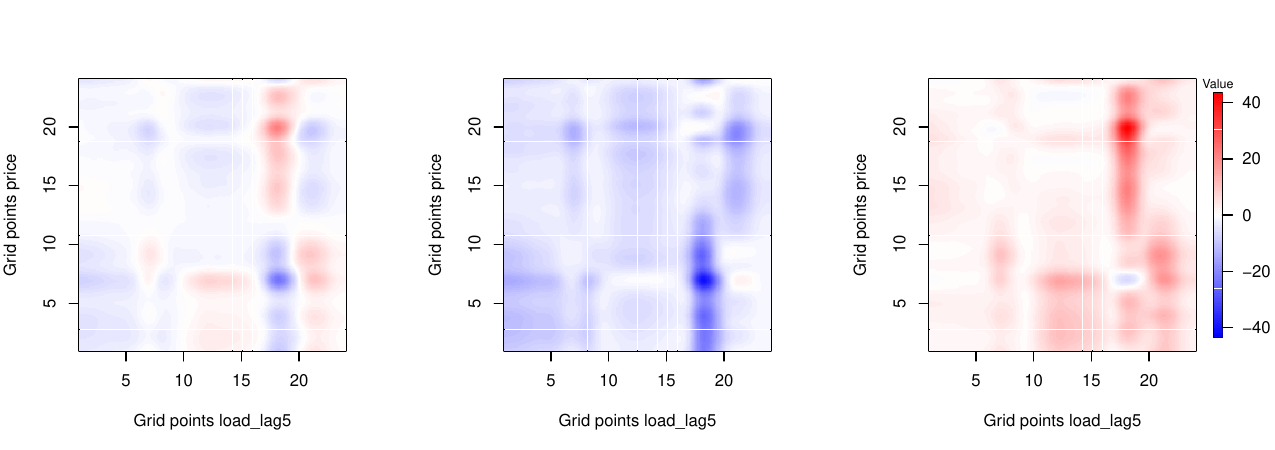}
		\caption{$\widehat{\beta}_{1,5}(r,s)$: impact of $L_{t-5}(s)$ on $P_{t}(r)$.}
	\end{subfigure}
		\caption{First column panels show the estimated functional regression coefficients of model (\ref{eq:elec.FAFR}). Second and third column panels display the lower and upper confidence regions of the $95\%$ confidence interval. The estimated number of factors used in this analysis are shown in Table \ref{ta:epf.K}. The underlying German power market data spans all working days between 2012 and 2017 (cont.).}
\end{figure}

\begin{figure}[H]\ContinuedFloat
	\centering
	\begin{subfigure}[b]{\textwidth}
		\centering
		\includegraphics[width=\textwidth]{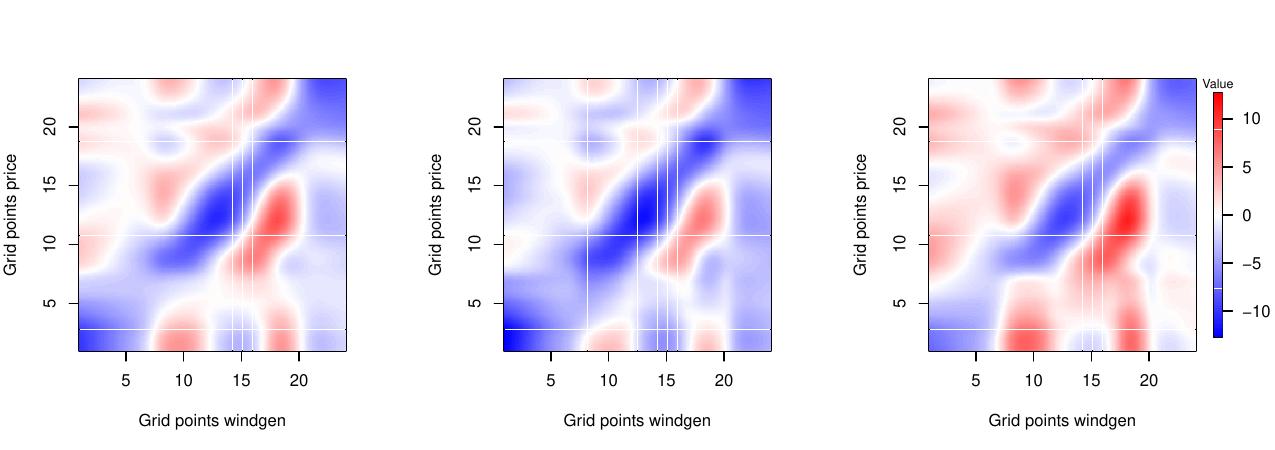}
		\caption{$\widehat{\beta}_{2,1}(r,s)$: impact of $G_{t}(s)$ on $P_{t}(r)$.}
	\end{subfigure}
	
	\begin{subfigure}[b]{\textwidth}
		\centering
		\includegraphics[width=\textwidth]{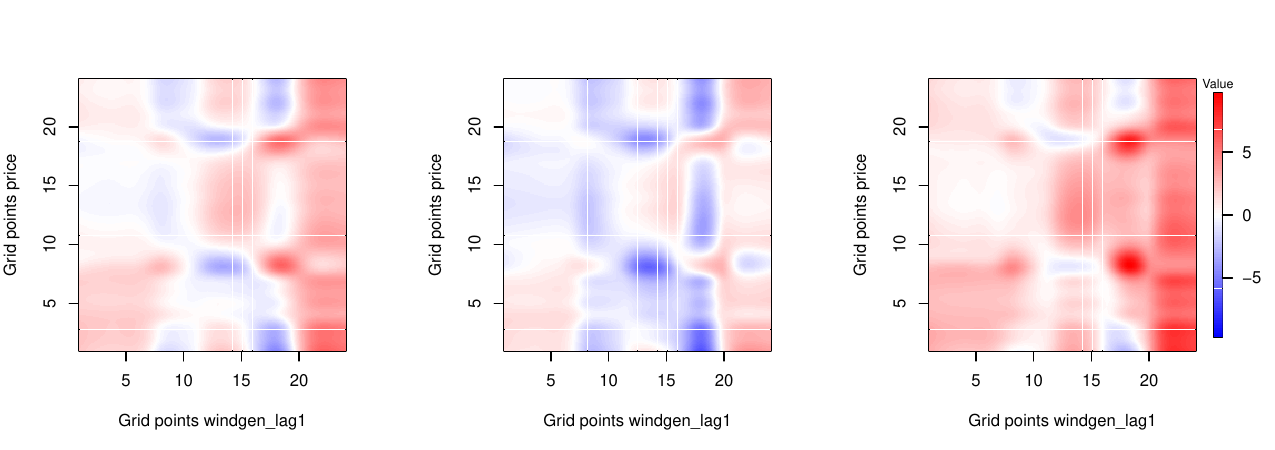}
		\caption{$\widehat{\beta}_{2,2}(r,s)$: impact of $G_{t-1}(s)$ on $P_{t}(r)$.}
	\end{subfigure}
	
	\begin{subfigure}[b]{\textwidth}
		\centering
		\includegraphics[width=\textwidth]{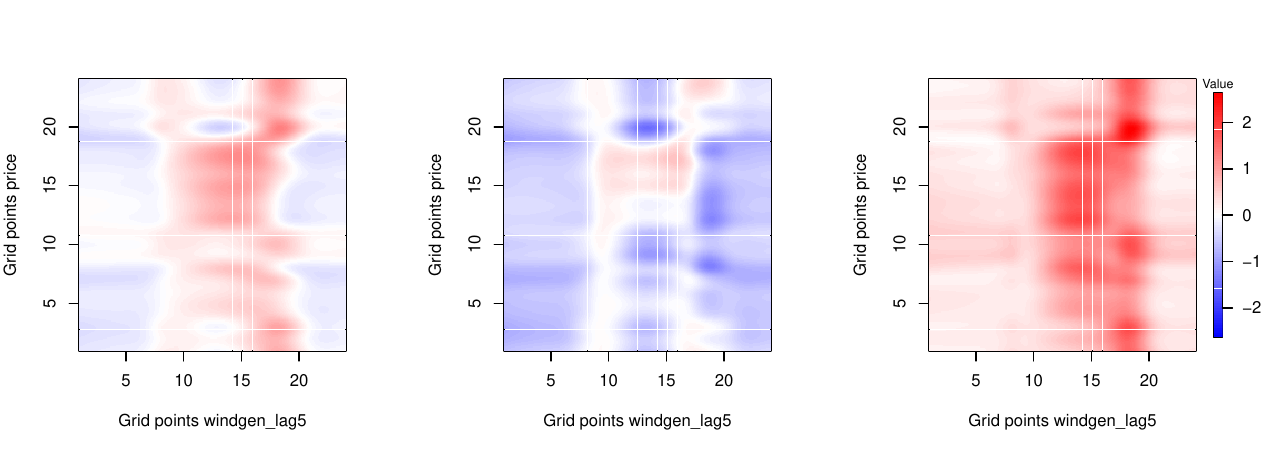}
		\caption{$\widehat{\beta}_{2,5}(r,s)$: impact of $G_{t-5}(s)$ on $P_{t}(r)$.}
	\end{subfigure}
	\caption{First column panels show the estimated functional regression coefficients of model (\ref{eq:elec.FAFR}). Second and third column panels display the lower and upper confidence regions of the $95\%$ confidence interval. The estimated number of factors used in this analysis are shown in Table \ref{ta:epf.K}. The underlying German power market data spans all working days between 2012 and 2017 (cont.).}
\end{figure}

\subsection{Functional pointwise p-values}\label{app:func.p.vals}
\begin{figure}[H]
	\centering
	\begin{subfigure}[b]{\textwidth}
		\centering
		\includegraphics[width=\textwidth]{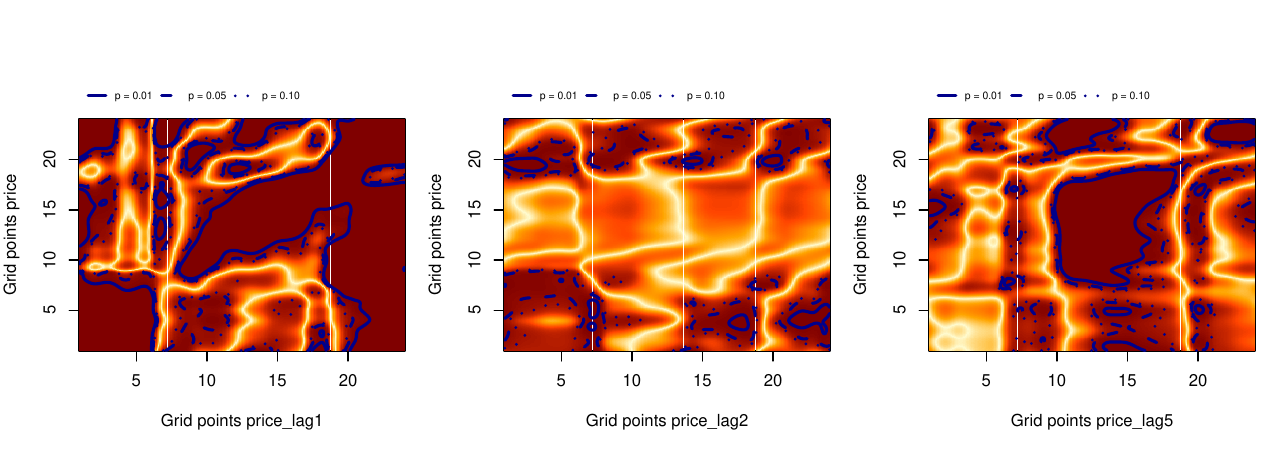}
		\caption{p-value regions for the price regression coefficients $\widehat{\varphi}_1(r,s)$, $\widehat{\varphi}_2(r,s)$  and $\widehat{\varphi}_5(r,s)$.}
	\end{subfigure}
	
	\begin{subfigure}[b]{\textwidth}
		\centering
		\includegraphics[width=\textwidth]{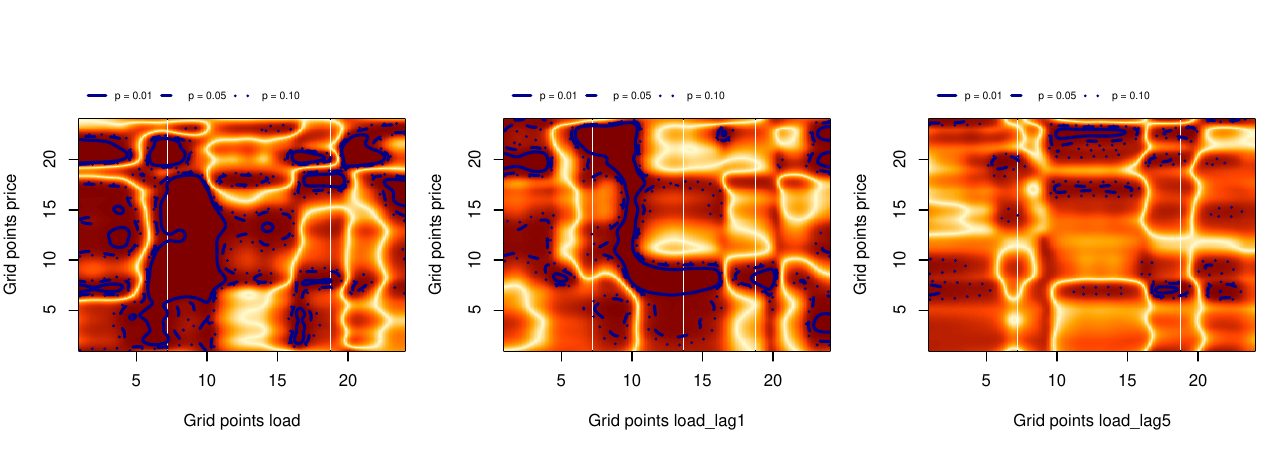}
		\caption{p-value regions for the load regression coefficients $\widehat{\beta}_{1,1}(r,s)$, $\widehat{\beta}_{1,2}(r,s)$ and $\widehat{\beta}_{1,5}(r,s)$.}
	\end{subfigure}
	
	\begin{subfigure}[b]{\textwidth}
		\centering
		\includegraphics[width=\textwidth]{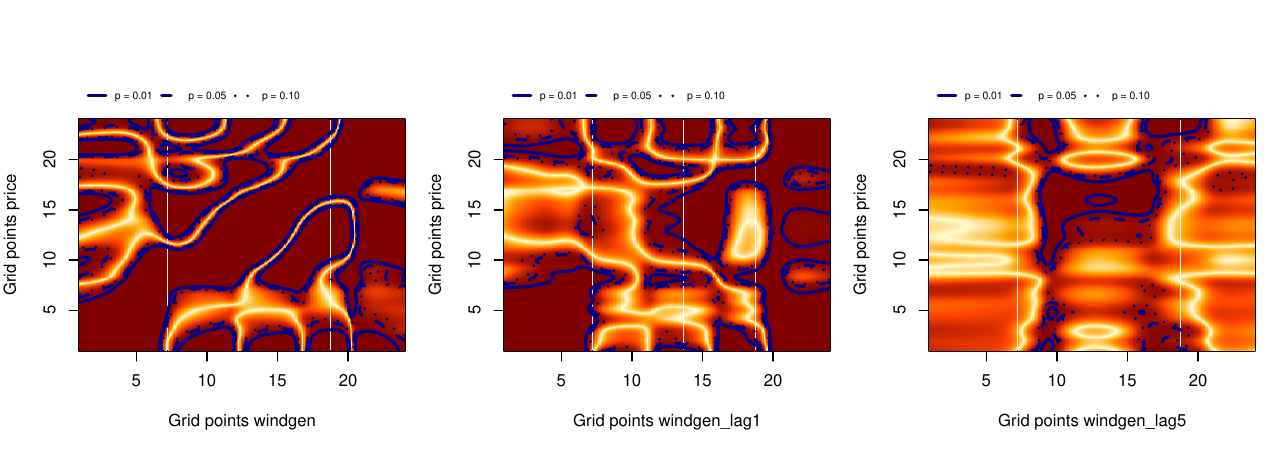}
		\caption{p-value regions for the wind- and solar generation regression coefficients $\widehat{\beta}_{2,1}(r,s)$, $\widehat{\beta}_{2,2}(r,s)$ and $\widehat{\beta}_{2,5}(r,s)$.}
	\end{subfigure}
		\caption{Pointwise p-values from a two-tailed t-test on difference to zero. Dark red color indicates smaller values and the contour lines as defined in the legends indicate significance according to standard alpha levels of 0.01, 0.05 and 0.1.}
\end{figure}